\documentclass[11pt]{article}
\usepackage{amssymb,amsmath,fullpage,color}
\title{\bf\boldmath
Construction of eigenfunctions for the elliptic \\ 
Ruijsenaars difference operators
}
\author{
Edwin Langmann\footnote{Department of Physics, KTH Royal Institute of Technology, SE-106 91 Stockholm, Sweden}, \ 
Masatoshi Noumi\footnote{Department of Mathematics, KTH Royal Institute of Technology, SE-100 44 Stockholm, Sweden 
\hfill \break\ \ (on leave from:  Department of Mathematics, Kobe University,  
Rokko, Kobe 657-8501, Japan)} \
and 
Junichi Shiraishi\footnote{Graduate School of Mathematical Sciences, The University of Tokyo, Komaba, 
Tokyo 153-8914, Japan}
}
\date{}
\newtheorem{thm}{Theorem}[section]

\newtheorem{prop}[thm]{Proposition}

\newtheorem{lem}[thm]{Lemma}
\newtheorem{rem}[thm]{Remark}

\numberwithin{equation}{section}
\newcommand{\bC}{\mathbb{C}}
\newcommand{\bN}{\mathbb{N}}

\newcommand{\bR}{\mathbb{R}}
\newcommand{\bZ}{\mathbb{Z}}
\newcommand{\bT}{\mathbb{T}}

\newcommand{\cA}{\mathcal{A}}
\newcommand{\cB}{\mathcal{B}}
\newcommand{\cC}{\mathcal{C}}
\newcommand{\cD}{\mathcal{D}}
\newcommand{\cE}{\mathcal{E}}

\newcommand{\cK}{\mathcal{K}}

\newcommand{\cN}{\mathcal{N}}

\newcommand{\cP}{\mathcal{P}}

\newcommand{\cR}{\mathcal{R}}
\newcommand{\cS}{\mathcal{S}}
\newcommand{\cU}{\mathcal{U}}
\newcommand{\cV}{\mathcal{V}}
\newcommand{\cW}{\mathcal{W}}
\newcommand{\frS}{\mathfrak{S}}
\newcommand{\ep}{\epsilon}
\newcommand{\vep}{\varepsilon}

\newcommand{\isom}{\stackrel{\sim}{\to}}
\newcommand{\aff}{{\mathrm{af}}}

\newcommand{\pr}[1]{\left\{#1\right\}}
\newcommand{\prm}[2]{\left\{{#1}\big|{#2}\right\}}
\newcommand{\br}[1]{\left\langle #1\right\rangle}

\newcommand{\htaf}{\mathrm{ht}}
\newcommand{\dd}{\mathrm{d}}
\newcommand{\ca}{\mathsf{a}}
\newcommand{\cb}{\mathsf{b}}
\newcommand{\cc}{\mathsf{c}}
\newcommand{\cs}{\mathsf{s}}
\newcommand{\ct}{\mathsf{t}}
\newcommand{\sP}{\mathsf{P}}
\newcommand{\sQ}{\mathsf{Q}}

\newcommand{\qed}{\hfill $\square$\par\noindent}

\newcommand{\fps}[1]{[\hspace{-1pt}[#1]\hspace{-1pt}]}
\newcommand{\cps}[1]{\pr{#1}}
\newenvironment{proof}[1]{\par\smallskip\noindent{{\sl #1}\,:}}
{\hfill $\square$\par\noindent}
\newenvironment{proofa}[1]{\par\smallskip\noindent{{\sl #1}\,:}}
{\par\noindent}
\newcommand{\pick}[1]{\begin{picture}(0,0){#1}\end{picture}}

\begin{document}
\maketitle
\begin{abstract}
We present a perturbative construction of two kinds of eigenfunctions of the commuting family of difference operators defining the elliptic Ruijsenaars system. The first kind corresponds to elliptic deformations of the Macdonald polynomials, and the second kind generalizes asymptotically free eigenfunctions previously constructed 
in the trigonometric case.  We obtain these eigenfunctions as infinite series which, as we show, converge in suitable domains of the variables and parameters.  Our results imply that, for the domain where the elliptic Ruijsenaars operators define a relativistic quantum mechanical system, the elliptic deformations of the Macdonald polynomials provide a family of orthogonal functions with respect to the pertinent scalar product. 
\par\medskip\noindent
{\sl Keywords\/}:\ \ 
Ruijsenaars difference operators, perturbative eigenfunctions, 
elliptic deformation of Macdonald polynomials
\par\noindent
{\sl 2010 Mathematics Subject Classification\/}:\ \ 81Q80; 33E30,  33D67
\end{abstract}
\tableofcontents


\section{Introduction}

We consider the relativistic generalization of $A$-type Calogero-Moser-Sutherland systems discovered by Ruijsenaars \cite{R1987}.  
It is well-known that, in the trigonometric case, the eigenfunctions of the commuting family of difference operators defining the Ruijsenaars system are given by the Macdonald polynomials \cite{M1995}. 
While the Macdonald polynomials provide eigenfunctions that describe identical particles of a quantum mechanical system, there exists another kind called asymptotically free eigenfunctions \cite{S2005, C2009, vMS2010, NS2012}
that are interesting from a mathematical point of view. 
The aim of this paper is to generalize these two kinds of eigenfunctions to the general elliptic Ruijsenaars system \cite{R1987}; 
see \cite{FV1997, FV1998, FV2004, R2009a, R2009b} 
preceding works in this direction. 
For that purpose, we generalize an approach based on kernel functions that was 
used to construct eigenfunctions of the elliptic Calogero-Moser-Sutherland system \cite{L2000, L2013}. 


\subsection{Motivation}\label{ssec:1.1}

We recall that the Macdonald polynomials, $P_\lambda(x_1,\ldots,x_n|q, t)$, can be interpreted as quantum mechanical eigenstates of a Hamiltonian, and 
the quantum numbers 
$\lambda=(\lambda_1,\ldots,\lambda_n)$ with integers $\lambda_i$ 
such that $\lambda_1\ge \lambda_2\ge\cdots\ge\lambda_n$ 
correspond to quasi-momenta; 
this Hamiltonian describes a system of an arbitrary number, $n$, of 
massive particles moving on the circle of circumference $L > 0$, 
and it depends on two parameters $mc > 0$ and 
$g > 0$ \cite{R1987}.\footnote{We mention in passing that 
$m$, $c$ and $g$ have the physical interpretation of particle mass, 
vacuum velocity of light, and coupling constant, respectively \cite{R1987};
note that $u_i$ here corresponds to $q_i$ in \cite{R1987}.} 
Moreover, the particle positions $u_i\in [-L/2, L/2]$ and parameters $mc$, $g$ 
in this quantum mechanical system are related to the variables $x_i$ and Macdonald parameters $q$, $t$ as follows \cite{vD1995, H1997},\footnote{
This relation between the trigonometric Ruijsenaars system and the Macdonald polynomials was already noted in by Koornwinder in 1987 \cite{vD1995}.}
\begin{equation}\label{eq:physdom}
x_i=\exp(2\pi\sqrt{-1} u_i/L)\quad(i=1,\ldots,n),\quad
q=\exp(-1/mc),\quad t=\exp(-g/mc). 
\end{equation}
Thus, from a physical point of view, one is mainly interested in the eigenfunctions 
on the physical domain of variables and parameters given by
\begin{equation}
|x_i|=1\quad (i=1,\ldots,n), \quad 0<q<1, \quad 0<t<1.
\end{equation}
However, mathematically, it is convenient to extend the variables $x_i$ 
and parameter $q$, $t$ to the complex domain \cite{R1997}.  

The asymptotically free eigenfunctions, $f(x_1,\ldots,x_n; s_1,\ldots, s_n|q,t)$, 
generalize the Macdonald polynomials in that they are defined for complex $s_i$ 
and essentially\footnote{
The precise relation is $P_\lambda(x|q,t) = x^\lambda f(x;t^{\rho}q^{\lambda}|q,t)$ 
with $\rho_i = n-i$, using a common shorthand notation introduced in the main text. 
} reduce to the Macdonald polynomials  
$P_\lambda(x_1,\ldots,x_n|q,t)$ in the special case $s_i = t^{n-i}q^{\lambda_i}$  
$(i = 1,\ldots,n)$; however, 
the series defining the asymptotically free eigenfunctions are only convergent in the asymptotically free domain 
$|x_1|\gg|x_2|\gg\cdots\gg|x_n|$ \cite{NS2012}. 
We mention in passing that one remarkable property of the functions 
$f(x_1,\ldots,x_n;s_1,\ldots,s_n|q,t)$ is that they can be represented by simple explicit formulas \cite[Eqs. (1.10)--(1.11)]{NS2012}.\footnote{
Note that $f(x; s|q, t)$ here is denoted as $p_n(x; s|q, t)$ in \cite{NS2012}.}

Our first main result is a construction of elliptic deformations of the Macdonald polynomials, $\cP_\lambda(x_1, \ldots, x_n; p|q, t)$, 
depending on the elliptic deformation parameter $p$ with
$p\to 0$ corresponding to the trigonometric limit, such that
\begin{equation}
\cP_{\lambda}(x_1,\ldots,x_n;p|q,t) = P_\lambda(x_1,\ldots,x_n|q,t)
+\sum_{k=1}^{\infty} p^k\,\cP_{\lambda,k}(x_1,\ldots,x_n|q,t)
\end{equation}
where $(x_1\cdots x_n)^{k-\lambda_n}\cP_{\lambda,k}(x_1,\ldots,x_n|q,t)$ are 
symmetric polynomials for all $k =1,2,\ldots$. 
Moreover, we show that the series defining these functions is absolutely convergent 
in a suitable domain of the variables and parameters, including the physical domain 
in \eqref{eq:physdom} and real nonnegative $p$.  
We also show that these eigenfunctions, restricted to the physical domain, provide 
a family of orthogonal functions with respect to the pertinent scalar product. 
 Our second main result is a construction of asymptotically free eigenfunctions,  
$f(x_1,\ldots,x_n;s_1,\ldots,s_n;p|q,t)$, 
of the elliptic Ruijsenaars operators in the asymptotic domain 
$|x_1|\gg|x_2|\gg\cdots\gg|x_n|\gg|px_1|$ 
such that
\begin{equation}
\begin{split}
&f(x_1,\ldots,x_n;s_1,\ldots,s_n;p|q,t)
\\
&=f(x_1,\ldots,x_n;s_1,\ldots,s_n|q,t)
+\sum_{k=1}^{\infty} p^k f_k(x_1,\ldots,x_n;s_1,\ldots,s_n|q,t), 
\end{split}
\end{equation}
where $(x_n/x_1)^k f_k(x_1,\ldots,x_n;s_1,\ldots,s_n|q,t)$ are convergent 
power series in the variables $x_2/x_1$, $\ldots$, $x_n/x_{n-1}$, with coefficients 
depending rationally on $s_i$, for all $k=1,2,\ldots$.  
In a similar way to the trigonometric case, 
they essentially recover the elliptic deformations $\cP_\lambda(x_1,\ldots,x_n;p|q,t)$ 
of Macdonald polynomials 
as the special case $s_i=t^{n-i}q^{\lambda_i}$ ($i=1,\ldots,n$).  

In the sense of formal power series in the deformation parameter $p$, 
the construction of the eigenfunctions $\cP_\lambda(x_1,\ldots,x_n;p|q,t)$ 
can be achieved by the standard approach of perturbation of eigenvalue problems
\cite{Kato1995, KT2002}.  
It seems technically demanding, however, to directly prove 
their convergence in suitable domains of the variables and parameters. 
To overcome this difficulty, we first establish the convergence 
of asymptotically free eigenfunctions $f(x_1,\ldots,x_n;s_1,\ldots,s_n;p|q,t)$ 
in the asymptotic domain $|x_1|\gg\cdots\gg|x_n|\gg|px_1|$ by 
the method of majorants for their coefficients which are determined by 
nonlinear recurrences.  
Then, by applying an integral transform 
to those eigenfunctions in the asymptotic domain 
$|x_1|\gg\cdots\gg|x_n|\gg|px_1|$, 
we obtain symmetric eigenfunctions around the physical domain $|x_i|=1$ ($i=1,\ldots,n$) which provide the elliptic deformation of Macdonald polynomials; 
this is a new type of integral transform based on a kernel function for the 
elliptic Ruijsenaars system. 

\par\medskip

In Section \ref{sec:Summary}, we give a detailed presentation of our results, 
ending with the plan (Subsection \ref{ssec:1.5}) 
for Sections 3 -- 6 where we present the proofs of our results. 
In our concluding remarks in Section 7, we summarize our results and discuss 
open problems. 


\subsection{Notations}  

We collect below some notations which will be used in various places of this paper. 

\par\smallskip\noindent (0)  
We use the standard notations $\bZ$, $\bR$, $\bC$ for the 
sets of integers, real numbers and complex numbers, respectively. 
We denote by $\bN=\bZ_{\ge 0}$ and $\bZ_{<0}$ 
the sets of nonnegative and negative integers, 
and by $\bC^\ast=\bC\backslash\!\pr{0}$ 
the set of nonzero complex numbers. 
We also use the notation $q^{S}=\prm{q^{k}}{k\in S}$ 
for a subset $S\subseteq \bZ$. 

\par\smallskip\noindent (1)  
For a set $x=(x_1,\ldots,x_n)$ of $n$ variables, 
we denote by $\bC[x]$, $\bC[x^{\pm1}]$, $\bC\cps{x}$ and $\bC\fps{x}$ 
the rings of all polynomials, Laurent polynomials, convergent power series and formal power series in $x$, 
respectively.   The symmetric group $\frS_n$ of degree $n$ acts on these rings by the permutation 
of indices of $x_i$ ($i=1,\ldots,n$). If $\cR$ is one of them, we denote by 
$\cR^{\frS_n}$ 
the subring of all $\frS_n$-invariant (symmetric) elements in $\cR$. 
When we regard the elliptic nome $p$ as an indeterminate, 
we denote by $\cR\fps{p}$ the ring of formal power series in $p$ 
with coefficients in a ring $\cR$. 

\par\smallskip\noindent (2)  
For each $r=0,1,\ldots,n$, we denote by 
\begin{equation}
e_r(x)=e_r(x_1,\ldots,x_n)=\sum_{1\le i_1<\cdots<i_r\le n}x_{i_1}\cdots x_{i_r}
\end{equation}
the $r$th elementary symmetric function of $x=(x_1,\ldots,x_n)$. 
We 
remark that 
the ring of symmetric polynomials 
$\bC[x]^{\frS_n}$ is generated by the elementary symmetric functions 
$e_1(x),\ldots,e_n(x)$, 
and that 
the ring of symmetric Laurent polynomials 
$\bC[x^{\pm1}]^{\frS_n}$ is generated by 
$e_1(x),\ldots,e_n(x)=x_1\cdots x_n$ and $(x_1\cdots x_n)^{-1}$.  
We use the generating function
\begin{equation}
\sum_{r=0}^{n}(-1)^r e_r(x)u^r=(1-x_1u)\cdots(1-x_nu)
\end{equation}
for the elementary symmetric functions throughout this paper. 

\par\smallskip\noindent (3)  
Let $\mu=(\mu_1,\ldots,\mu_n)\in\bZ^n$ be an $n$-vector of integers.
For a nonzero complex number $c\in\bC^\ast$, and an $n$-vector
$a=(a_1,\ldots,a_n)\in(\bC^\ast)^n$, we define 
\begin{equation}
c^{\mu}=(c^{\mu_1},\ldots,c^{\mu_n}),\quad
a^{\mu}=(a_1^{\mu_1},\ldots,a_n^{\mu_n})\in(\bC^{\ast})^n
\end{equation}
respectively.  
Also, for an $n$-vector $a=(a_1,a_2,\ldots,a_n)\in\bC^n$, 
we denote by $a^{\vee}=(a_n,\ldots,a_2,a_1)\in\bC^n$ the {\em reversal} of $a$.  


\section{Summary of results}\label{sec:Summary}


\subsection{Elliptic Ruijsenaars difference operators}\label{ssec:1.2}
Fixing a base $p\in\bC$ with $|p|<1$, throughout this paper 
we use the notations 
\begin{equation}
\theta(z;p)=(z;p)_\infty(p/z;p)_\infty,\quad
(z;p)_\infty=\prod_{i=0}^{\infty}(1-p^iz), 
\end{equation}
of an elliptic 
theta function with elliptic nome $p$ in the multiplicative variable
and the infinite $p$-shifted factorial,
respectively.  
For $n=1,2,\ldots$, 
we denote by $\cD_x(p)=\cD_x(p|q,t)$ the {\em elliptic Ruijsenaars 
$q$-difference operator} of first order of type $A_{n-1}$, in $n$ variables 
$x=(x_1,\ldots,x_n)$ with parameter $t\in\bC^\ast=\bC\backslash\pr{0}$;
$\cD_x(p)$ is defined by
\begin{equation}
\cD_x(p)=\sum_{i=1}^{n}\prod_{\substack{1\le j\le n\\ j\ne i}}
\frac{\theta(tx_i/x_j;p)}{\theta(x_i/x_j;p)}\,T_{q,x_i},
\end{equation}
where $q\in\bC^\ast$, $|q|<1$, and for $i=1,\ldots,n$, 
$T_{q,x_i}$ stands for the $q$-shift operator in $x_i$ 
such that $T_{q,x_i}(x_j)=q^{\delta_{i,j}}x_j$ ($j=1,\ldots,n$). 
We regard $x=(x_1,\ldots,x_n)$ as the canonical coordinates 
of the direct product $(\bC^\ast)^n$ of $n$ copies of $\bC^\ast$, 
and denote by 
\begin{equation}
\bT^n=\prm{x=(x_1,\ldots,x_n)\in(\bC^\ast)^n\,}{\ |x_i|=1\ \ (i=1,\ldots,n)}
\end{equation}
the $n$-dimensional torus in $(\bC^\ast)^{n}$. 

For each $r=0,1,\ldots,n$, we define a 
$q$-difference operator $\cD^{(r)}_x(p)=\cD^{(r)}_x(p|q,t)$ of $r$th order by 
\begin{equation}\label{eq:defcDrx}
\cD^{(r)}_x(p)=t^{\binom{r}{2}}
\sum_{\substack{I\subseteq\pr{1,\ldots,n}\\ |I|=r}}\,
\prod_{\substack{i\in I\\ j\notin I}}
\frac{\theta(tx_i/x_j;p)}{\theta(x_i/x_j;p)}\,
\prod_{i\in I}T_{q,x_i}
\end{equation}
so that $\cD^{(0)}_x(p)=1$, $\cD^{(1)}_x(p)=\cD_x(p)$ and 
$\cD^{(n)}_x(p)=t^{\binom{n}{2}}T_{q,x_1}\cdots T_{q,x_n}$.  
We remark that the $q$-difference operators 
$\cD^{(-r)}_x(p)=\cD^{(-r)}_x(p|q,t)$ of negative orders defined by 
\begin{equation}
\cD^{(-r)}_x(p)=\cD^{(n-r)}_x(p) \cD^{(n)}_x(p)^{-1}
\end{equation}
satisfy $\cD^{(-r)}_x(p|q,t)=\cD^{(r)}_x(p|q^{-1},t^{-1})$.  
It is known by the pioneering work of Ruijsenaars \cite{R1987} that 
all these operators $\cD^{(r)}_x(p)$ ($r=\pm 1,\ldots,\pm n$) 
commute with each other. 
It seems, however, that understanding of the  
eigenfunctions of this {\em elliptic Ruijsenaars system} is rather limited. 
The purpose of this paper is to investigate the joint eigenvalue problem 
\begin{equation}
\cD^{(r)}_x(p) \psi(x;p)=\vep^{(r)}(p)\psi(x;p)\qquad(r=\pm 1,\ldots,\pm n)
\end{equation}
for this family of commuting difference operators, in which 
the joint eigenfunction $\psi(x;p)$ and the eigenvalues $\vep^{(r)}(p)$ 
are expressible as {\em convergent} power series in $p$:  
\begin{equation}
\psi(x;p)=\sum_{k=0}^{\infty}\,p^k\psi_k(x),\quad
\vep^{(r)}(p)=\sum_{k=0}^{\infty}\,p^k\vep^{(r)}_k\quad(r=1,\ldots,n).  
\end{equation}
We propose in particular two kinds of joint eigenfunctions\footnote{
Any joint eigenfunction of the operators $\cD^{(r)}_x(p)$ 
($r=1,\ldots,n$)
is necessarily an eigenfunction of $\cD^{(r)}_x(p)$ for $r=-1,\ldots,-n$ 
as well.  For this reason, we will mainly consider 
the elliptic Ruijsenaars operators 
$\cD^{(r)}_x(p)$ 
($r=1,\ldots,n$) of positive orders 
in the construction of joint eigenfunctions. }:
\par\medskip\noindent
\quad(1) Symmetric joint eigenfunctions around the torus $\bT^n$, and 
\par\smallskip\noindent 
\quad(2) Asymptotically free joint eigenfunctions in the 
domain $|x_1|\gg|x_2|\gg\cdots\gg|px_1/x_n|$. 
\par\medskip\noindent
The joint eigenfunctions of the first kind can be thought of as elliptic deformations 
of the Macdonald polynomials (\cite{M1995, M2003}), 
and those of the second kinds as elliptic deformations 
of the asymptotically free solutions of the trigonometric case in the domain 
$|x_1|\gg|x_2|\gg\cdots\gg |x_n|$ (\cite{S2005}, \cite{C2009}, \cite{vMS2010}, \cite{NS2012}).  


\subsection{Elliptic deformation of Macdonald polynomials}\label{ssec:1.3}

In the limit as $p\to 0$, the operators $D^{(x)}_x=\cD_x^{(r)}(0)$ appear as
\begin{equation}
D_x^{(r)}=t^{\binom{r}{2}}\sum_{\substack{I\subseteq\pr{1,\ldots,n}\\ |I|=r}}
\prod_{\substack{i\in I\\ j\notin I}}\frac{1-tx_i/x_j}{1-x_i/x_j}\prod_{i\in I}T_{q,x_i}\qquad(r=0,1,\ldots,n),
\end{equation}
which we call the {\em Macdonald-Ruijsenaars operators}.  
They define a commuting family of $\bC$-linear operators
$D_x^{(r)}:\ \bC[x]^{\frS_n}\to \bC[x]^{\frS_n}$ 
on the ring of symmetric polynomials in $x=(x_1,\ldots,x_n)$.  
If the parameter $t\in\bC^\ast$ is generic, they 
are simultaneously diagonalized by the {\em Macdonald polynomials} 
$P_\lambda(x)=P_\lambda(x|q,t)$ attached to partitions $\lambda$ 
with number of positive parts $\ell(\lambda)\le n$ (\cite{M1995}).  
In considering the elliptic deformations of Macdonald polynomials, 
however, we need to consider the operators
$D_x^{(r)}:\ \bC[x^{\pm1}]^{\frS_n}\to \bC[x^{\pm1}]^{\frS_n}$ 
acting on the ring of symmetric {\em Laurent} polynomials as well.  

We say that an $n$-vector $\lambda=(\lambda_1,\ldots,\lambda_n)\in\bZ^n$ 
of integers 
is {\em dominant} if $\lambda_1\ge\lambda_2\ge\cdots\ge\lambda_n$.  
This condition is equivalent to saying that 
it is expressed as $\lambda=\mu+(k^n)$, $(k^n)=(k,\ldots,k)$, 
by a partition $\mu$ with $\ell(\mu)\le n$ and an integer $k\in\bZ$. 
We extend the notation $P_\lambda(x)$ of Macdonald polynomials 
to symmetric Laurent polynomials 
attached to the dominant vectors $\lambda$ in $\bZ^n$  by the rule 
\begin{equation}
P_{\lambda}(x)=P_{\mu+(k^n)}(x)=(x_1\cdots x_n)^k P_{\mu}(x)\quad(k\in \bZ). 
\end{equation}
Then, 
the Macdonald-Ruijsenaars operators 
$D^{(r)}_{x}$ acting on $\bC[x^{\pm1}]^{\frS_n}$ 
are simultaneously diagonalized by them. 

In the following, 
to each $n$-vector $\mu=(\mu_1,\ldots,\mu_n)\in\bZ^n$ of integers 
we attach the monomial $x^{\mu}=x_1^{\mu_1}\cdots x^{\mu_n}$ 
of $x$ variables, 
and denote by $|\mu|=\mu_1+\cdots+\mu_{n}$ the degree of $x^{\mu}$. 
For each dominant vector $\lambda$ in $\bZ^n$,  
we denote by 
\begin{equation}
m_{\lambda}(x)=\sum_{\mu\in\frS_n.\lambda}x^{\mu}
\end{equation}
the {\em monomial symmetric function} of type $\lambda$, i.e. 
the sum 
of all monomials $x^\mu$ that are obtained from $x^{\lambda}$ by 
permuting the variables. 
Then we have 
\begin{equation}
\bC[x^{\pm1}]^{\frS_n}=
\bigoplus_{\substack{\lambda\in\bZ^n:\,\mbox{\scriptsize dominant}}}
\,\bC\,m_{\lambda}(x).  
\end{equation}
Note that $m_{\lambda+(k^n)}(x)=(x_1\cdots x_n)^{k}m_{\lambda}(x)$ 
for any $k\in\bZ$.  
We define the {\em dominance order} $\mu\le\nu$ on $\bZ^n$ by 
the condition
\begin{equation}
|\mu|=|\nu|\quad\mbox{and}\quad 
\mu_1+\cdots+\mu_i\le\nu_1+\cdots+\nu_i\quad
(i=1,\ldots, n-1). 
\end{equation}
Also, for each dominant vector $\lambda$ in $\bZ^n$, we denote by 
$\bC[x^{\pm1}]^{\frS_n}_{\le \lambda}$ 
the vector subspace 
of all symmetric Laurent polynomials of the form
\begin{equation}
\varphi(x)=
\sum_{\mu\le \lambda}\varphi_{\mu}m_{\mu}(x), 
\end{equation}
where the sum is over all dominant vectors $\mu$ in $\bZ^n$ such that 
 $\mu\le\lambda$.  
Then the existence theorem of Macdonald (Laurent) polynomials can be 
formulated as follows. 
\par\medskip
Suppose that the parameter $t\in\bC^\ast$ satisfies the condition 
$t^{k}\notin q^{\bZ_{<0}}$ ($k=1,\ldots,n-1$). 
Then, 
for each dominant vector $\lambda$ in $\bZ^n$, 
there exists a unique 
symmetric Laurent polynomial 
$P_\lambda(x)=P_{\lambda}(x|q,t)\in\bC[x^{\pm1}]^{\frS_n}$ 
such that $P_\lambda(x)$ is expressed in the form 
\begin{equation}
P_\lambda(x)=m_\lambda(x)+\sum_{\mu<\lambda}P_{\lambda,\mu}\,m_{\mu}(x) 
\end{equation}
with leading term $m_\lambda(x)$ with respect the dominance order, 
and satisfies the joint eigenfunction equations
\begin{equation}\label{eq:MAEq}
D_x^{(r)}P_\lambda(x)=e_r(t^{\rho}q^{\lambda})P_{\lambda}(x)\qquad(r=1,\ldots,n); 
\end{equation}
the eigenvalues $e_r(t^{\rho}q^{\lambda})$ are elementary symmetric functions of 
$t^{\rho}q^{\lambda}=(t^{n-1}q^{\lambda_1},t^{n-2}q^{\lambda_2},\ldots,q^{\lambda_n})$.
We use hereafter the notation 
$\rho=(n-1,n-2,\ldots,0)$ for the 
staircase partition, instead of $\delta$ in Macdonald's monograph \cite{M1995}, 
reserving the symbol $\delta$ for other uses. 
\par\medskip
Returning to the elliptic Ruijsenaars operators,  
we expand each $q$-difference operator $\cD^{(r)}_x(p)$ as a power series of $p$ in the form
\begin{equation}
\cD^{(r)}_x(p)=\sum_{k=0}^{\infty} p^k\cD^{(r)}_{x,k}.  
\end{equation}

\begin{prop}\label{prop:RAOp}
For each $r=1,\ldots,n$, the coefficients $\cD^{(r)}_{x,k}$ $(k=0,1,2,\ldots)$ 
are $\frS_n$-invariant $q$-difference operators, and all stabilize the 
ring $\bC[x^{\pm1}]^{\frS_n}$ of symmetric Laurent polynomials. 
Furthermore, they have triangularity property 
\begin{equation}
\cD^{(r)}_{x,k} \big(\bC[x^{\pm1}]^{\frS_n}_{\le \mu}\big)\subseteq 
\bC[x^{\pm1}]^{\frS_n}_{\le \mu+k\phi}
\end{equation}
for each dominant vector $\mu$ in $\bZ^n$, where $\phi=(1,0,\ldots,0,-1)$.
\end{prop}

By Proposition \ref{prop:RAOp}, the elliptic Ruijsenaars operators define a commuting family 
of linear operators
\begin{equation}
\cD^{(r)}_x(p): \ \bC[x^{\pm1}]^{\frS_n}\fps{p}\to\bC[x^{\pm1}]^{\frS_n}\fps{p}\quad (r=1,\ldots,n)
\end{equation}
acting on the ring of formal power series in $p$ with coefficients 
in $\bC[x^{\pm1}]^{\frS_n}$. 
We first investigate formal solutions of the joint eigenfunction equations 
\begin{equation}\label{eq:RAEq}
\cD^{(r)}_{x}(p)\psi(x;p)=\vep^{(r)}(p)\psi(x;p)\quad(r=1,\ldots,n)
\end{equation}
for $\psi(x;p)\in \bC[x^{\pm1}]^{\frS_n}\fps{p}$ and $\vep^{(r)}(p)\in\bC\fps{p}$.  

\begin{thm}\label{thm:RAFm}
Suppose that 
$t^{k}\notin q^{\bZ_{<0}}$ $(k=1,\ldots,n-1)$.  
Then, for each dominant vector $\lambda$ in $\bZ^n$, 
the initial value problem $\psi(x;0)=P_\lambda(x)$ 
of the joint eigenfunction equations \eqref{eq:RAEq}
has a formal solution $\cP_\lambda(x;p)\in\bC[x^{\pm1}]^{\frS_n}\fps{p}$ with 
uniquely determined formal eigenvalues $\vep^{(r)}_\lambda(p)\in\bC\fps{p}$ $(r=1,\ldots,n)$.  
The formal solution $\cP_\lambda(x;p)$ is determined uniquely up to multiplication 
by a power series $\gamma(p)\in\bC\fps{p}$ such that $\gamma(0)=1$. 
\end{thm}

Proposition \ref{prop:RAOp} also implies 
that the formal joint eigenfunction $\cP_\lambda(x;p)$ has $p$-expansion 
of the form
\begin{equation}\label{eq:cPpexp}
\cP_\lambda(x;p)=\sum_{k=0}^{\infty} p^k \cP_{\lambda,k}(x),
\quad \cP_{\lambda,k}(x)\in\bC[x^{\pm1}]^{\frS_n}_{\le \lambda+k\phi},
\end{equation}
with $\cP_{\lambda,0}(x)=P_\lambda(x)$.  
If $\lambda$ is a partition with $\ell(\lambda)\le n$, in particular, 
we have $\cP_{\lambda,k}(x)\in (x_1\cdots x_n)^{-k}\bC[x]^{\frS_n}$ 
for $k=0,1,2,\ldots$.  

We say that 
the formal solution $\cP_{\lambda}(x;p)$ is {\em normalized} 
if the coefficient of $m_{\lambda}(x)$ in $\cP_{\lambda,k}(x)$ is 0 for $k>0$, i.e. 
\begin{equation}\label{eq:cPpexp2}
\cP_{\lambda,k}(x)=\sum_{\mu\le \lambda+k\phi} \cP_{\lambda,k;\mu}\,m_{\mu}(x),
\quad \cP_{\lambda,k;\lambda}=\delta_{k,0}\quad(k=0,1,2,\ldots). 
\end{equation}
This condition is equivalent to saying that $x^{-\lambda}\cP_{\lambda}(x;p)$ has constant term 1
with respect to the $x$ variables.  
Hence, there exists a unique normalized formal solution $\cP_{\lambda}(x;p)$ 
to the initial value problem of Theorem \ref{thm:RAFm}, 
and the other solutions are expressed as $\gamma(p)\cP_{\lambda}(x;p)$ 
with some $\gamma(p)\in\bC\fps{p}$ with $\gamma(0)=1$.  
Note also that any nonzero formal solution $\varphi(x;p)$ of the 
joint eigenfunction equations \eqref{eq:RAEq} is expressed as  
$\varphi(x;p)=\gamma(p)\cP_\lambda(x;p)$ for some 
dominant vector $\lambda$ in $\bZ^n$ and $\gamma(p)\in\bC\fps{p}$.  

We remark that, for each $k=0,1,2,\ldots$, 
the coefficient of the leading term 
$m_{\lambda+k\phi}(x)$ in $\cP_{\lambda,k}(x)$ of \eqref{eq:cPpexp2} is explicitly computed as 
\begin{equation}
\cP_{\lambda,k;\lambda+k\phi}=
\frac{(t;q)_k(t^{n}q^{\lambda_1-\lambda_n};q)_k}{(q;q)_k(t^{n-1}q^{\lambda_1-\lambda_n+1};q)_k}(q/t)^k
\quad(k=0,1,2,\ldots),
\end{equation}
where $(a;q)_k=(1-a)(1-qa)\cdots(1-q^{k-1}a)$. 
It should be noted that the ``groundstate'' $\cP_{0}(x;p)$ determined by the initial condition 
$\cP_{0}(x;0)=1$ is already a {\em nontrivial} power series in $x$.\footnote{
Here, the subcript $0$ stands for the empty partition. 
}
We also remark that the normalized formal solutions $\cP_\lambda(x;p)$ 
satisfy the property 
\begin{equation}
\cP_{\lambda+(k^n)}(x;p)=(x_1\cdots x_n)^k\cP_{\lambda}(x;p)\quad(k\in\bZ),
\end{equation}
similarly to the case of Macdonald polynomials.  In particular, 
$\cP_{\lambda}(x;p)$ for general dominant vectors $\lambda$ in $\bZ^n$ 
are determined 
from those attached to partitions $\lambda$ with $\ell(\lambda)\le n$.  

\par\medskip
As to the convergence of the formal joint eigenfunctions, we propose the following theorem.  
\begin{thm}\label{thm:RACv}
Suppose that the parameter $t\in\bC^\ast$ satisfies the condition 
$t^k\notin q^{\bZ}$ $(k=1,\ldots,n-1)$ 
and  let $\lambda$ be a dominant vector in $\bZ^n$. 
Then, there exists a constant $\tau\in(0,1]$, not depending on $p$, 
such that
the normalized formal joint eigenfunction $\cP_\lambda(x;p)$ is absolutely convergent 
in the domain
\begin{equation}
|p|/\tau<|x_j/x_i|<\tau/|p|\qquad(1\le i<j\le n), 
\end{equation}
for $|p|<\tau$\,$;$  
the formal eigenvalues $\vep^{(r)}_{\lambda}(p)$ $(r=1,\ldots,n)$ are also absolutely 
convergent for $|p|<\tau$.  
\end{thm}
Note that, if $t^k\ne 1$ ($k=1,\ldots,n-1$) and $|q|<\min\pr{|t|^{n-1},|t|^{-n+1}}$, 
then the condition $t^k\notin q^{\bZ}$ ($k=1,\ldots,n-1$) is automatically satisfied. 

\par\medskip
We remark that the symmetric joint eigenfunctions 
$\cP_{\lambda}(x;p)$ 
of Theorem \ref{thm:RACv} are holomorphic around the torus $\bT^n$; 
they can be thought of as elliptic deformation of the Macdonald polynomials.  
In fact, they satisfy 
the joint eigenfunction equations
\begin{equation}\label{eq:cPEFEq}
\cD^{(r)}_x(p)\cP_\lambda(x;p)=\vep^{(r)}_\lambda(p)\cP_\lambda(x;p)\quad(r=1,\ldots,n)
\end{equation}
for the elliptic Ruijsenaars operators, with initial condition 
$\cP_\lambda(x;0)=P_\lambda(x)$, $\vep^{(r)}_\lambda(0)=e_r(t^{\rho}q^{\lambda})$, 
in the domain 
$|p|/|q|\tau <|x_j/x_i|<|q|\tau/|p|$ ($1\le i<j\le n$) for $|p|<|q|\tau$.  
We expect that $\cP_{\lambda}(x;p)$ could be continued to 
global meromorphic functions on $(\bC^\ast)^n$, 
by means of the $q$-difference equations \eqref{eq:cPEFEq},
but we do not have a proof of this conjecture yet.  
\par\medskip
With this existence theorem established, we can prove the orthogonality relations for 
the normalized joint eigenfunctions $\cP_{\lambda}(x;p)$. 
We denote by $\Gamma(z;p,q)$ the {\em Ruijsenaars elliptic gamma function}
with bases $p,q$: 
\begin{equation}
\Gamma(z;p,q)=\frac{(pq/z;p,q)_\infty}{(z;p,q)_\infty}, \quad
(z;p,q)_\infty=\prod_{i,j=0}^{\infty}(1-p^iq^jz).  
\end{equation}
We define a meromorphic function $w^{\mathrm{sym}}(x;p)$ on 
$(\bC^\ast)^n$ by 
\begin{equation}\label{eq:wtfnruij}
w^{\mathrm{sym}}(x;p)=\prod_{1\le i<j\le n}\frac{\Gamma(tx_i/x_j;p,q)\Gamma(tx_j/x_i;p,q)}{\Gamma(x_i/x_j;p,q)\Gamma(x_j/x_i;p,q)}.  
\end{equation} 
We also denote by $d\omega_n(x)$ the normalized invariant measure on $\bT^n$: 
\begin{equation}
d\omega_n(x)=\frac{1}{(2\pi\sqrt{-1})^n}\frac{dx_1\cdots dx_n}{x_1\cdots x_n},
\qquad \int_{\bT^n}d\omega_n(x)=1.  
\end{equation}

\begin{thm}\label{thm:RAOr}
Suppose that $|t|<1$ and $t^k\notin q^{\bZ_{>0}}$ $(k=1,\ldots,n-1)$.
Then, the normalized joint eigenfunctions $\cP_{\lambda}(x;p)$ 
of the elliptic Ruijsenaars operators, 
attached to the dominant vectors $\lambda$ in $\bZ^n$, 
are orthogonal with respect to the weight function 
$w^{\mathrm{sym}}(x;p)$ of \eqref{eq:wtfnruij}  in the sense that 
\begin{equation}\label{eq:cPorth}
\int_{\bT^n}\cP_{\lambda}(x^{-1};p)\cP_{\mu}(x;p)w^{\mathrm{sym}}(x;p)d\omega_n(x)=\delta_{\lambda,\mu}\,\cN_{\lambda}(p)
\end{equation}
with some constants $\cN_\lambda(p)$ depending on $p$.  
\end{thm}
(We expect that the assumption $t^{k}\in q^{\bZ_{>0}}$ ($k=1,\ldots,n-1$) could be removed.) 
\par\medskip
We remark that the integral \eqref{eq:cPorth} for $p=0$ is 
Macdonald's scalar product $\br{\ ,\ }'$ in \cite{M1995} 
up to a factor $n!$.  
It would be an intriguing problem to evaluate the norm $\cN_\lambda(p)$ in  \eqref{eq:cPorth}, although it seems challenging at present. 
Notice here that the weight function 
$w^{\mathrm{sym}}(x;p)=w^{\mathrm{sym}}(x;p,q,t)$ is invariant by 
exchanging $p$ and $q$. 
This suggests that the orthogonal functions $\cP_{\lambda}(x;p)=\cP_{\lambda}(x;p|q,t)$ should also be symmetric with respect to $(p,q)$.  


\subsection{Asymptotically free eigenfunctions in the domain 
$|x_1|\gg\cdots\gg|x_n|\gg|px_1|$}\label{ssec:1.4}

Denoting by 
\begin{equation}
\Delta(x)=\prod_{1\le i<j\le n}(x_i-x_j)=x^{\rho}\prod_{1\le i<j\le n}(1-x_j/x_i)
\end{equation}
the difference product of $x$ variables, 
we write the Macdonald-Ruijsenaars operators in the form 
\begin{equation}
D^{(r)}_x=\sum_{I\subseteq\pr{1,\ldots,n};\ |I|=r} A_I(x)T_{q,x}^{I},
\quad A_I(x)=\frac{T_{t,x}^{I}(\Delta(x))}{\Delta(x)}, 
\end{equation}
where $T_{q,x}^{I}=\prod_{i\in I}T_{q,x_i}$.  
Then, from the expression 
\begin{equation}\label{eq:MRAB}
A_I(x)=
(t^{\rho})^I
B_I(x),\quad 
B_I(x)=
\prod_{\substack{1\le i<j\le n\\ i\in I,\,j\notin I}}\frac{1-x_j/tx_i}{1-x_j/x_i}
\prod_{\substack{1\le i<j\le n\\ i\notin I,\,j\in I}}\frac{1-tx_j/x_i}{1-x_j/x_i}
\end{equation}
with $(t^{\rho})^I=\prod_{i\in I}t^{n-i}$, 
it follows that the coefficients 
$A_I(x)$ are holomorphic in the domain $|x_2/x_1|<1, \ldots, |x_n/x_{n-1}|<1$, 
and they are represented 
there by convergent power series of the variables $x_2/x_1,\ldots,x_n/x_{n-1}$.
Hence  
we see that, for each complex vector $\lambda=(\lambda_1,\ldots,\lambda_n)\in\bC^n$, 
the Macdonald-Ruijsenaars operators
define a commuting family of linear operators 
\begin{equation}
D^{(r)}_x:\ x^{\lambda} \bC\pr{x_2/x_1,\ldots,x_n/x_{n-1}}\to x^{\lambda} \bC\pr{x_2/x_1,\ldots,x_n/x_{n-1}}
\quad(r=1,\ldots,n)
\end{equation}
on the space of convergent power series of the form
\begin{equation}
\psi(x)=x_1^{\lambda_1}\cdots x_n^{\lambda_n}\sum_{k_1,\ldots,k_{n-1}\in\bZ_{\ge 0}} 
c_{k_1,\ldots,k_{n-1}} (x_2/x_1)^{k_1}\cdots (x_n/x_{n-1})^{k_{n-1}}.
\end{equation}
Regarding the joint eigenfunctions in this setting, 
it is known that, 
for each complex vector 
$\lambda\in\bC^n$ such that $t^{i-j}q^{\lambda_j-\lambda_i}\notin 
q^{\bZ}$ $(1\le i<j\le n)$, 
the joint eigenvalue problem
\begin{equation}
D^{(r)}_x \psi_\lambda(x)=e_r(t^{\rho}q^{\lambda})\psi_\lambda(x)\quad(r=1,\ldots,n)
\end{equation}
has a unique solution $\psi_\lambda(x)\in x^{\lambda}\bC\pr{x_2/x_1,\ldots,x_n/x_{n-1}}$ 
with leading term $x^{\lambda}$.  
Such a solution is 
called {\em asymptotically free} in the domain $|x_1|\gg|x_2|\gg\cdots\gg|x_n|$ 
(see Shiraishi \cite{S2005}, Cherednik \cite{C2009}, van Meer-Stockman \cite{vMS2010}, Noumi-Shiraishi \cite{NS2012}).

It is also known that there exists a meromorphic function  $f(x;s)$ 
in the {\em position variables} $x=(x_1,\ldots,x_n)$ and the {\em spectral variables} $s=(s_1,\ldots,s_n)$ 
such that $\psi_\lambda(x)=x^{\lambda}f(x;t^{\rho}q^{\lambda})$ for generic 
$\lambda\in\bC^n$ 
under the identification $s_i=t^{n-i}q^{\lambda_i}$ ($i=1,\ldots,n$).  
We remark 
that the conjugation of the Macdonald-Ruijsenaars operators $D^{(r)}_x$ by $x^\lambda$ 
are expressed as
\begin{equation}
x^{-\lambda}D_x^{(r)}x^{\lambda}=\sum_{I\subseteq\pr{1,\ldots,n};\ |I|=r}
(t^\rho q^\lambda)^I
B_I(x)T_{q,x}^{I},
\end{equation}
where $(t^{\rho}q^{\lambda})^I=\prod_{i\in I}t^{n-i}q^{\lambda_i}$. 
In view of this expression, replacing $(t^{\rho}q^{\lambda})^I$ by 
$s^I=\prod_{i\in I}s_i$, 
we introduce the {\em modified Macdonald-Ruijsenaars operators} 
$E^{(r)}_{x,s}$ with parameters $s=(s_1,\ldots,s_n)$ by
\begin{equation}\label{eq:mMROp}
E^{(r)}_{x,s}=\sum_{I\subseteq\pr{1,\ldots,n};\ |I|=r} s^{I}B_I(x)T_{q,x}^{I}\qquad(r=0,1,\ldots,n),
\end{equation}
so that $E^{(r)}_{x,\,t^{\rho}q^{\lambda}}=x^{-\lambda}D^{(r)}_{x}x^{\lambda}$.  
Then, 
as a function of the $x$ variables, 
$f(x;s)$ mentioned above is characterized as 
the unique solution in $\bC\pr{x_2/x_1,\ldots,x_n/x_{n-1}}$  with leading term 1
of the joint eigenfunction equations 
\begin{equation}
E^{(r)}_{x,s}f(x;s)=e_r(s)f(x;s)\qquad(r=1,\ldots,n). 
\end{equation}
We remark that the unique existence of $f(x;s)$ is guaranteed 
under the condition $s_j/s_i\notin q^{\bZ}$ ($1\le i<j\le n$) \cite{NS2012}.  
We call this $f(x;s)=f(x;s|q,t)$ the {\em Macdonald function} in the 
asymptotic domain
$|x_1|\gg|x_2|\gg\cdots\gg|x_n|$.  
It is known that $f(x;s|q,t)$ can be continued to a meromorphic function on 
$(\bC^\ast)^n\times(\bC^\ast)^n$ 
with pole divisors explicitly described,  
and it has various remarkable properties comparable to those of Macdonald polynomials, such as 
combinatorial formula, bispectral self-duality and transformation under the 
reflection $t\leftrightarrow q/t$ of the parameter $t$ 
(see \cite{NS2012}; $f(x;s|q,t)$ is denoted there by $p_n(x;s|q,t)$).  

\par\medskip
In the elliptic case, the elliptic Ruijsenaars operators are expressed as
\begin{equation}
\cD^{(r)}_x(p)=\sum_{I\subseteq\pr{1,\ldots,n};\ |I|=r} \cA_I(x;p)T_{q,x}^{I},
\quad \cA_I(x;p)=\frac{T_{t,x}^{I}(\Delta(x;p))}{\Delta(x;p)}, 
\end{equation}
in terms of the elliptic deformation $\Delta(x;p)$ of $\Delta(x)$ defined by 
\begin{equation}
\Delta(x;p)=\prod_{1\le i<j\le n}x_i\theta(x_j/x_i;p)
=x^{\rho}\prod_{1\le i<j\le n}\theta(x_j/x_i;p).  
\end{equation}
From the expression
\begin{equation}
\cA_I(x;p)=(t^{\rho})^I\cB_I(x;p),\quad 
\cB_I(x;p)=
\prod_{\substack{1\le i<j\le n\\ i\in I,\,j\notin I}}\frac{\theta(x_j/tx_i;p)}{\theta(x_j/x_i;p)}
\prod_{\substack{1\le i<j\le n\\ i\notin I,\,j\in I}}\frac{\theta(tx_j/x_i;p)}{\theta(x_j/x_i;p)}
\end{equation}
we have

\begin{prop}\label{prop:RBOp}
The coefficients $\cA_I(x;p)$ of the elliptic Ruijsenaars operators are holomorphic in the domain 
$|px_1/x_n|<1, |x_2/x_1|<1,\ldots,|x_n/x_{n-1}|<1$
and, in particular, they can be represented 
 as convergent power series in the variables $z_0=px_1/x_n, z_1=x_2/x_1, \ldots, z_{n-1}=x_n/x_{n-1}$.  
\end{prop}

Regarding the coefficients $\cA_I(x;p)$ as convergent power series in 
\begin{equation}
\bC\pr{px_1/x_n,x_2/x_1,\ldots,x_n/x_{n-1}}, 
\end{equation}
we consider the commuting family of elliptic Ruijsenaars operators
\begin{equation}
\cD^{(r)}_x(p):\ x^\lambda\bC\pr{px_1/x_n,x_2/x_1,\ldots,x_n/x_{n-1}}\to x^{\lambda}\bC\pr{px_1/x_n,x_2/x_1,\ldots,x_n/x_{n-1}}
\end{equation}
($r=1,\ldots,n$) for each complex vector $\lambda\in\bC^n$. 
As we will see below, if $\lambda\in\bC^n$ is generic 
in the sense that $t^{i-j}q^{\lambda_j-\lambda_i}\notin q^{\bZ}$ 
($1\le i<j\le n$), 
these linear operators have 
a joint eigenfunction $\psi_\lambda(x;p)$ in $x^\lambda\bC\pr{px_1/x_n,x_2/x_1,\ldots,x_n/x_{n-1}}$ with leading term $x^\lambda$, 
determined uniquely up to multiplication by a power series $\gamma(p)\in\bC\pr{p}$ with $\gamma(0)=1$.  
We call such an eigenfunction 
$\psi_{\lambda}(x;p)$ {\em asymptotically free}  in the domain 
$|x_1|\gg|x_2|\gg\cdots\gg|x_n|\gg|px_1|$.  

Similarly to the trigonometric case explained above, we define the {\em modified elliptic Ruijsenaars operators} 
$\cE_{x,s}^{(r)}(p)$ ($r=0,1,\ldots,n$) by
\begin{equation}
\cE_{x,s}^{(r)}(p)=\sum_{I\subseteq\pr{1,\ldots,n};\ |I|=r}\,
s^{I}\cB_I(x;p) T_{q,x}^{I}  
\end{equation}
so that $\cD^{(r)}_{x}(p)=\cE^{(r)}_{x,t^\rho}(p)$ and  
$x^{-\lambda}\cD^{(r)}_{x}(p)x^{\lambda}=\cE^{(r)}_{x,t^\rho q^{\lambda}}(p)$.  
Regarding $s=(s_1,\ldots,s_n)\in(\bC^\ast)^n$ as parameters, 
we investigate the joint eigenvalue problem of the modified elliptic Ruijsenaars operators 
\begin{equation}\label{eq:RBEq}
\cE^{(r)}_{x,s}(p)f(x;s;p)=\vep^{(r)}(s;p)f(x;s;p)\quad (r=1,\ldots,n)
\end{equation}
for $f(x;s;p)\in \bC\pr{px_1/x_n,x_2/x_1,\ldots,x_n/x_{n-1}}$ and $\vep^{(r)}(s;p)\in\bC\pr{p}$.  
We first establish an existence theorem for 
formal solutions such that 
$f(x;s;p)\in\bC\fps{px_1/x_n,x_2/x_1,\ldots,x_n/x_{n-1}}$ and $\vep^{(r)}(s;p)\in\bC\fps{p}$, 
and then propose a convergence theorem for them.  
We remark that the ring 
\begin{equation}
\bC\fps{px_1/x_n,x_2/x_1,\ldots,x_n/x_{n-1}}
\end{equation}
of formal power series in $(px_1/x_n,x_2/x_1,\ldots,x_n/x_{n-1})$ can be defined alternatively 
as the ring of all power series in $p$ of the form
\begin{equation}
f(x;p)=\sum_{k=0}^{\infty} p^k f_k(x),\quad f_k(x)\in (x_1/x_n)^k\bC\fps{x_2/x_1,\ldots,x_n/x_{n-1}}. 
\end{equation}

\begin{thm}\label{thm:RBFm}
Suppose that the parameters $s=(s_1,\ldots,s_n)\in(\bC^\ast)^n$ satisfy the condition $s_j/s_i\notin q^{\bZ^\ast}$ 
$(1\le i<j\le n)$, where $\bZ^\ast=\bZ\backslash\!\pr{0}$.  
Then the joint eigenvalue problem \eqref{eq:RBEq} for the modified elliptic Ruijsenaars operators 
$\cE^{(r)}_{x,s}(p)$ $(r=1,\ldots,n)$
has a formal solution $f(x;s;p)\in\bC\fps{px_1/x_n,x_2/x_1,\ldots,x_n/x_{n-1}}$ 
with uniquely determined formal eigenvalues $\vep^{(r)}(s;p)\in\bC\fps{p}$ 
$(r=1,\ldots,n)$. 
The solution $f(x;s;p)$ is determined uniquely up to multiplication by a power series $\gamma(s;p)\in\bC\fps{p}$ 
such that $\gamma(s;0)=1$.  
\end{thm}

We say that the solution $f(x;s;p)$ is {\em normalized} if it has constant term 1 
with respect to 
$x=(x_1,\ldots,x_n)$.  
In terms of the expansion
\begin{equation}
f(x;s;p)=\sum_{k_0,k_1,\ldots,k_{n-1}}f_{k_0,k_1,\ldots,k_{n-1}}(s) (px_1/x_n)^{k_0}
(x_2/x_1)^{k_1}\cdots (x_n/x_{n-1})^{k_{n-1}}
\end{equation}
as a formal power series in $(px_1/x_n,x_2/x_1,\ldots,x_n/x_{n-1})$, 
the constant term of $f(x;s;p)$ with respect to $x=(x_1,\ldots,x_n)$ 
is given by
\begin{equation}
\sum_{k=0}^{\infty}p^k f_{k,k,\ldots,k}(s)\in\bC\fps{p}. 
\end{equation}
Hence, the normalization condition for $f(x;s;p)$ 
mentioned above is equivalent to 
\begin{equation}
f_{k,k,\ldots,k}(s)=\delta_{k,0}\qquad(k=0,1,2,\ldots). 
\end{equation}
We remark that 
the joint eigenvalue problem \eqref{eq:RBEq} 
has a unique normalized solution $f(x;s;p)$ in 
$\bC\fps{px_1/x_n,x_2/x_1,\ldots,x_n/x_{n-1}}$,
and the other solutions are expressed as $\gamma(s;p)f(x;s;p)$ for some $\gamma(s;p)\in\bC\fps{p}$.  

If we specialize the $s$ variables by $s=t^{\rho}q^{\lambda}$ with a complex vector 
$\lambda\in\bC^n$, 
we have $s_j/s_i=t^{i-j}q^{\lambda_j-\lambda_i}$.  
This means that, if $t^{i-j}q^{\lambda_j-\lambda_i}\notin q^{\bZ^\ast}$ ($1\le i<j\le n$), 
$\psi_{\lambda}(x;p)=x^{\lambda}f(x;t^\rho q^\lambda;p)$ gives the formal 
asymptotically free solution of the joint eigenvalue problem 
of the elliptic Ruijsenaars operators, with formal eigenvalues $\vep^{(r)}(t^{\rho}q^{\lambda};p)$,  
in the domain $|x_1|\gg\dots\gg|x_n|\gg|px_1|$.  
Also, if $\lambda$ is a dominant vector in $\bZ^n$, and if 
$t^{k}\notin q^{\bZ}$ ($k=1,\ldots,n-1$), 
the normalized formal solution of Theorem \ref{thm:RAFm} 
is recovered as $\cP_{\lambda}(x;p)=x^{\lambda}f(x;t^{\rho}q^{\lambda};p)$ 
consistently with the normalization condition.  

\par\medskip
As to the convergence of the formal solution $f(x;s;p)$ we have
\begin{thm}\label{thm:RBCv}
Suppose that $s_j/s_i\notin q^{\bZ}$ $(1\le i<j\le n)$. 
Then, there exist a constant $\sigma\in (0,1)$, not depending on $p$, such that 
the normalized formal solution $f(x;s;p)$ of the joint eigenvalue problem \eqref{eq:RBEq} is absolutely 
convergent in the domain
\begin{equation}
|px_1/x_n|<\sigma,\ |x_2/x_1|<\sigma,\ \ldots, |x_n/x_{n-1}|<\sigma
\end{equation}
for $|p|<\sigma^n$.  
\end{thm}
Note that the joint eigenfunction equations \eqref{eq:RBEq} hold in the domain 
\begin{equation}
|px_1/x_n|<|q|\sigma,\ |x_2/x_1|<|q|\sigma,\ \ldots, |x_n/x_{n-1}|<|q|\sigma
\end{equation}
 for $|p|<|q|^n\sigma$.  
We also remark that $f(x;s;p)=f(x;s;p|q,t)$ depend holomorphically on 
$s=(s_1,\ldots,s_n)\in(\bC^\ast)^n$, and $q\in\bC^\ast$ with $|q|<1$ 
and $t\in\bC^\ast$, provided that the condition 
$1-q^l s_j/s_i\ne 0$ $(l\in\bZ)$ is satisfied for any distinct $i,j\in\pr{1,\ldots,n}$. 
More specifically, 
$s_i\ne s_j$ for any distinct $i,j\in\pr{1,\ldots,n}$ and that 
$\tau_0\le |s_j/s_i|\le \tau_0^{-1}$ $(1\le i<j\le n)$
for some $\tau_0\in(0,1]$, then $f(x;s;p|q,t)$ is 
a holomorphic function of $(q,t)$ in the domain 
$\prm{q\in\bC}{\ |q|<\tau_0}\times \bC^\ast$. 

We refer to this normalized joint eigenfunction $f(x;s;p)=f(x;s;p|q,t)$ 
as the {\em Ruijsenaars function} in the domain 
$|x_1|\gg|x_2|\gg\cdots\gg|x_n|\gg|px_1|$.  
In the context of \cite{S2019} and \cite{LNS2020a}, 
$f(x;s;p|q,t)$ is the {\em stationary} Ruijsenaars function, which is conjectured 
to arise through a limiting procedure, as $\kappa\to1$, from the {\em non-stationary}  
Ruijsenaars function involving an extra parameter $\kappa$.  
We expect that the stationary and non-stationary Ruijsenaars functions should inherit various symmetry properties 
of the Macdonald function in the trigonometric case.  

\par\medskip
We give below two results regarding symmetries of the Ruijsenaars function.   

\begin{thm}\label{thm:RBRo}
The normalized Ruijsenaars function $f(x;s;p)=f(x;s;p|q,t)$ is invariant under the simultaneous 
rotation of indices of the $x$ and $s$ variables in the following sense\,$:$
\begin{equation}
f(x_2,\ldots,x_n,px_1;s_2,\ldots,s_n,s_1;p)=f(x_1,\ldots,x_n;s_1,\ldots,s_n;p).
\end{equation}
The eigenvalues $\vep^{(r)}(s;p)$ of the modified elliptic Ruijsenaars operators are also invariant 
under the rotation of indices of the $s$ variables\,$:$  
\begin{equation}
\vep^{(r)}(s_2,\ldots,s_n,s_1;p)=\vep^{(r)}(s_1,\ldots,s_n;p)\quad(r=1,\ldots,n).  
\end{equation}
\end{thm}

\begin{thm}\label{thm:RBRe}
The normalized Ruijsenaars function $f(x;s;p|q,t)$ 
satisfies the following transformation formula 
under the reflection $t\leftrightarrow q/t$ of the parameter $t$\,$:$ 
\begin{equation}
f(x;s;p|q,q/t)=\gamma(s;p|q,t) \prod_{1\le i<j\le n}
\frac{\Gamma(tx_j/x_i;p,q)}{\Gamma(qx_j/tx_i;p,q)}\cdot f(x;s;p|q,t), 
\end{equation}
for some constant $\gamma(s;p|q,t)$ depending on $p$ such that 
$\gamma(s;p|q,t)\gamma(s;p|q,q/t)=1$.
Also, the eigenvalues $\vep^{(r)}(s;p)=\vep^{(r)}(s;p|q,t)$ 
are invariant under the reflection $t\leftrightarrow q/t$ of the parameter $t$\,$:$
\begin{equation}
\vep^{(r)}(s;p|q,q/t)=\vep^{(r)}(x;p|q,t)\quad(r=1,\ldots,n).  
\end{equation}
\end{thm}
If we define 
\begin{equation}
\psi(x;s;p|q,t)=\prod_{1\le i<j\le n}
\frac{\Gamma(tx_j/x_i;p,q)}{\Gamma(x_j/x_i;p,q)} \cdot f(x;s;p|q,t),
\end{equation}
Theorem \ref{thm:RBRe} is equivalently rewritten as
\begin{equation}
\psi(x;s;p|q,q/t)=\gamma(s;p|q,t) \psi(x;s;p|q,t).  
\end{equation}

It is known that $\gamma(s;0;q,t)=1$ by a result of \cite{NS2012}.  
We conjecture that the undetermined constant $\gamma(s;p|q,t)$ 
is 1, namely the function $\psi(x;s;p|q,t)$ is symmetric with respect to 
the reflection $t\leftrightarrow q/t$ of the parameter $t$. 

\subsection{Plan of this paper}\label{ssec:1.5}

The rest of this paper is organized as follows.  
\par\medskip
Theorems \ref{thm:RAFm} and \ref{thm:RBFm} 
on existence of joint eigenfunctions in the formal sense will be 
established in Section \ref{sec:Existence}.  
Once we have set up appropriate classes of 
formal power series on which the joint eigenvalue problems 
should be investigated, their formal solutions are obtained  
along a more or less standard procedure of perturbation theory as in Kato \cite{Kato1995} and Komori-Takemura \cite{KT2002} 
under the separation conditions of eigenvalues.   

The convergence of formal joint eigenfunctions will be discussed in 
Section \ref{sec:Convergence}.  
This problem for the elliptic Ruijsenaars operators requires careful analysis 
of the recurrence relations for the coefficients if we follow the direct approach. 
We first prove the convergence of the joint eigenfunction $f(x;s;p)$ 
(stationary Ruijsenaars function) 
of the modified elliptic Ruijsenaars operators in the asymptotic 
domain $|x_1|\gg\cdots\gg|x_n|\gg|px_1|$.  
The absolute convergence of $f(x;s;p)$ is established by 
the method of majorants for the expansion coefficients defined through 
quadratic recurrences. 
It seems technically demanding, however, to directly apply a similar method to the case 
of symmetric joint eigenfunctions $\cP_\lambda(x;p)$ 
that deform the Macdonald polynomials.  
We overcome this difficulty by employing an integral transform 
that generates symmetric eigenfunctions in the {\em physical domain}
(around the torus) from eigenfunctions in the {\em asymptotic domain}.  

Since this approach of integral transforms seems novel and characteristic in this paper, 
we explain in advance in Section \ref{sec:Prototype} 
how this method works in the trigonometric case, presenting 
a prototype of our arguments of the elliptic case in Section \ref{sec:Convergence}.  

As a consequence of the convergence theorem, 
we include a proof
of Theorem \ref{thm:RAOr} on the orthogonality relations 
for $\cP_\lambda(x;p|q,t)$ 
in Section \ref{sec:Convergence}.  
Theorems \ref{thm:RBRo} and \ref{thm:RBRe}
on symmetries of $f(x;s;p|q,t)$ will be proved in Section \ref{sec:Symmetry}.


\section{Integrable transform to the physical domain: a prototype}\label{sec:Prototype}

In this section, we deal with joint eigenfunctions of 
the Macdonald-Ruijsenaars operators in the trigonometric case, 
and introduce an integral transform that maps 
holomorphic eigenfunctions in the asymptotic domain 
$|x_1|\gg|x_2|\gg\cdots\gg|x_n|$ 
to symmetric eigenfunctions in the {\em physical domain} 
around the torus.  
The method we explain below 
will be extended to the elliptic case in Section \ref{sec:Convergence} 
so as to establish the convergence of the 
symmetric eigenfunctions 
that deform Macdonald polynomials. 


\subsection{Integral transform}

We first recall some basic results from the theory of Macdonald polynomials \cite{M1995}.  
The {\em Cauchy kernel}
\begin{equation}
K(x,y)=\prod_{i=1}^{n}\prod_{j=1}^{n}\frac{(tx_i/y_j;q)_\infty}{(x_i/y_j;q)_\infty}
\end{equation}
for the Macdonald polynomials in variables $x=(x_1,\ldots,x_n)$ and $y=(y_1,\ldots,y_n)$ 
is a meromorphic function on $(\bC^\ast)^n \times(\bC^\ast)^n$, 
and it has an expansion
\begin{equation}\label{eq:Cauchy-trig}
K(x,y)=\sum_{\ell(\lambda)\le n} b_{\lambda}P_{\lambda}(x)P_{\lambda}(y^{-1}),
\quad 
b_{\lambda}=
\prod_{1\le i\le j\le n}
\frac{(t^{j-i+1}q^{\lambda_i-\lambda_j};q)_{\lambda_j-\lambda_{j+1}}}
{(t^{j-i}q^{\lambda_i-\lambda_j+1};q)_{\lambda_j-\lambda_{j+1}}}
\end{equation}
in the domain $|x_i/y_j|<1$ ($i,j\in\pr{1,\ldots,n}$), 
where the sum is taken over 
all partitions $\lambda$ with $l(\lambda)\le n$,
and $(a;q)_k=(1-a)(1-qa)\cdots(1-q^{k-1}a)$ ($k=0,1,2,\ldots$).  
We define the weight function $w(y)$ by
\begin{equation}
w(y)=\prod_{1\le i<j\le n}(1-y_j/y_i)\frac{(qy_j/ty_i;q)_\infty}{(ty_j/y_i;q)_\infty},
\end{equation}
which is holomorphic in the domain  $|y_j/y_i|<|t|^{-1}$ ($1\le i<j\le n$).

For each $\theta>0$, we denote by $V_\theta$ the domain in $(\bC^\ast)^n$ 
specified by 
\begin{equation}
V_\theta=\prm{y=(y_1,\ldots,y_n)\in(\bC^\ast)^n }
{\ |y_2/y_1|<\theta,\ |y_3/y_2|<\theta,\ \ldots,\ |y_n/y_{n-1}|<\theta}.
\end{equation}
Fixing positive constants $r,\sigma$ such that $\sigma<\min\pr{1,\theta,|t|^{-1}}$, 
we define an $n$-cycle $C_{r,\sigma}$ in $(\bC^\ast)^n$ 
by 
\begin{equation}
C_{r,\sigma}=\prm{y=(y_1,\ldots,y_n)\in(\bC^\ast)^n}
{\ |y_i|=\sigma^{i-1}r\ \ (i=1,\ldots,n)}
\end{equation}
with positive orientation.  
Note that $w(y)$ is holomorphic in an neighborhood of $C_{r,\sigma}$ since 
$\sigma<\min\pr{1,|t|^{-1}}$.  
For each holomorphic function $\psi=\psi(y)$ in $V_{\theta}$, we define 
its integral transform $\cK(\psi)$ by 
\begin{equation}
\cK(\psi)(x)=\int_{C_{r,\sigma}} K(x,y)w(y)\psi(y)d\omega_n(y).  
\end{equation}
\begin{lem}
The function $\varphi(x)=\cK(\psi)(x)$ defined as above is continued to an 
entire symmetric holomorphic function in $\bC^n$.  
\end{lem}
\begin{proof}{Proof}  The symmetry of $\varphi(x)$ follows from that of $K(x,y)$ in $x$.  
Since $K(x,y)$ is holomorphic in the domain $|x_i/y_j|<1$ ($i,j=1,\ldots,n$), it is 
is holomorphic in the domain $|x_i|<\sigma^{n-1}r$ $(i=1,\ldots,n)$ when $y\in C_{r,\sigma}$.  
This implies that $\varphi(x)$ is holomorphic in the same domain.  
Since the integral does not depend on $r\in\bR_{>0}$, 
$\varphi(x)$ is continued to an entire function in $\bC^n$.  
\end{proof}


\subsection{From eigenfunctions in the asymptotic domain 
to Macdonald polynomials}\label{ssec:3.2}

In what follows, we use the generating function 
\begin{equation}
D_x(u)=\sum_{r=0}^{n}(-u)^r D^{(r)}_x=\sum_{I\subseteq\pr{1,\ldots,n}}
(-u)^{I}A_I(x)T_{q,x}^{I}
\end{equation}
of the Macdonald-Ruijsenaars operators $D^{(r)}_x$ ($r=0,1,\ldots,n$), 
so that 
\begin{equation}\label{eq:DP=eP}
D_x(u)P_\lambda(x)=\vep(u)P_\lambda(x),\quad
\vep(u)=\sum_{r=0}^{n}(-u)^re_r(t^{\rho}q^{\lambda})=\prod_{j=1}^{n}(1-ut^{n-j}q^{\lambda_j}), 
\end{equation}
for each partition $\lambda$ with $\ell(\lambda)\le n$. 
Also, 
using the multi-index notation 
$T_{q,x}^{\mu}=T_{q,x_1}^{\mu_1}\cdots T_{q,x_n}^{\mu_n}$, 
for a $q$-difference operator $L=\sum_{\mu\in\bZ^n}a_{\mu}(y) T_{q,y}^{\mu}$ 
(finite sum) in $y$ variables, we denote by 
$L^\ast=\sum_{\mu\in\bZ^n}T_{q,y}^{-\mu} a_{\mu}(y)$ 
the formal adjoint of $L$ with respect to $d\omega_n(y)$, noting that 
$(L_1L_2)^\ast=L_2^\ast L_1^\ast$.  

Recall that the Ruijsenaars-Macdonald operator $D_x(u)$ is formally self-adjoint with respect 
to the weight function 
\begin{equation}\label{eq:wsym-trig}
w^{\mathrm{sym}}(y)=\prod_{1\le i<j\le n}\frac{(y_i/y_j;q)_\infty(y_j/y_i;q)_\infty}{(ty_i/y_j;q)_\infty(ty_j/y_i;q)_\infty}
\end{equation}
for the orthogonality of Macdonald polynomials; 
the superscript ``sym" is used to indicate that this weight function is symmetric with respect to 
$y=(y_1,\ldots,y_n)$.  
Namely we have an identity 
\begin{equation}
D_y(u)^{\ast}=w^{\mathrm{sym}}(y)D_{y^{-1}}(u) w^{\mathrm{sym}}(y)^{-1} 
\end{equation}
of $q$-difference operators.  Also, the kernel function $K(x,y)$ satisfies 
the kernel function identity
\begin{equation}
D_x(u)K(x,y)=D_{y^{-1}}(u)K(x,y).  
\end{equation}
(See for instance, \cite{MN1997} and \cite{KNS2009}.)
Hence we have 
\begin{equation}\label{eq:DKw-trig}
D_x(u)K(x,y)w^{\mathrm{sym}}(y)=w^{\mathrm{sym}}(y)D_{y^{-1}}K(x,y)
=D_{y}(u)^\ast (K(x,y)w^{\mathrm{sym}}(y)).  
\end{equation}
We use the idea of integral representation of \cite{MN1997} for our weight 
function $w(y)$ defined above.  
For this purpose we rewrite $w^{\mathrm{sym}}(y)$ as follows:  
\begin{equation}
\begin{split}
w^{\mathrm{sym}}(y)
&=
\prod_{1\le i<j\le n}\frac{(y_j/y_i;q)_\infty(qy_j/ty_i;q)_\infty
}{(ty_j/y_i;q)_\infty(qy_j/y_i;q)_\infty}
\frac{\theta(y_i/y_j;q)}{\theta(ty_i/y_j;q)}
\\
&=
\prod_{1\le i<j\le n}(1-y_j/y_i)
\frac{(qy_j/ty_i;q)_\infty
}{(ty_j/y_i;q)_\infty}
\cdot
\prod_{1\le i<j\le n}
\frac{\theta(y_i/y_j;q)}{\theta(ty_i/y_j;q)}. 
\end{split}
\end{equation}
This means that 
\begin{equation}
w^{\mathrm{sym}}(y)=w(y)g(y),\quad g(y)=\prod_{1\le i<j\le n}
\frac{\theta(y_i/y_j;q)}{\theta(ty_i/y_j;q)}. 
\end{equation}
Noting that 
\begin{equation}
T_{q,y_i}g(y)=t^{n-2i+1}g(y)\qquad(i=1,\ldots,n), 
\end{equation}
we define the $n$-vector $\kappa=(\kappa_1,\ldots,\kappa_n)\in\bC^n$ 
by $\kappa_i=(n-2i+1)\log t/\log q$ $(i=1,\ldots,n)$. 
Then we have
\begin{equation}
T_{q,y}^{\mu}g(y)=q^{\br{\mu,\kappa}}g(y)\qquad(\mu\in \bZ^n),
\end{equation}
where $\br{\mu,\kappa}=\sum_{i=1}^{n}\mu_i\kappa_i$.  
This means that the ratio $g(y)/y^{\kappa}$ is a $q$-periodic function in each 
variable $y_i$ ($i=1,\ldots,n$).  

We introduce the generating function 
\begin{equation}
E_{x,s}(u)=\sum_{I\subseteq\pr{1,\ldots,n}}(-u)^{|I|}s^{\ep_I}
B_I(x)
T_{q,x}^{\ep_I} 
\end{equation}
for the 
modified Ruijsenaars-Macdonald operators $E^{(r)}_{x,s}$ of \eqref{eq:mMROp}, 
so that $D_x(u)=E_{x,t^{\rho}}(u)$.  
Then the Macdonald function 
$f(x;s)$ in the asymptotic domain $|x_1|\gg|x_2|\gg\cdots\gg|x_n|$ 
satisfies 
\begin{equation}
E_{x,s}(u)f(x;s)=\vep(s;u)f(x;s),\quad
\vep(s;u)=\sum_{r=0}^{n}(-u)^r e_r(s)=\prod_{j=1}^{n}(1-us_j).  
\end{equation}
It is known by \cite{NS2012} that 
$f(x;s)$ is holomorphic in the domain 
$|x_j/x_i|<|t/q|$ ($1\le i<j\le n$) and $|s_j/s_i|<1/|q|$ ($1\le i<j\le n$). 

Returning to the kernel function, 
by \eqref{eq:DKw-trig} we obtain
\begin{equation}
\begin{split}
D_x(u)K(x,y)w(y)
&=g(y)^{-1}D_{y}(u)^\ast\big(g(y)K(x,y)w(y)\big)
\\
&=(g(y)D_{y}(u)g(y)^{-1})^{\ast}\big(K(x,y)w(y)\big)
\\
&=(E_{y,t^{\rho}q^{-\kappa}})^\ast(K(x,y)w(y)).  
\end{split}
\end{equation}
Note that 
\begin{equation}
t^{\rho}q^{-\kappa}
=(t^{n-1}\cdot t^{-n+1},t^{n-2}\cdot t^{-n+3},\ldots,1\cdot t^{n-1})
=(1,t,\ldots,t^{n-1}).  
\end{equation}
Hereafter, for a vector $a=(a_1,\ldots,a_n)$, we denote by $a^\vee=(a_n,\ldots,a_1)$ 
the {\em reversal} of $a$ (the action of the longest element of $\frS_n$). 
In this notation of reversals, we obtain the kernel function identity
\begin{equation}
D_x(u)K(x,y)w(y)=(E_{y,{t^{\rho^{\vee}}}})^{\ast}(K(x,y)w(y)). 
\end{equation}

\begin{thm} \label{thm:MacIT}
Let $\lambda$ be a partition with $l(\lambda)\le n$.  
Setting $s=t^\rho q^{\lambda}$, let $f(y;s^{\vee})$ be the Macdonald function in the 
asymptotic domain $V_{\theta}$ 
with parameters $s^\vee=t^{\rho^\vee}\!q^{\lambda^\vee}$
with $\theta^{n-1}<|t/q|$ . 
Then we have
\begin{equation}\label{eq:Kf-trig}
\int_{C_{r,\sigma}}K(x,y)w(y)\,y^{\lambda^\vee}\!f(y;s^\vee)d\omega_n(y)=b_\lambda P_{\lambda}(x). 
\end{equation}
\end{thm}
\begin{proof}{Proof}
Taking $\sigma>0$ such that $\sigma |q|<\min\pr{1,\theta,1/|t|}$, we denote by 
$\varphi(x)=\cK(f(y;s^\vee))$ the left-hand side of \eqref{eq:Kf-trig}.  
In this setting, if $|x_i|<|q|\sigma^{n-1}r$ ($i=1,\ldots,n$), then the $n$-cycle $C_{r,\sigma}$ 
is $q$-shift invariant with respect to the variables $y_i$ ($i=1,\ldots,n$) by Cauchy's theorem.   
Hence we obtain
\begin{equation}
\begin{split}
D_x(u)\varphi(x)
&=
\int_{C_{r,\sigma}}D_x(u)K(x,y)w(y)\,y^{\lambda^\vee}\!f(y;s^\vee)d\omega_n(y)\\
&=
\int_{C_{r,\sigma}}(E_{y,t^{\rho^\vee}})^\ast\big(K(x,y)w(y))\,y^{\lambda^\vee}\!f(y;s^\vee)d\omega_n(y)\\
&=
\int_{C_{r,\sigma}}K(x,y)w(y)
E_{y,t^{\rho^\vee}}\big(y^{\lambda^\vee}\!f(y;s^\vee)\big)d\omega_n(y)\\
&=
\int_{C_{r,\sigma}}K(x,y)w(y)\,
y^{\lambda^\vee}
E_{y,t^{\rho^\vee}q^{\lambda^\vee}}\big(f(y;s^\vee)\big)d\omega_n(y)\\
\end{split}
\end{equation}
Here we compute
\begin{equation}
E_{y,t^{\rho^\vee}}\big(y^{\lambda^\vee}\!f(y;s^\vee)\big)
=
y^{\lambda^\vee}
E_{y,t^{\rho^\vee}q^{\lambda^\vee}}f(y;t^{\rho^\vee}q^{\lambda^\vee})
=\vep(t^{\rho^\vee}q^{\lambda^\vee};u) 
y^{\lambda^\vee}f(y;t^{\rho^\vee}q^{\lambda^\vee}). 
\end{equation}
Since $\vep(s^\vee;u)=\vep(s;u)=\prod_{i=1}^{n}(1-us_i)$, we obtain
\begin{equation}\label{eq:Dxuphi}
D_{x}(u)\varphi(x)=\vep(t^\rho q^{\lambda};u)\varphi(x). 
\end{equation}
Note that the coefficient of $u^n$ gives $t^{\binom{n}{2}}T_{q,x_1}\cdots T_{q,x_n}\varphi(x)=
t^{\binom{n}{2}}q^{|\lambda|}\varphi(x)$, which implies $\varphi(x)$ is homogenous of 
degree $|\lambda|$ in $x$.
Since $\varphi(x)$ is a symmetric holomorphic function in $\bC^n$, it is actually a 
symmetric polynomial satisfying \eqref{eq:Dxuphi}, and hence a constant multiple 
of the Macdonald polynomial $P_{\lambda}(x)$. 
On the other hand, by \eqref{eq:Cauchy-trig} we have
\begin{equation}
\varphi(x)=\sum_{\ell(\mu)\le n}\ c_{\mu}b_{\mu}P_{\mu}(x),
\qquad
c_{\mu}=\int_{C_{r,\sigma}}P_{\mu}(y^{-1})\,y^{\lambda^\vee}\!f(y;s^\vee)d\omega_n(y).  
\end{equation}
Since $\varphi(x)$ is a constant multiple of $P_{\lambda}(x)$, we have 
$c_{\mu}=0$ ($\mu\ne \lambda$) and $\varphi(x)=c_{\lambda}b_{\lambda}P_{\lambda}(x)$.  
In fact we have $c_{\lambda}=1$, since 
$P_{\lambda}(y^{-1})$ and $y^{\lambda^\vee}f(y;s^{\vee})$ and $w(y)$ 
have leading terms 
$y^{-\lambda^\vee}$, $y^{\lambda^\vee}$ and $1$, respectively, with respect to the dominance order.
This implies $\varphi(x)=b_\lambda P_\lambda(x)$.  
\end{proof}

This theorem explains how the Macdonald polynomials are reconstructed 
by the integral transform from the 
Macdonald functions $f(x;s)$ in the asymptotic domain 
$|x_1|\gg|x_2|\gg\cdots\gg|x_n|$.  
The same idea will be used in Section \ref{sec:Convergence} for proving 
the convergence of elliptic deformations of the Macdonald polynomials. 


\section{Existence of formal joint eigenfunctions}\label{sec:Existence}

We begin this section by remarks on the condition on the parameter $t\in\bC^\ast$ that 
we impose for the separation of eigenvalues.  
After giving proofs to Propositions \ref{prop:RAOp} and \ref{prop:RBOp}, 
we establish the existence theorems of formal joint eigenfunctions 
for the two cases of elliptic deformations of Macdonald polynomials, 
and asymptotically free solutions in the domain $|x_1|\gg\cdots\gg|x_n|\gg|px_1|$, 
respectively. 

In order to clarify how the Ruijsenaars operators are related to the root system 
of the Lie algebra $\mathfrak{gl}_n$, we recall here some standard notations and terminologies relevant to our arguments. 
We denote by $\sP=\bZ\ep_1\oplus\cdots\oplus\bZ\ep_n$ the (integral) 
{\em weight lattice} of $\mathfrak{gl}_n$; it is a free $\bZ$-module with 
canonical basis $\ep_1,\ldots,\ep_n$, endowed with the 
the scalar product (symmetric bilinear form) 
$\br{\ ,\,}:\ \sP\times\sP\to\bZ$ such that $\br{\ep_i,\ep_j}=\delta_{i,j}$ 
($i,j\in\pr{1,\ldots,n}$).  We freely identify an integral weight 
$\mu=\mu_1\ep_1+\cdots+\mu_n\ep_n\in\sP$ 
with the $n$-vector of integers $\mu=(\mu_1,\ldots,\mu_n)\in\bZ^n$, 
and denote by $x^{\mu}=x_1^{\mu_1}\cdots x_n^{\mu_n}$ the corresponding 
monomial of the variables $x=(x_1,\ldots,x_n)$. 
Fixing the {\em simple roots} $\alpha_i=\ep_i-\ep_{i+1}\in\sP$ ($i=1,\ldots,n-1$), 
we denote by 
\begin{equation}
\sQ=\bZ\alpha_1\oplus\cdots\oplus\bZ\alpha_{n-1}\subseteq \sP
\end{equation}
the {\em root lattice} of $\mathfrak{gl}_n$, and by 
\begin{equation}
\sP_+=\prm{\lambda\in\sP}{\ \br{\alpha_i,\lambda}\ge 0\ \ (i=1,\ldots,n-1)}\subseteq \sP,
\end{equation}
the cone of {\em dominant integral weights}. 
Note that an integral weight 
$\lambda=\lambda_1\ep_1+\cdots+\lambda_n\ep_n\in \sP$,
or an integer vector $\lambda=(\lambda_1,\ldots,\lambda_n)\in\bZ^n$, 
belongs 
to $\sQ$ if and only if $|\lambda|=\lambda_1+\cdots+\lambda_n=0$, 
and to $\sP_+$ if and only if $\lambda_1\ge\cdots\ge\lambda_n$. 
In terms of the root system, 
the dominance order $\mu\le\nu$ on $\sP$ is defined by the condition
$\nu-\mu\in \sQ_+=\bZ_{\ge 0}\alpha_1\oplus\cdots\oplus\bZ_{\ge 0}\alpha_{n-1}$. 
We denote by 
$\Delta_+=\prm{\ep_i-\ep_j}{1\le i<j\le n}\subseteq \sQ_{+}$ 
the set of {\em positive roots}, so that 
\begin{equation}
\Delta(x)=\prod_{1\le i<j\le n}(x_i-x_j)=
\prod_{i=1}^{n}x_i^{n-i}
\prod_{1\le i<j\le n}(1-x_j/x_i)
=x^{\rho}\prod_{\alpha\in\Delta_+}(1-x^{-\alpha}),
\end{equation}
where $\rho=\sum_{i=1}^{n}(n-i)\ep_i$.  


\subsection{Separation of eigenvalues}\label{ssec:2.1}

Introducing a parameter $u$, we define a generating function $D_x(u)$ 
of the Macdonald-Ruijsenaars operators by 
\begin{equation}
D_x(u)=\sum_{r=0}^{n}(-u)^r D_x^{(r)}=\sum_{I\subseteq\pr{1,\ldots,n}}(-u)^{|I|}A_I(x)T_{q,x}^{\ep_I},
\quad A_I(x)=\frac{T_{t,x}^{\ep_I}(\Delta(x))}{\Delta(x)}. 
\end{equation}
Here we 
used the multi-index notation for 
$T_{q,x}^{\mu}=T_{q,x_1}^{\mu_1}\cdots T_{q,x_n}^{\mu}$ 
for $q$-shift operators, 
setting $\ep_I=\sum_{i\in I}\ep_i\in \sP$ for each subset $I\subseteq\pr{1,\ldots,n}$. 
Then, the joint eigenfunction equations \eqref{eq:MAEq} for the 
Macdonald polynomial $P_\lambda(x)$ ($\lambda\in \sP_+$) can be written as a 
single eigenfunction equation 
\begin{equation}
D_x(u)P_\lambda(x)=d_\lambda(u)P_\lambda(x),\quad
d_\lambda(u)=\sum_{r=0}^{n}(-u)^r e_r(t^{\rho}q^{\lambda})=\prod_{j=1}^{n}(1-ut^{n-j}q^{\lambda_j}),  
\end{equation}
containing a parameter $u$.  
Note that, for a pair $\lambda,\mu\in \sP$ of integral weights, $d_\lambda(u)=d_\mu(u)$ as polynomials in $u$ 
if and only if there exists a permutation $\sigma\in\frS_n$ such that 
\begin{equation}
t^{n-j}q^{\mu_j}=t^{n-\sigma(j)}q^{\lambda_{\sigma(j)}} \quad (j=1,\ldots,n). 
\end{equation}
For a subset $L\subseteq \sP$, we say that the polynomials 
$d_\lambda(u)$ {\em separate} $L$ if $d_\lambda(u)\ne d_{\mu}(u)$ 
for any distinct pair $\lambda,\mu\in L$.  
The condition $t^k\notin q^{\bZ_{<0}}$ $(k=1,\ldots,n-1)$ that we imposed for the existence 
of Macdonald polynomials is equivalent to saying that 
the polynomials $d_\lambda(u)$ separate $\sP_+$. 
In fact, under our assumption $|q|<1$, we have 
\begin{lem}\label{lem:separation}
$(1)$ The polynomials $d_\lambda(u)$ separate $\sP$ if and only if $t^{k}\notin q^{\bZ}$ $(k=1,\ldots,n-1)$. 
\newline
$(2)$ The polynomials $d_\lambda(u)$ separate $\sP_+$ if and only if $t^k\notin q^{\bZ_{<0}}$ 
$(k=1,\ldots,n-1)$.   
\end{lem}
\begin{proof}{Proof} 
(1)\ \ If $d_\lambda(u)=d_\mu(u)$ for a distinct pair $\lambda,\mu\in \sP$, 
we have $t^{n-j}q^{\mu_j}=t^{n-\sigma(j)}q^{\lambda_{\sigma(j)}}$ ($j=1,\ldots,n$) 
for a permutation $\sigma\in\frS_n$.  Since $\lambda\ne \mu$, we have $\sigma\ne 1$, 
and hence there exists an index $j\in\pr{1,\ldots,n}$ such that $\sigma(j)\ne j$ and 
$t^{\sigma(j)-j}=q^{\lambda_{\sigma(j)}-\mu_j}\in q^{\bZ}$,
which means $t^k\in q^{\bZ}$ for $k=|\sigma(j)-j|\in \pr{1,\ldots,n-1}$. 
Conversely, 
suppose that $t^{n-i}=q^{-l}$ for some $i\in\pr{1,\ldots,n-1}$ and $l\in\bZ$, setting $i=n-k$.  
Then we have 
\begin{equation}
\begin{split}
(t^{n-1}q^{\lambda_1},\ldots,t^{n-i}q^{\lambda_i},\ldots,q^{\lambda_n})&=
(t^{n-1}q^{\lambda_1},\ldots,q^{\lambda_i-l},\ldots,t^{n-i}q^{\lambda_n+l})
\\
&=
(i,n). (t^{n-1}q^{\lambda_1},\ldots,t^{n-i}q^{\lambda_n+l},\ldots,q^{\lambda_i-l}),
\end{split}
\end{equation}
where $(i,n)$ stands for the transposition of $i$ and $n$. 
This implies that $d_{\lambda}(u)=d_{\mu}(u)$ for the pair 
\begin{equation}\label{eq:ldmu}
\lambda=(\lambda_1,\ldots,\lambda_i,\ldots,\lambda_n),\quad 
\mu=(\lambda_1,\ldots,\lambda_n+l,\ldots,\lambda_i-l)\in \sP. 
\end{equation}
(2)\ \ 
We first show that, if $|t|\le 1$, then the polynomials $d_{\lambda}(u)$ separate 
$\lambda\in \sP_+$.  
Under the assumption $|t|\le 1$, 
the sequence $|t^{n-j}q^{\lambda_j}|$ ($j=1,\ldots,n$) is non-decreasing 
for any dominant $\lambda\in \sP_+$. 
If $d_\lambda(u)=d_\mu(u)$ for a pair $\lambda,\mu\in \sP_+$, then we have 
$|t^{n-j}q^{\lambda_j}|=|t^{n-j}q^{\mu_j}|$ 
and hence $|q|^{\lambda_j}=|q|^{\mu_j}$ for $j=1,\ldots,n$.  
This implies $\lambda_j=\mu_j$ ($j=1,\ldots,n$) (since $|q|<1$).  
Suppose now that $d_\lambda(u)=d_\mu(u)$ for a distinct pair $\lambda,\mu\in\sP$.  
Then we have $|t|>1$, and also $t^k\in q^{\bZ}$ for some $k\in\pr{1,\ldots,n-1}$ by (1).  
Hence $t^k\in q^{\bZ_{<0}}$ for some $k\in\pr{1,\ldots,n-1}$. 
Conversely, suppose that $t^{n-i}=q^{-l}$ for some $i\in\pr{1,\ldots,n-1}$ and $l\in\bZ_{>0}$.  
Then $d_\lambda(u)=d_{\mu}(u)$ for 
$\lambda=(m,\ldots,m,0,0\ldots,0)$ with $m$'s in the first $i-1$ places, 
and $\mu=(m,\ldots,m,l,0,\ldots,-l)$, both in $\sP_+$ 
if $m\ge l$. 
\end{proof}

The following lemma is crucial in our proof of the existence theorem of formal joint eigenfunctions 
for the elliptic Ruijsenaars operators. 
\begin{lem}\label{lem:genericconst}
If the polynomials $d_\lambda(u)$ separate a subset $L\subseteq \sP$, then 
there exists a constant $c\in\bC^{\ast}$, such that $d_\lambda(c)\ne d_\mu(c)$ 
for any distinct pair $\lambda,\mu\in L$.  
\end{lem}
\begin{proof}{Proof}
The set 
\begin{equation}\label{eq:defineS}
S=\prm{a\in \bC^\ast}{\ d_\lambda(a)=d_{\mu}(a)\ \ \mbox{for some distinct pair $\lambda,\mu\in L$}}\subseteq \bC^\ast
\end{equation}
is countable, and hence the complement $\bC^\ast\backslash S$ is non-empty.  
This implies that any element $c\in\bC^\ast\backslash S$ satisfies the condition 
$d_{\lambda}(c)\ne d_{\mu}(c)$ for any distinct pair $\lambda,\mu\in\bC^\ast$. 
\end{proof}


\subsection{Elliptic Ruijsenaars operators as power series of $p$} \label{ssec:2.2}
In this subsection, 
we clarify the structure of $p$-expansion of the elliptic Ruijsenaars operators, 
thereby proving Propositions \ref{prop:RAOp} and \ref{prop:RBOp}. 
\par\medskip
Similarly to the trigonometric case, we introduce the generation function
\begin{equation}
\cD_x(p;u)=\sum_{r=0}^{n}(-u)^{r}\cD^{(r)}_x(p)=
\sum_{I\subseteq \pr{1,\ldots,n}}(-u)^{|I|}\cA_I(x;p)T_{q,x}^{\ep_I},\quad
\cA_I(x;p)=\frac{T_{t,x}^{\ep_I}(\Delta(x;p))}{\Delta(x;p)},
\end{equation}
for the elliptic Ruijsenaars operators, 
where $\Delta(x;p)=x^{\rho}\prod_{1\le i<j\le n}\theta(x_j/x_i;p)$. 
We express the coefficients 
$\cA_I(x;p)$ in two ways:
\begin{equation}
\cA_I(x;p)=t^{\br{\ep_I,\rho}}\cB_I(x;p)=A_I(x)\cC_I(x;p),
\end{equation}
where $A_I(x)=\cA_I(x;0)$ are the coefficients of $D_x(u)$ 
in \eqref{eq:MRAB} and 
\begin{equation}\label{eq:defBC}
\begin{split}
\cB_I(x;p)&=
\prod_{\substack{1\le i<j\le n\\ i\in I,\,j\notin I}}
\frac{\theta(x_j/tx_i;p)}{\theta(x_j/x_i;p)}
\prod_{\substack{1\le i<j\le n\\ i\notin I,\,j\in I}}
\frac{\theta(tx_j/x_i;p)}{\theta(x_j/x_i;p)},
\\
\cC_I(x;p)&=
\prod_{ i\in I,\,j\notin I}
\frac{(ptx_i/x_j;p)(px_j/tx_i;p)}{(px_i/x_j;p)(px_j/x_i;p)}. 
\end{split}
\end{equation}
We remark here that, for each $I\subseteq\pr{1,\ldots,n}$, 
$\cB_I(x;p)$ is a holomorphic function of the variables 
$z_0=px_1/x_n, z_1=x_2/x_1, \ldots,z_{n-1}=x_n/x_{n-1}$ 
in the polydisc $|z_0|<1, |z_1|<1,\ldots,|z_{n-1}|<1$.
In fact, we have $z_0z_1\cdots z_{n-1}=p$ and 
\begin{equation}
x_j/x_i=z_i\cdots z_{j-1},\quad px_i/x_j=z_j\cdots z_{n-1}z_0\cdots z_{i-1}\quad(1\le i<j\le n).  
\end{equation}
This implies in particular that $\cB_I(x;p)\in\bC\pr{px_1/x_n,x_2/x_1,\ldots,x_n/x_{n-1}}$, and hence, 
$\cB_I(x;p)$ is represented by a power series in $p$ of the form
\begin{equation}\label{eq:BIk}
\cB_I(x;p)=\sum_{k=0}^{\infty}p^k\,
\cB_{I,k}(x),
\quad \cB_{I,k}(x)\in (x_1/x_n)^k\bC\pr{x_2/x_1,\ldots,x_n/x_{n-1}}\quad(k=0,1,2,\ldots),
\end{equation}
which proves Proposition \ref{prop:RBOp} for $\cA_I(x;p)=t^{\br{\ep_I,\rho}}\cB_I(x;p)$. 
On the other hand, from \eqref{eq:defBC}, $\cC_I(x;p)$ is 
holomorphic in the domain $|p|<|x_j/x_i|<|p|^{-1}$, and 
it can be represented there
by a power series of $p$ whose coefficients are Laurent polynomials in $x$: 
\begin{equation}\label{eq:CIk}
\cC_I(x;p)=\sum_{k=0}^{p}p^k \cC_{I,k}(x),\quad \cC_{I,k}(x)\in \bC[x^{\pm1}]
\quad(k=0,1,2,\ldots).  
\end{equation}
Hence, 
each coefficients of $\cA_I(x;p)$ of the elliptic Ruijsenaars operators has a $p$-expansion of the form
\begin{equation}
\cA_I(x;p)=\sum_{k=0}^{\infty}p^k\,\cA_{I,k}(x),
\quad
\cA_{I,k}(x)=t^{\br{\ep_I,\rho}}\cB_{I,k}(x)=A_I(x)\cC_{I,k}(x)
\quad(k=0,1,2,\ldots). 
\end{equation}
Note that both expansions \eqref{eq:BIk} and \eqref{eq:CIk} make sense 
in the domain $|p|^{\frac{1}{n-1}}<|x_{i+1}/x_{i}|<1$ ($i=1,\ldots,n-1$). 
Combining \eqref{eq:BIk} and \eqref{eq:CIk}, we have
\begin{equation}\label{eq:AIk}
\cA_{I,k}(x)\in \big(A_I(x)\bC[x^{\pm1}]\big)\cap\big((x_1/x_n)^{k}\bC\pr{x_2/x_1,\ldots,x_n/x_{n-1}}\big)
\end{equation}
for $k=0,1,2,\ldots$.  
For each $r=1,\ldots,n$, 
the elliptic Ruijsenaars operator 
$\cD_x^{(r)}(p)$ is a $\frS_n$-invariant $q$-difference operator, 
and hence 
the coefficients $\cD^{(r)}_{x,k}(u)$ $(k=0,1,2\ldots)$
in the $p$-expansion
\begin{equation}
\cD^{(r)}_{x}(p)=\sum_{k=0}^{\infty}p^k\,\cD^{(r)}_{x,k},\quad
\cD^{(r)}_{x,k}=\sum_{I\subseteq\pr{1,\ldots,n};\ |I|=r} \cA_{I,k}(x)T_{q,x}^{\ep_I}
\end{equation}
are also $\frS_n$-invariant.  
Since $\Delta(x)\cD^{(r)}_{x,k}$ $(k=0,1,2,\ldots)$ 
have Laurent polynomial coefficients as well, 
they stabilize the ring $\bC[x^{\pm1}]^{\frS_n}$ of symmetric Laurent polynomials, 
namely
$\cD^{(r)}_{x,k}\big(\bC[x^{\pm1}]^{\frS_n}\big)\subseteq \bC[x^{\pm1}]^{\frS_n}$. 
(Note that $\cD^{(r)}_{x,k}$ do {\em not} stabilize $\bC[x^{\pm1}]$.)
Suppose that 
$\varphi(x)\in \bC[x^{\pm1}]^{\frS_n}_{\le \mu}$ for some $\mu\in\sP_+$.  
Since $\varphi(x)\in x^{\mu}\bC\cps{x_2/x_1,\ldots,x_n/x_{n-1}}$, 
by \eqref{eq:AIk} we obtain 
\begin{equation}
\psi(x)=\cD^{(r)}_{x,k}\big(\varphi(x)\big)\in  x^{\mu+k\phi}\bC\cps{x_2/x_1,\ldots,x_n/x_{n-1}}, 
\end{equation}
where $x^{\phi}=x_1/x_n$.  Hence
\begin{equation}
\psi(x)\in \bC[x^{\pm1}]^{\frS_n}\cap\big( x^{\mu+k\phi}\bC\cps{x_2/x_1,\ldots,x_n/x_{n-1}}\big)
=\bC[x^{\pm1}]^{\frS_n}_{\le \mu+k\phi}.
\end{equation}
This means that 
$\cD^{(r)}_{x,k}\big(\bC[x^{\pm1}]^{\frS_n}_{\le \mu}\big)\subseteq 
\bC[x^{\pm1}]^{\frS_n}_{\le \mu+k\phi}$ 
for all $k=0,1,2,\ldots$ and $\mu\in \sP_+$,
which proves Proposition \ref{prop:RAOp}.  


\subsection{Elliptic deformation of Macdonald polynomials}\label{ssec:2.3}

In this subsection, we complete our proof of Theorem \ref{thm:RAFm}. 
\par\medskip
Denoting by  $\cR=\bC[x^{\pm1}]^{\frS_n}$ the ring of symmetric 
Laurent polynomials in $x=(x_1,\ldots,x_n)$, 
we consider the commuting family of $\bC\fps{p}$-linear operators
\begin{equation}
\cD^{(r)}_x(p):\ \ \cR\fps{p}\to\cR\fps{p}\qquad(r=1,\ldots,n)
\end{equation}
defined by the elliptic Ruijsenaars operators, and investigate the joint eigenvalue problems 
\begin{equation}\label{eq:EVPA}
\cD^{(r)}_{x}(p)\psi(x;p)=\vep^{(r)}(p)\psi(x;p)\quad(r=1,\ldots,n)
\end{equation}
where 
\begin{equation}
\psi(x;p)=\sum_{k=0}^{\infty}p^k \psi_k(x)
\in \cR\fps{p},\quad \vep^{(r)}(p)
=\sum_{k=0}^{\infty}p^k\vep^{(r)}_k\in\bC\fps{p}. 
\end{equation}
In terms of the generating function $\cD_x(p;u)=\sum_{r=0}^{n}(-u)^r \cD^{(r)}_x(p)$, 
this joint eigenvalue problem is expressed as an equation
\begin{equation}\label{eq:EVPAu}
\cD_x(u)\psi(x;p)=\vep(p;u)\psi(x;p),\quad \vep(p;u)=\sum_{r=0}^{n}(-u)^{r}\vep^{(r)}(p)
=\sum_{k=0}^{\infty}p^k \vep_k(u)
\end{equation}
with a parameter $u$. 

\par\medskip
We assume hereafter that the parameter $t$ satisfies the condition
$t^{k}\notin q^{\bZ_{<0}}$ ($k=1,\ldots,n-1$).  
Then, by Lemma \ref{lem:separation}, the polynomials 
$d_{\lambda}(u)=\prod_{j=1}^n(1-ut^{n-j}q^{\lambda_j})$ 
in $u$ separate the cone $\sP_+$ of dominant integral weights. 
We also fix an arbitrary constant $c\in\bC^\ast\backslash S$ with $S$ 
defined as in \eqref{eq:defineS} of Lemma \ref{lem:genericconst}, 
so that $d_\lambda(c)\ne d_{\mu}(c)$ for any distinct $\lambda,\mu\in\sP_+$. 
We first solve equation \eqref{eq:EVPAu}, 
specializing the parameter $u$ to $c$, 
and then show the solution $\psi(x;p)$ gives a joint eigenfunction for 
all $\cD^{(r)}_x(p)$ ($r=1,\ldots,n$).  

In terms of the expansion coefficients, the eigenfunction equation 
\begin{equation}\label{eq:EVPAc}
\cD_x(p;c)\psi(x;p)=\vep(p;c)\psi(x;p),\quad\psi(x;p)\in\cR\fps{p}
\end{equation}
is transferred to the recurrence relations
\begin{equation}\label{eq:reck0}
\sum_{i=0}^{k}\cD_{x,i}(c)\psi_{k-i}(x)=\sum_{i=0}^{k}\vep_{i}(c)\psi_{k-i}(x)
\qquad(k=0,1,2,\ldots).
\end{equation}
The conditions for $k=0$ is the  
eigenfunction equation
\begin{equation}
D_x(c)\psi_0(x)=\vep_0(c)\psi_0(x) 
\end{equation}
for $\psi_0(x)\in\cR=\bC[x^{\pm1}]^{\frS_n}$. 
If $\psi_{0}(x)\ne 0$, it must be a constant multiple of 
the Macdonald polynomial $P_\lambda(x)$ for some $\lambda\in\sP_+$. 
Fixing a dominant integral weight $\lambda\in\sP_+$, we take 
\begin{equation}
\psi_0(x)=P_{\lambda}(x),\quad \vep_0(c)=d_\lambda(c)
\end{equation}
for the initial data. 
We rewrite the recurrence relations 
for $k>0$ as 
\begin{equation}\label{eq:reck1}
\cD_{x,0}(c)\psi_{k}(x)+\sum_{i=1}^{k}\cD_{x,i}(c)\psi_{k-i}(x)
=\vep_0(c)\psi_k(x)+\sum_{i=1}^{k-1}\vep_i(c)\psi_{k-i}(x)+\vep_k(c)\psi_0(x), 
\end{equation}
namely,
\begin{equation}\label{eq:reck2}
\big(D_{x}(c)-d_\lambda(c)\big)\psi_{k}(x)-\vep_k(c)P_\lambda(x)
=
-\sum_{i=1}^{k}\cD_{x,i}(c)\psi_{k-i}(x)
+\sum_{i=1}^{k-1}\vep_i(c) \psi_{k-i}(x). 
\end{equation}
At this step, regarding the right-hand side as known, 
we need to solve this equation for $\psi_k(x)\in\cR$ and $\vep_k(c)\in\bC$. 
Under the genericity assumption on $c$, 
we verify that $\psi_{k}(x)$ and $\vep_{k}(c)$ are determined inductively 
from $\psi_i(x)$ and $\vep_{i}(c)$ ($i=0,1,\ldots,k-1$).  
For that purpose, we use the decomposition 
\begin{equation}
\cR=\bC P_\lambda(x)\oplus \cS_{\lambda},\quad
\cS_\lambda=\bigoplus_{\mu\in P_+;\,\mu\ne \lambda}\bC P_{\mu}(x). 
\end{equation}
of $\cR$ by means of the Macdonald basis. 
Note that 
\begin{equation}
\big(D_x(c)-d_\lambda(c)\big)P_\mu(x)
=\big(d_\mu(c)-d_\lambda(c))P_\mu(x)
\end{equation}
and $d_{\mu}(c)\ne d_{\lambda}(c)$ for any $\mu\in P_+$ with $\mu\ne\lambda$.  
This implies that
$\mathrm{Ker}(D_x(c)-d_\lambda(c): \cR\to\cR)=\bC P_{\lambda}(x)$ 
and that $D_x(c)-d_{\lambda}(c)$ induces a $\bC$-isomorphism $\cS_{\lambda}\isom\cS_{\lambda}$. 
In the recurrence relation \eqref{eq:reck2}, decompose the right-hand side into the 
form $a_{k} P_{\lambda}(x)+\varphi_{k}(x)$ 
with $a_{k}\in\bC$ and $\varphi_{k}(x)\in \cS_\lambda$.  Then 
$\psi_{k}(x)$ and $\vep_{k}(c)$ are determined by
\begin{equation}
\vep_{k}(c)=-a_{k},\quad (D_x(c)-d_\lambda(c))\psi_{k}(x)=\varphi_{k}(x).  
\end{equation}
If we assume $\psi_{k}(x)\in \cS_\lambda$ ($k>0$), then 
$\psi_{k}(x)$ and $\vep_k(c)$ are determined uniquely by the recurrence.  
Otherwise, $\psi_{k}(x)$ is determined with freedom of 
adding a multiple of $P_{\lambda}(x)$ at each $k$th step.  
(As we will see below, this ambiguity corresponds to the multiplication 
by a power series $\gamma(p)\in\bC\fps{p}$ with $\gamma(0)=1$.)
Also, it is verified inductively by the recurrence \eqref{eq:reck2}
that $\psi_{k}(x)\in\cR_{\le \lambda+k\phi}$, 
since $\psi_{\lambda,k-i}(x)\in\cR_{\le\lambda+(k-i)\phi}$ 
and 
$\cD_{x,i}(c)(\cR_{\le\lambda+(k-i)\phi})\subseteq \cR_{\le\lambda+k\phi}$
by Proposition \ref{prop:RAOp}. 
\par\medskip
From now on, we denote by 
\begin{equation}
\varphi_\lambda(x;p)=\sum_{k=0}^{\infty}p^k\,\varphi_{\lambda;k}(x)\in \cR\fps{p}
\end{equation}
the {\em unique} formal solution of the eigenfunction equation \eqref{eq:EVPAc}
determined by the condition 
that $\varphi_{\lambda;0}(x)=P_{\lambda}(x)$, 
and $\varphi_{\lambda;k}(x)\in\cS_{\lambda}$ ($k>0$).  
(This normalization of $\varphi_{\lambda}(x;p)$ 
is different from the one we adopted below Theorem \ref{thm:RAFm}.)
Also, we denote by $\vep_\lambda(p;c)$ the corresponding formal eigenvalue.
At this stage, $\varphi_{\lambda}(x;p)$ may depend on the constant 
$c\in \bC^\ast\backslash S$ 
that we started with, but it is not the case as it will turn out later. 
For this reason, we are suppressing the possible dependence on $c$ in the notation 
of $\varphi_\lambda(x;p)$. 
We also remark that any formal solution of the eigenfunction equation 
\eqref{eq:EVPAc}
with initial condition $\psi(x;0)=P_{\lambda}(x)$ 
is expressed as $\psi(x;p)=\gamma(p)\varphi_{\lambda}(x;p)$ for some $\gamma(p)\in\bC\fps{p}$
with $\gamma(0)=1$.  
In fact, express $\psi(x;p)$ as 
\begin{equation}
\psi(x;p)=\sum_{k=0}^{\infty}p^k \psi_k(x),\quad \psi_k(x)=\sum_{\mu\in P_+}\psi_{k,\mu}P_{\mu}(x)
\quad(k=0,1,2,\ldots)
\end{equation}
with $\psi_0(x)=P_{\lambda}(x)$ and $\psi_{0,\mu}=\delta_{\mu,\lambda}$.  
If we set $\gamma(p)=\sum_{k=0}^{\infty}p^k\psi_{k,\lambda}$, it is easy to show 
that $\gamma(p)^{-1}\psi(x;p)$ satisfies the normalization condition 
for $\varphi_{\lambda}(x;p)$. 
\par\medskip
We next show 
that the formal solution $\varphi_{\lambda}(x;p)$ is in fact a joint eigenfunction 
for the elliptic Ruijsenaars operators $\cD^{(r)}_x(p)$ ($r=1,\ldots,n$) acting on 
$\cR\fps{p}$.  
Since each $\cD_x^{(r)}(p)$ commute with $\cD_x(p;c)$, we have 
\begin{equation}
\cD_x(p;c)\cD_{x}^{(r)}(p)\varphi_{\lambda}(x;p)
=\cD_{x}^{(r)}(p)\cD_x(p;c)\varphi_{\lambda}(x;p)=\vep_{\lambda}(p;c)\cD^{(r)}_x(p)\varphi_\lambda(x;p). 
\end{equation}
This means that $\cD^{(r)}_x(p)\varphi_\lambda(x;p)$ is a formal solution of the 
eigenfunction equation \eqref{eq:EVPAc} with initial condition 
\begin{equation}
\cD^{(r)}_x(p)\varphi_\lambda(x;p)\big|_{p=0}=D^{(r)}_xP_\lambda(x)=e_r(t^\rho q^\lambda)P_{\lambda}(x). 
\end{equation}
Hence, $\cD^{(r)}_x(p)\varphi_{\lambda}(x;p)$ is expressed as 
$\vep^{(r)}_\lambda(p)\varphi_{\lambda}(x;p)$ with some $\vep^{(r)}_\lambda(p)\in\bC\fps{p}$ 
such that $\vep^{(r)}_\lambda(0)=e_r(t^\rho q^{\lambda})$.  
Namely, we have
\begin{equation}
\cD^{(r)}_{x}(p)\varphi_\lambda(x;p)=\vep^{(r)}_\lambda(p)\varphi_{\lambda}(x;p),
\quad
\vep^{(r)}_\lambda(0)=e_r(t^\rho q^{\lambda})\quad(r=0,1,\ldots,n),
\end{equation}
as desired.  In particular, we have 
\begin{equation}
\cD_x(p;u)\varphi_\lambda(x;p)=\vep_\lambda(p;u)\varphi_{\lambda}(x;p),
\quad \vep_\lambda(p;u)=\sum_{r=0}^{n}(-1)^r \vep^{(r)}_{\lambda}(p)
\end{equation}
for arbitrary $u\in\bC$.  This implies also that $\varphi_\lambda(x;p)$ does {\em not} depend on 
the constant $c\in \bC^\ast\backslash S$ that 
we have chosen for the construction of $\varphi_\lambda(x;p)$.  
This completes the proof of Theorem \ref{thm:RAFm}. 
\begin{rem}\rm 
If we regard $q$ and $t$ as indeterminates, 
the formal solution $\varphi_\lambda(x;p)$ mentioned above, as well as 
the normalized formal solution in the sense of \eqref{eq:cPpexp2}, 
is uniquely determined 
as a formal power series of $p$ with 
coefficients in $\bC(q,t)[x^{\pm1}]^{\frS_n}$. 
\end{rem}
\subsection{Asymptotically free eigenfunctions in the domain $|x_1|\gg\cdots\gg|x_n|\gg|px_1|$}\label{ssec:2.4}

In this subsection, we complete our proof of Theorem \ref{thm:RBFm}. 
\par\medskip
As we remarked in Subsection \ref{ssec:2.2}, the coefficients 
$\cA_{I}(x;p)$ of the elliptic Ruijsenaars operators 
define formal power series in 
\begin{equation}
\bC\fps{px_1/x_n,x_2/x_1,\ldots,x_n/x_{n-1}},  
\end{equation}
which is the set of all formal power series of the form 
\begin{equation}\label{eq:phi=sum}
\varphi(x;p)=\sum_{l_0,l_1,\ldots,l_{n-1}\in\bZ_{\ge0}} \varphi_{l_0,l_1,\dots,l_{n-1}} (px_1/x_n)^{l_0}
(x_2/x_1)^{l_1}\cdots (x_n/x_{n-1})^{l_{n-1}}. 
\end{equation}
In order to simplify the presentation, 
we use below the notation of monomials attached 
to the root lattice $\sQ^\aff$ of the affine Lie algebra $\widehat{\mathfrak{gl}}_n$ 
\cite{Kac1990}.  
Adding the {\em null root} $\delta$ to the basis, we extend the weight lattice $\sP$ 
of $\mathfrak{gl}_n$ to 
\begin{equation}
\widetilde{\sP}=\bZ\delta\oplus\sP
=\bZ\delta\oplus\bZ\ep_1\oplus\cdots\oplus\bZ\ep_n, 
\end{equation}
together with the scalar product $\widetilde{\sP}\times\widetilde{\sP}\to\bZ$ extended by 
$\br{\mu,\delta}=\br{\delta,\mu}=0$ ($\mu\in\sP$) and $\br{\delta,\delta}=0$.  
For each element $\lambda=k\delta+\mu\in\widetilde{\sP}$ 
$(k\in\bZ,\ \mu\in\sP)$, 
we attach the monomial 
\begin{equation}
x^{\lambda}=p^{-k}x^{\mu}=p^{-k}x_1^{\mu_1}\cdots x_n^{\mu_n}
\end{equation}
with the convention $x^{\delta}=p^{-1}$.  Introducing the 0th simple affine root 
$\alpha_0=\delta-\phi=\delta-\ep_1+\ep_n$, we denote by
\begin{equation}
\sQ^{\aff}=\bZ\alpha_0\oplus\bZ\alpha_1\oplus\cdots\oplus\bZ\alpha_{n-1}
=\bZ\delta\oplus \sQ
\subseteq \widetilde{\sP}
\end{equation}
the root lattice of the affine Lie algebra $\widehat{\mathfrak{gl}}_n$. 
Note here that $\alpha_0+\alpha_1+\cdots+\alpha_{n-1}=\delta$.  
We also set $\sQ^\aff_+=\bZ_{\ge0}\alpha_0\oplus\bZ_{\ge0}\alpha_1\oplus\cdots
\oplus\bZ_{\ge0}\alpha_{n-1}.$
Noting that
\begin{equation}
x^{-\alpha_0}=px_1/x_n,\ \ 
x^{-\alpha_1}=x_2/x_1,\ \ \ldots,\ \ x^{-\alpha_{n-1}}=x_n/x_{n-1}, 
\end{equation}
we rewrite the monomials in $\bC\fps{px_1/x_n,x_2/x_1,\ldots,x_n/x_{n-1}}$ 
as 
\begin{equation}
(px_1/x_n)^{l_0}(x_2/x_1)^{l_1}\cdots (x_n/x_{n-1})^{l_{n-1}}
=x^{-\beta},\quad \beta=l_0\alpha_0+l_1\alpha_1+\cdots+l_{n-1}\alpha_{n-1},
\end{equation}
and call 
\begin{equation}
\htaf(\beta)=l_0+l_1+\cdots+l_{n-1}\in\bZ_{\ge 0}
\end{equation}
the {\em height} of $\beta\in \sQ^\aff_+$.  
In this notation, the formal power series in 
$\bC\fps{px_1/x_n,x_2/x_1,\ldots,x_n/x_{n-1}}$ 
are simply expressed as
\begin{equation}
\varphi(x;p)=\sum_{\beta\in\sQ^\aff_+}\varphi_{\beta} x^{-\beta},
\quad
\varphi_\beta\in\bC\ \ (\beta\in\sQ^\aff_+).  
\end{equation}

We recall the definition of the {\em modified elliptic Ruijsenaars operators} 
$\cE^{(r)}_{x,s}(p)$
that we introduced in Subsection \ref{ssec:1.4}: 
\begin{equation}
\cE^{(r)}_{x,s}(p)=
\sum_{I\subseteq\pr{1,\ldots,n};\ |I|=r} 
s^{\ep_I}
\cB_{I}(x;p)T_{q,x}^{\ep_I}\qquad(r=0,1,\ldots,n), 
\end{equation}
where
\begin{equation}
\cB_{I}(x;p)=
\prod_{\substack{1\le i<j\le n\\ i\in I;\,j\notin I}}
\frac{\theta(x_j/tx_i;p)}{\theta(x_j/x_i;p)}
\prod_{\substack{1\le i<j\le n\\ i\notin I;\,j\in I}}
\frac{\theta(tx_j/x_i;p)}{\theta(x_j/x_i;p)}
\qquad(I\subseteq\pr{1,\ldots,n}). 
\end{equation}
Note that the coefficients $\cB_I(x;p)$ are expressed as 
\begin{equation}\label{eq:expBIx}
\cB_I(x;p)=\sum_{\beta\in \sQ^{\aff}_+}x^{-\beta}\,b^{I}_\beta\in 
\bC\fps{px_1/x_n,x_2/x_1,\ldots,x_n/x_{n-1}},\quad b^{I}_0=1, 
\end{equation}
through the expansion in the domain  
$|px_1/x_n|<1$, $|x_2/x_1|<1$, $\ldots$, $|x_n/x_{n-1}|<1$. 

\par\medskip
In what follows, we regard $s=(s_1,\ldots,s_n)\in(\bC^\ast)^n$ as fixed constants, 
supposing that the genericity condition $s_j/s_i\notin q^{\bZ}$ ($1\le i<j\le n$) is satisfied.
In this setting, we investigate the joint eigenvalue problem 
\begin{equation}
\cE^{(r)}_{x,s}(p)\varphi(x;p)=\vep^{(r)}(p)\varphi(x;p),
\quad 
\vep^{(r)}(0)=e_r(s)
\qquad(r=1,\ldots,n)
\end{equation}
for a formal power series 
\begin{equation}
\varphi(x;p)=\sum_{\mu\in \beta\in\sQ^\aff_+}\varphi_{\beta} x^{-\beta}
\in\bC\fps{px_1/x_n,x_2/x_1,\ldots,x_n/x_{n-1}}
\end{equation}
with leading term $\varphi_0=1$.  
Here we have suppressed the dependence on $s$ in the notation of 
$\varphi(x;p)$ and $\vep^{(r)}(p)$.  
We also use the generation functions
\begin{equation}
\cE_{x,s}(p;u)=\sum_{r=0}^{n} (-u)^r \cE^{(r)}_{x,s}(p),
\quad
\vep(p;u)=\sum_{r=0}^{n}(-u)^r\vep^{(r)}(p)
\end{equation}
and consider the eigenfunction equation
\begin{equation}
\cE_{x,s}(p;u)\varphi(x;p)=\vep(p;u)\varphi(x;p),\quad \vep(0;u)=\prod_{i=1}^{n}(1-s_iu),
\end{equation}
with a parameter $u$. 

In view of the expansion \eqref{eq:expBIx}, we express the operator
$\cE_{x,s}(p;u)$ as follows: 
\begin{equation}
\begin{split}\label{eq:defbfun}
\cE_{x,s}(p;u)
&=\sum_{\beta\in \sQ_+^\aff}
x^{-\beta}\,b_{\beta}(T_{q,x};u),
\quad
b_{\beta}(\xi;u)=
\sum_{I\subseteq\pr{1,\ldots,n}}(-u)^{|I|}b^{I}_{\beta}s^{\ep_I}\xi^{\ep_I}, 
\end{split}
\end{equation}
with symbols $\xi=(\xi_1,\ldots,\xi_n)$ attached to $T_{q,x}=(T_{q,x_1},\ldots,T_{q,x_n})$. 
Since $b^{I}_0=1$, we have 
\begin{equation}
b_0(\xi;u)=\prod_{j=1}^{n}(1-us_j\xi_j). 
\end{equation}

\begin{lem}\label{lem:43}\ \ 
Suppose that $s_j/s_i\notin q^{\bZ^\ast}$ $(1\le i<j\le n)$, 
where $\bZ^\ast=\bZ\backslash\pr{0}$. 
Then, there exists a countable set $S\subseteq\bC^\ast$ such that,  if $c\in\bC^\ast\backslash S$, 
\begin{equation}
\prod_{j=1}^{n}(1-cs_jq^{m_j})\ne\prod_{j=1}^{n}(1-cs_j)
\end{equation}
for any nonzero $m=(m_1,\ldots,m_n)\in\bZ^n$.  
\end{lem}
\begin{proof}{Proof}
For each $m\in\bZ^n\backslash\pr{0}$, 
$g_m(u)=\prod_{j=1}^{m}(1-us_jq^{m_j})-\prod_{j=1}^{m}(1-us_j)$ is a nonzero 
polynomial in $u$ by the assumption $s_j/s_i\notin q^{\bZ^\ast}$ ($1\le i<j\le n)$. 
(In fact, 
suppose that there exists a nonzero $m\in\bZ^n$ and 
a permutation $\sigma\in\frS_n$ such 
$s_jq^{m_j}=s_{\sigma(j)}$ for $j=1,\ldots,m$.    
Then, under our assumption that $|q|<1$, 
for each index $i\in\pr{1,\ldots,n}$ such that $m_i\ne0$, we have
$\sigma(i)\ne i$ and $s_{\sigma(i)}/s_i\in q^{\bZ^\ast}$, 
leading to a contradiction.) 
Then, the set
\begin{equation}
S=\prm{a\in\bC^\ast}{\ g_m(a)=0\ \mbox{for some}\ m\in\bZ^{m}\backslash\pr{0}}
\end{equation}
is countable, and fulfills the condition of Lemma. 
\end{proof}

Fixing a {\em generic} constant $c\in\bC^\ast\backslash S$ in this sense,  
we investigate the eigenfunction equation
\begin{equation}\label{eq:EFEqf}
\cE_{x,s}(p;c)f(x;p)=\vep(p;c)f(x;p),
\quad \vep(p;c)=\sum_{k=0}^{\infty}p^k\vep_k(c)=\sum_{k=0}^{\infty}x^{-k\delta}\vep_k(c),
\end{equation}
where 
\begin{equation}
f(x;p)=\sum_{\mu\in \sQ^{\aff}_+} x^{-\mu}\,f_{\mu}\in\bC\fps{px_1/x_n,x_2/x_1,\ldots,x_n/x_{n-1}}.  
\end{equation}
We have suppressed the dependence of $f(x;p)$ and $\vep(p;u)$ on the $s$ parameters in order to simplify the notation. 
The eigenfunction equation \eqref{eq:EFEqf} 
is then equivalent to the recurrence relations 
\begin{equation}
\sum_{\beta+\nu=\mu}b_{\beta}(q^{-\nu};c)f_{\nu}
=\sum_{k\delta+\nu=\mu}\vep_{k}(c)f_\nu\qquad(\mu\in \sQ^{\aff}_+) 
\end{equation}
for the coefficients of $f(x;p)$.  
In view of condition 
\begin{equation}
b_{0}(1;c)f_0=\vep_0(c)f_0,
\end{equation}
of the case $\mu=0$, 
we start with
\begin{equation}
f_0=1,\quad \vep_0(c)=b_0(1;c)=\prod_{j=1}^{n}(1-cs_j). 
\end{equation}
For $\mu>0$ (i.e. $\mu\in\sQ^\aff_+$, $\mu\ne 0$), we have 
\begin{equation}\label{eq:rec1}
(b_0(q^{-\mu};c)-\vep_0(c))f_{\mu}
+\sum_{\beta+\nu=\mu;\,\beta>0} 
b_\beta(q^{-\nu};c)f_\nu
=\sum_{k\delta+\nu=\mu;\,k>0}\vep_k(c)f_{\nu}. 
\end{equation}
Note here that, if $\mu=k_0\alpha_0+k_1\alpha_1+\cdots+k_{n-1}\alpha_{n-1}\in\sQ^\aff_+$, 
we have
\begin{equation}
\br{\ep_j,\mu}=k_j-k_{j-1}\quad(j=1,\ldots,n),
\end{equation}
with the convention $k_n=k_0$, and hence
\begin{equation}
b_{0}(q^{-\mu};c)=\prod_{j=1}^{n}(1-cs_jq^{-k_j+k_{j-1}}). 
\end{equation}

If $\mu\notin\bZ_{>0}\delta$, there exists an index $i\in\pr{1,\ldots,n}$ 
such that $k_i-k_{i-1}\ne 0$, 
and hence
\begin{equation}
\prod_{j=1}^{n}(1-cs_jq^{-k_j+k_{j-1}})\ne \prod_{j=1}^{n}(1-cs_j)
\end{equation}
by the genericity assumption of $c$. 
Since $b_0(q^{-\mu};c)\ne \vep_0(c)$, 
by the recurrence relation \eqref{eq:rec1}
$f_{\mu}$ is determined uniquely
from $f_{\nu}$ ($\nu<\mu$) and $\vep_k(c)$ with $\htaf(k\delta)=kn<\htaf(\mu)$. 

When $\mu=l\delta$ ($l>0$), we have 
$b_{0}(q^{-l\delta};u)=b_0(1;c)=\vep_0(c)$.
Hence the recurrence relation is written as
\begin{equation}\label{eq:rec2}
\sum_{\beta+\nu=l\delta;\,\beta>0} 
b_\beta(q^{-\nu};c)f_\nu
=\vep_l(c)+\sum_{k\delta+\nu=l\delta;\,0<k<l}\vep_k(c)f_{\nu}. 
\end{equation}
This implies that 
$\vep_{l}(c)$ is determined uniquely from $f_{\nu}$ 
($\nu<l\delta$) and $\vep_k(c)$ ($k<l$), and $f_{l\delta}$ can be fixed arbitrarily.
Since this ambiguity of $f_{k\delta}$ ($k>0$) corresponds to the multiplication 
by an element $\gamma(p)\in\bC\fps{p}$ with $\gamma(0)=1$, we adopt here the 
normalization
\begin{equation}
f_0=1,\quad f_{k\delta}=0\quad\mbox{for all}\quad k>0. 
\end{equation}
Under this normalization, the recurrence relation for $\vep_l(u)$ $(l>0)$ simplifies as
\begin{equation}\label{eq:rec3}
\vep_l(c)=\sum_{\beta+\nu=l\delta;\,\beta>0} 
b_\beta(q^{-\nu};c)f_\nu. 
\end{equation}
This procedure implies the unique existence of a normalized formal solution 
$f(x;p)$.  Also, one can verify that any eigenfunction is 
expressed as $\gamma(p)f(x;p)$ with some $\gamma(p)\in\bC\fps{p}$, 
by a similar argument to the one we used in Subsection \ref{ssec:2.3}. 
\par\medskip
We now prove that $f(x;p)$ is a joint eigenfunction for all 
$\cE_{x,s}^{(r)}(p)$ ($r=0,1,\ldots,n$).  
Since $\cE_{x,s}^{(r)}(p)$ commutes with $\cE_{x,s}(p;c)$, we have
\begin{equation}
\cE_{x,s}(p;c)\cE_{x,s}^{(r)}(p)f(x;p)=\cE_{x,s}^{(r)}(p)\cE_{x,s}(p;c)f(x;p)
=\vep(p;c)\cE_{x,s}^{(r)}(p)f(x;p). 
\end{equation}
This means that $g_r(x;p)=\cE_{x,s}^{(r)}(p)f(x;p)$ is a formal eigenfunction 
of $\cE_{x,s}(p;c)$, and hence 
$g_r(x;p)=\vep^{(r)}(p)f(x;p)$ for some $\vep^{(r)}(p)\in\bC\fps{p}$,
namely, 
$\cE_{x,s}^{(r)}(p)f(x;p)=\vep^{(r)}(p)f(x;p)$ for $r=0,1,\ldots,n$. 
This also implies that the normalized formal solution $f(x;p)$ 
satisfies
\begin{equation}
\cE_{x,s}(p;u)f(x;p)=\vep(p;u)f(x;p),\quad \vep(p;u)=\sum_{r=0}^{n}(-u)^r \vep^{(r)}(p)
\end{equation}
for arbitrary $u\in\bC$. 
In particular, $f(x;p)$ does not depend on the choice of $c\in\bC^\ast\backslash S$ 
that we started with.  
This completes the proof of Theorem \ref{thm:RBFm}. 

\begin{rem}\rm 
If we regard as $q$ and $t$ as indeterminates, 
the coefficients $f_{\beta}$ $(\beta\in\sQ^\aff_+$) of 
the normalized formal solution $f(x;p)$ are determined uniquely 
as rational functions in $(q,t)$. 
\end{rem}

\section{Convergence of perturbative eigenfunctions}\label{sec:Convergence}

In this section, we first prove the convergence of 
the formal joint eigenfunctions $f(x;s;p)$ of 
the modified elliptic Ruijsenaars operators 
in the asymptotic domain $|x_1|\gg\cdots\gg|x_n|\gg|px_1|$, 
and then apply an integral transform to them 
for establishing the convergence of 
symmetric formal joint eigenfunctions $\cP_\lambda(x;p)$ $(\lambda\in \sP_+)$ 
of the elliptic Ruijsenaars operators around of the torus $\bT^n$.  
We also include a proof of Theorem \ref{thm:RAOr} on the orthogonality 
of elliptic deformations $\cP_\lambda(x;p)$ of Macdonald polynomials. 

\subsection{Ruijsenaars functions in the asymptotic domain 
$|x_1|\gg\cdots\gg|x_n|\gg|px_1|$}
In what follows, we fix the parameters $s=(s_1,\ldots,s_n)\in(\bC^\ast)^n$, assuming that 
they satisfy the genericity condition $s_j/s_i\notin q^{\bZ}$ ($1\le i<j\le n$).  
Keeping the notation of Subsection \ref{ssec:2.4}, 
we denote by 
$f(x;p)$ of the normalized formal solution of the eigenfunction equation
\begin{equation}
\cE_{x,s}(p;u)f(x;p)=\vep(p;u)f(x;p), \quad f(x;p)\in\bC\fps{px_1/x_n,x_2/x_1,\ldots,x_n/x_{n-1}}
\end{equation}
in the asymptotic domain $|x_1|\gg\cdots\gg|x_n|\gg|px_1|$. 
(We write $f(x;p)=f(x;s;p)$ and $\vep(p;u)=\vep(s;p;u)$ when we need 
to make the dependence on $s$ manifest.) 
As before, we expand the modified elliptic Ruijsenaars operator in the form
\begin{equation}
\cE_{x,s}(p;u)=\sum_{I\subseteq\pr{1,\ldots,n}}
(-u)^{|I|}s^{\ep_I}\cB_{I}(x;p)T_{q,x}^{\ep_I}=
\sum_{\beta\in\sQ^\aff_+}x^{-\beta}b_{\beta}(T_{q,x};u), 
\end{equation}
where $b_{\beta}(\xi;u)$ are the symbols of $q$-difference operators defined by 
\eqref{eq:defbfun}. 
Then the formal power series 
\begin{equation}
f(x;p)=\sum_{\mu\in\sQ^\aff_+}x^{-\mu}f_{\mu},\quad 
\vep(p;u)=\sum_{k=0}^{\infty} p^k\vep_k(u)
\end{equation}
satisfy the initial condition 
\begin{equation}
f_0=1,\quad \vep_0(u)=b_0(1;u)=\prod_{j=1}^{n}(1-us_i)
\end{equation}
and 
the recurrence formulas 
\begin{equation}\label{eq:rec4}
(b_0(q^{-\mu};u)-\vep_0(u))f_{\mu}
=-\sum_{\beta+\nu=\mu;\,\beta>0} 
b_\beta(q^{-\nu};u)f_\nu
+\sum_{k\delta+\nu=\mu;\,k>0}\vep_k(u)f_{\nu}
\quad(\mu\in \sQ^{\aff}_+\backslash\bZ\delta), 
\end{equation}
and 
\begin{equation}\label{eq:rec5}
f_{l\delta}=0,\quad \vep_l(u)=\sum_{\beta+\nu=l\delta;\,\beta>0} 
b_\beta(q^{-\nu};u)f_\nu\qquad(l=1,2,\ldots). 
\end{equation}
Our goal is to establish estimates for the coefficients $f_{\mu}$ and $\vep_k(u)$ 
satisfying these recurrence formulas, so as to guarantee the absolute convergence of 
$f(x;p)$ and $\vep(p;u)$. 
As far as the condition $s_j/s_i\notin q^{\bZ}$ ($1\le i<j\le n$) is satisfied, 
our estimates for $f_{\mu}$ will be uniform on every compact subset 
with respect to the parameters $(s;q,t)$. 

\par\medskip
Setting $U=\prm{q\in\bC^\ast}{\ |q|<1}$, we consider the following subset 
$D$ of  the space of parameters $(s;q)=(s_1,\ldots,s_n;q)\in(\bC^\ast)^n\times U$: 
\begin{equation}
D=\prm{ (s;q)\in(\bC^\ast)^n\times U}{\ 
1-q^l s_j/s_i\ne 0 \ \ \mbox{for any distinct $i,j\in\pr{1,\ldots,n}$ and 
$l\in\bZ$}}. 
\end{equation}
Note that $D$ is a locally finite intersection of 
complements of hypersurfaces, and hence an open subset of 
$(\bC^\ast)^n\times U$. 

\begin{lem}\label{lem:tauconst} For any compact subset $K\subseteq D$, there exist 
real constants $\tau_0\in(0,1]$ and $\tau_1>0$ such that
\begin{equation}\label{eq:tau01}
\tau_0\le |s_j/s_i|\le\tau_0^{-1},
\quad |1-q^l s_j/s_i|\ge \tau_1\quad(l\in\bZ_{\ge 0})
\end{equation}
for any distinct $i,j\in\pr{1,\ldots,n}$. 
\end{lem}
\begin{proof}{Proof}
We first take $\tau_0\in(0,1]$ and $\rho\in(0,1)$ so that 
$\tau_0\le |s_j/s_i|\le \tau_0^{-1}$ $(i,j\in\pr{1,\ldots,n};\ i\ne j)$ 
and $|q|\le \rho$ for all $(s;q)\in K$. 
Let $l_0\in\bZ_{\ge0}$ be sufficient large so that $\rho^{l_0}\le \tau_0$. 
Then, for $l>l_0$, we have 
$|q^ls_j/s_i|\le \rho^{l_0+1}\tau_0^{-1}\le \rho$ 
hence $|1-q^{l}s_j/s_i|\ge 1-\rho>0$ on $K$, 
for any distinct $i,j$.  
The condition of Lemma is fulfilled if we let $\tau_1>0$ 
sufficiently small so that $\tau_1\le 1-\rho$ and 
$|1-q^{l}s_j/s_i|\ge \tau_1$ $(0\le l\le l_0)$ on $K$, 
for any distinct $i,j$. 
\end{proof}

From now on, we assume that the parameters $(s;q,t)$ belong 
to an arbitrary compact subset of $\in D\times\bC^\ast$, 
and fix two real constants $\tau_0\in (0,1]$ and $\tau_1>0$ 
as in Lemma \ref{lem:tauconst}.  

\par\medskip
As we already remarked in Subsection 2.2, 
the coefficients $\cB_I(x;p)$ of $\cE_{x,s}(p;u)$ are holomorphic 
functions in the variables $z=(z_0,z_1,\ldots,z_{n-1})$, defined by 
\begin{equation}
z_0=px_1/x_n,\ z_1=x_2/x_1,\ \ldots,\ z_{n-1}=x_n/x_{n-1}, 
\end{equation}
in the polydisc $|z_0|<1$, $|z_1|<1$, $\ldots$, $|z_{n-1}|<1$. 
Hence, their Taylor expansions 
\begin{equation}
\cB_I(x;p)=\sum_{k_0,k_1,\ldots,k_{n-1}\ge 0}z_0^{k_0}z_1^{k_1}\ldots z_{n-1}^{k_{n-1}}\,
b^{I}_{k_0,k_1,\ldots,k_{n-1}} \in \bC\fps{z_0,z_1,\ldots,z_{n-1}}
\end{equation}
satisfy the Cauchy estimate:  
For each $r\in(0,1)$, there exists a constant $M_r>0$ such that
\begin{equation}
|b^I_{k_0,k_1,\ldots,k_{n-1}}|\le \frac{M_r}{r^{k_0+k_1+\cdots+k_{n-1}}}
\qquad(I\subseteq\pr{1,\ldots,n}; k_0,k_1,\ldots,k_{n-1}\ge 0). 
\end{equation}
By the identification $\beta=k_0\alpha_0+k_1\alpha_1+\cdots+k_{n-1}\alpha_{n-1}$, 
the Cauchy estimate can be rewritten as 
\begin{equation}\label{eq:Cauchy0}
|b^I_{\beta}|\le \frac{M_r}{r^{\htaf(\beta)}}
\qquad(I\subseteq\pr{1,\ldots,n}; \beta\in \sQ^{\aff}_+),  
\end{equation}
where $\htaf(\beta)=k_0+k_1+\cdots+k_{n-1}$ denotes the {\em height} of $\beta\in \sQ^{\aff}_+$. 

This Cauchy estimate \eqref{eq:Cauchy0} implies the estimate 
\begin{equation}\label{eq:Cauchy1}
|b_{\beta}(\xi;u)|\le 
\frac{M_r}{r^{\htaf(\beta)}}
\sum_{I\subseteq\pr{1,\ldots,n}}|u||s^{\ep_I}||\xi^{\ep_I}|
=
\frac{M_r}{r^{\htaf(\beta)}}
\prod_{j=1}^{n}(1+|u||s_j||\xi_j|)
\end{equation}
for the coefficients 
$b_{\beta}(\xi;u)$ $(\beta\in \sQ^{\aff}_+)$ of \eqref{eq:defbfun} 
appearing in the recurrence relations. 
If $\nu=l_0\alpha_0+l_1\alpha_1+\cdots+l_{n-1}\alpha_{n-1}\in \sQ_+^{\aff}$,
we have
\begin{equation}
|b_{\beta}(q^{-\nu};u)|\le \frac{M_r}{r^{\htaf(\beta)}}
\prod_{j=1}^{n}(1+|u||s_j||q|^{-l_j+l_{j-1}}),
\end{equation}
with $l_n=l_0$ as before. 
With the notation $|x|_+=\max\pr{x,0}$, we obtain 
\begin{equation}
\begin{split}
\prod_{j=1}^{n}(1+|u||s_j||q|^{-l_j+l_{j-1}})
&\le
\prod_{j=1}^{n}(1+|u||s_j||q|^{-|l_j-l_{j-1}|_+})
\\
&\le |q|^{-\sum_{j=1}^{n}|l_j-l_{j-1}|_{+}}
\prod_{i=1}^{n}(1+|u||s_j|). 
\end{split}
\end{equation}
In view of this estimate, we define
\begin{equation}
\dd(\nu)=\sum_{j=1}^{n}|l_j-l_{j-1}|_{+}=
\sum_{1\le j\le n;\ l_j>l_{j-1}}(l_j-l_{j-1})\in\bZ_{\ge 0}. 
\end{equation}
Note also that $\sum_{j=1}^{n}|l_j-l_{j-1}|=2\dd(\nu)$,  since $\sum_{j=1}^{n}(l_j-l_{j-1})=0$. 
Then we have 
\begin{lem}\label{lem:A}
For any $\beta\in \sQ_{+}^{\aff}$ with $\beta>0$, 
\begin{equation}
|b_{\beta}(q^{-\nu};u)|\le 
\frac{M_r}{r^{\htaf(\beta)}|q|^{\dd(\nu)}}
\prod_{j=1}^{n}(1+|u||s_j|)\qquad(\nu\in \sQ_{+}^\aff).
\vspace{-12pt}
\end{equation}
\qed
\end{lem}

In estimating the leading factor in the recurrence relation \eqref{eq:rec4},
\begin{equation}
b_0(q^{-\mu};u)-\vep_0(u)=\prod_{j=1}^{n}(1-us_jq^{-k_{j}+k_{j-1}})-\prod_{j=1}^{n}(1-us_j) 
\end{equation}
for $\mu\in \sQ^{\aff}_{+}\backslash\bZ\delta$ with $\mu=\sum_{j=0}^{n-1}k_j\alpha_j$, 
we choose an index $i\in\pr{1,\ldots,n}$ such that $\br{\ep_i,\mu}=
k_i-k_{i-1}\ne 0$, 
and make the substitution $u=s_i^{-1}$ to obtain
\begin{equation}\label{eq:busp}
b_0(q^{-\mu};s_i^{-1})-\vep_0(s_i^{-1})
=b_0(q^{-\mu};s_i^{-1})=(1-q^{-k_i+k_{i-1}})\prod_{j\ne i}(1-q^{-k_j+k_{j-1}}s_j/s_i).
\end{equation}

\begin{lem}\label{lem:B}
With the real constants $\tau_0, \tau_1$ as in \eqref{eq:tau01}, 
we have 
\begin{equation}
|b_0(q^{-\mu};s_i^{-1})|\ge \tau|q|^{-d(\mu)},\quad 
\tau=|1-q|(\tau_0\tau_1)^{n-1}, 
\end{equation}
for any $\mu\in\sQ^\aff_+\backslash\bZ\delta$, and 
for any $i\in\pr{1,\ldots,n}$ such that $\br{\ep_i,\mu}\ne 0$. 
\end{lem}
\begin{proof}{Proof}
For $l\in\bZ$, we have 
\begin{equation}
|1-q^{-l}|\ge |1-q|\quad(l<0),\qquad
|1-q^{-l}|\ge |q|^{-l}|1-q| \quad(l>0),
\end{equation}
and hence 
\begin{equation}
|1-q^{-l}|\ge |q|^{-|l|_+}|1-q|\quad(l\in\bZ,\ l\ne 0).  
\end{equation}
Similarly, for any distinct $i,j\in\pr{1,\ldots,n}$ we have 
\begin{equation}
\begin{split}
&|1-q^{-l} s_j/s_i|\ge\tau_1\quad(l\le 0),\quad 
\\
&|1-q^{-l}s_j/s_i|=|q|^{-l}|s_j/s_i||1-q^ls_i/s_j|
\ge |q|^{-l}\tau_0\tau_1\quad(l\ge 0), 
\end{split}
\end{equation}
and hence 
\begin{equation}
|1-q^{-l}s_j/s_i|\ge |q|^{-|l|_+}\tau_0\tau_1\quad(l\in\bZ). 
\end{equation}
These inequalities imply 
\begin{equation}
\begin{split}
|b_0(q^{-\mu};s_i^{-1})|&\ge |q|^{-|k_i-k_{i-1}|_+}|1-q|
\prod_{j\ne i}\big(|q|^{-|k_j-k_{j-1}|}\tau_0\tau_1\big)
\\
&\ge |q|^{-\dd(\mu)}|1-q|(\tau_0\tau_1)^{n-1}
\end{split}
\end{equation}
for any $i\in\pr{1,\ldots,n}$ such that $\br{\ep_i,\mu}=k_i-k_{i-1}\ne 0$. 
\end{proof}
\par\medskip
We give a proof of convergence of $f(x;p)$ on the basis of Lemma \ref{lem:A} 
and Lemma \ref{lem:B}.  
In view of the recurrence relation, for each $\mu\in \sQ^{\aff}_+$ we set
\begin{equation}
g_{\mu}(u)=\sum_{\beta+\nu=\mu;\,\beta>0} 
b_\beta(q^{-\nu};u)f_\nu,
\end{equation}
so that $\vep_{l}(u)=g_{l\delta}(u)$ ($l=1,2,\ldots$)
and 
\begin{equation}
(b_0(q^{-\mu};u)-\vep_0(u))f_{\mu}
=-
g_{\mu}(u)
+\sum_{k\delta+\nu=\mu;\,k>0}g_{k\delta}(u)f_{\nu}. 
\quad 
\end{equation}
We apply Lemma \ref{lem:A} to the estimate of $g_{\mu}(s_i^{-1})$ ($i=1,\ldots,n$). 
By Lemma \ref{lem:A} we have
\begin{equation}
|b_\beta(q^{-\nu};s_i^{-1})|\le 
\frac{M_r}{r^{\htaf(\beta)}|q|^{\dd(\nu)}}
\prod_{j=1}^{n}(1+|s_j/s_i|)\quad(i=1,\ldots,n)
\end{equation}
for any $\beta>0$ and $\nu\in \sQ^{\aff}_+$.  
With the constant $\tau_0$ of \eqref{eq:tau01}, we have 
$\prod_{j=1}^{n}(1+|s_j/s_i|)\le 2(1+\tau_0^{-1})^{n-1}$,  
and hence
\begin{equation}
|b_\beta(q^{-\nu};s_i^{-1})|\le 
\frac{M_rL}{r^{\htaf(\beta)}|q|^{\dd(\nu)}}\quad(i=1,\ldots,n),
\quad L=2(1+\tau_0^{-1})^{n-1}. 
\end{equation}
We apply this estimate to $g_{\mu}(u)$ to obtain 
\begin{equation}
|g_{\mu}(s_i^{-1})|\le \sum_{\beta+\nu=\mu;\,\beta>0} 
\frac{M_rL}{r^{\htaf(\beta)}|q|^{\dd(\nu)}}|f_{\nu}|, 
\end{equation}
or equivalently, 
\begin{equation}\label{eq:recA}
r^{\htaf(\mu)}|g_{\mu}(s_i^{-1})|\le M_rL\sum_{\beta+\nu=\mu;\,\beta>0} 
\frac{r^{\htaf(\nu)}}{|q|^{\dd(\nu)}}|f_{\nu}|\quad(i=1,\ldots,n).
\end{equation}
Next we apply Lemmas \ref{lem:A} and \ref{lem:B} to the estimate of 
$f_{\mu}$. 
As to the recurrence relation
for $\mu>0$ with $\mu\notin\bZ_{>0}\delta$, 
we choose an
index $i\in\pr{1,\ldots,n}$ with $\br{\ep_i,\mu}\ne 0$,  
and specialize it by the substitution $u=s_i^{-1}$ to obtain 
\begin{equation}
b_0(q^{-\mu};s_i^{-1})f_{\mu}
=-g_{\mu}(s_i^{-1})
+\sum_{k\delta+\nu=\mu;\,k>0}g_{k\delta}(s_i^{-1})f_{\nu}.
\end{equation}
Then, by Lemma \ref{lem:A} and Lemma \ref{lem:B}, 
we have the estimate 
\begin{equation}
\tau |q|^{-\dd(\mu)}|f_{\mu}|
\le 
|g_{\mu}(s_i^{-1})|
+\sum_{k\delta+\nu=\mu;\,k>0}|g_{k\delta}(s_i^{-1})||f_{\nu}|,
\end{equation}
which we rewrite as 
\begin{equation}\label{eq:recB}
\tau \frac{r^{\htaf(\mu)}}{|q|^{\dd(\mu)}}|f_{\mu}|
\le 
r^{\htaf(\mu)}|g_{\mu}(s_i^{-1})|
+\sum_{k\delta+\nu=\mu;\,k>0}r^{\htaf(k\delta)}|g_{k\delta}(s_i^{-1})|r^{\htaf(\nu)}|f_{\nu}|. 
\end{equation}

\par\medskip
Having the two estimates \eqref{eq:recA} and \eqref{eq:recB}, we propose to 
define two families of positive constants $\ca_{\mu}$ $(\mu>0)$
and $\cb_{\mu}$\ \ ($\mu\ge 0$) (majorants)
such that 
\begin{equation}\label{eq:abcoversgf}
r^{\htaf(\mu)}|g_{\mu}(s_i^{-1})|\le \ca_{\mu}\quad(\mu>0;\ i=1,\ldots,n),
\qquad
\frac{r^{\htaf(\mu)}}{|q|^{\dd(\mu)}}|f_{\mu}|\le \cb_{\mu}\quad(\mu\ge 0).
\end{equation}
We specify these constants by the recurrence relations
\begin{equation}\label{eq:recab}
\ca_{\mu}=\cs \sum_{\substack{\beta+\nu=\mu\\ \beta>0}}\cb_{\nu},\qquad
\cb_{\mu}=\ct \sum_{\beta+\nu=\mu;\,\beta>0}\ca_{\beta}\cb_{\nu}\qquad(\mu>0)
\end{equation}
with the initial condition $\cb_0=1$ so that $|f_0|=\cb_0$, 
where $\cs=M_rL$ and $\ct=\tau^{-1}$. 
We recursively verify that these constants fulfill the estimates of \eqref{eq:abcoversgf}. 
In fact, by \eqref{eq:recA} we have
\begin{equation}
\begin{split}
r^{\htaf(\mu)}
|g_{\mu}(s_i^{-1})|&\le \cs \sum_{\beta+\nu=\mu;\,\beta>0} \cb_{\nu}=\ca_{\mu}.
\end{split}
\end{equation}
for $\mu>0$, 
and also, by \eqref{eq:recB} we obtain 
\begin{equation}
\begin{split}
\frac{r^{\htaf(\mu)}}{|q|^{\dd(\mu)}}|f_{\mu}|
&\le 
\ct\Big(
\ca_{\mu}
+\sum_{k\delta+\nu=\mu;\,k>0}\ca_{k\delta}|q|^{\dd(\nu)}\cb_{\nu}
\Big)
\le
\ct\sum_{\beta+\nu=\mu;\,\beta>0}\ca_{\beta}\cb_{\nu}=\cb_{\mu}
\end{split}
\end{equation}
for $\mu>0$ with $\mu\notin\bZ\delta$. 

\par\medskip
Finally we establish Cauchy 
estimates for $\ca_{\mu}$ $(\mu>0)$ and $\cb_{\mu}$ $(\mu\ge 0)$ 
by means of the generating functions
\begin{equation}
\ca(z)=\sum_{\mu>0}\ca_{\mu} z_0^{k_0}z_1^{k_1}\cdots z_{n-1}^{k_{n-1}},
\qquad
\cb(z)=\sum_{\mu\ge 0}\cb_{\mu} z_0^{k_0}z_1^{k_1}\cdots z_{n-1}^{k_{n-1}},
\end{equation}
in $n$ variables $z=(z_0,z_1,\ldots,z_{n-1})$, 
with the parametrization $\mu=\sum_{j=0}^{n-1}k_j\alpha_j\in \sQ^{\aff}_+$ as before.  
Note that these formal power series satisfy $\ca(0)=0$ and $\cb(0)=1$.  
The recurrence relations for $\ca_{\mu}$ and $\cb_{\mu}$ are expressed as 
the algebraic equations 
\begin{equation}\label{eq:algrel}
\ca(z)=\cs\, \cc(z)\cb(z),\quad \cb(z)-1=\ct\,\ca(z)\cb(z)
\end{equation}
for $\ca(z)$, $\cb(z)$, where
\begin{equation}
\cc(z)=\prod_{i=0}^{n-1}(1-z_i)^{-1}-1.  
\end{equation}
Eliminating $\cb(z)$ from \eqref{eq:algrel}, we obtain a quadratic equation 
\begin{equation}
\ct\,\ca(z)^2-\ca(z)+\cs\,\cc(z)=0
\end{equation}
for $\ca(z)$, and hence, by $\ca(0)=0$ we obtain
\begin{equation}
\ca(z)=\frac{1}{2\,\ct}\Big(1-\sqrt{1-4\,\cs \ct\,\cc(z)}\Big). 
\end{equation}
Similarly, eliminating $\ca(z)$ from \eqref{eq:algrel}, we obtain a quadratic equation 
\begin{equation}
\cs\ct\, \cc(z)\cb(z)^2-\cb(z)+1=0
\end{equation}
for $\cb(z)$, and hence, by $\cb(0)=1$ we obtain
\begin{equation}
\cb(z)=\frac{2}{1+\sqrt{1-4\,\cs \ct\,\cc(z)}}. 
\end{equation}
Note that $\ct=\tau^{-1}$ and $\cs\ct=M_rL\tau^{-1}$. 
Since $\cc(z)$ is a holomorphic function in the polydisc 
$\cU=\pr{\,|z_i|<1\ (i=0,\ldots,n-1)}$ with $\cc(0)=0$, 
the formal power series $\ca(z)$ and $\cb(z)$ in $z$ represent holomorphic functions  
on the domain $\prm{z\in \cU}{\,|\cc(z)|<1/4\,\cs\ct}$ containing 0.  
By the Cauchy estimate, 
this implies that there exist positive constants $\rho>0$ and $K>0$,
depending on $r$ and $\tau_0$, $\tau_1$, 
such that 
\begin{equation}
\ca_{\mu}\le \frac{K}{{\rho}^{\htaf(\mu)}},\quad
\cb_{\mu}\le \frac{K}{{\rho}^{\htaf(\mu)}}
\quad(\mu>0).  
\end{equation}

Returning to \eqref{eq:abcoversgf}, we obtain
\begin{equation}\label{eq:estimategf}
|g_{\mu}(s_i^{-1})|\le \frac{K}{(r \rho)^{\htaf(\mu)}},\quad
|f_{\mu}|\le |q|^{\dd(\mu)}
\frac{K}{(r \rho)^{\htaf(\mu)}}
\le \frac{K}{(r \rho)^{\htaf(\mu)}}
\end{equation}
for all $\mu>0$.  
This means that the formal power series 
\begin{equation}
f(x;p)=\sum_{\mu\in \sQ^{\aff}_+}x^{-\mu}f_{\mu}
=\sum_{\mu\in \sQ^{\aff}_+} z_0^{k_0}z_1^{k_1}\cdots z_{n-1}^{k_{n-1}}f_{\mu}
\end{equation}
in $z=(z_0,z_1,\ldots,z_{n-1})$ converges in the domain 
$|z_i|<\sigma=r\rho$ ($i=0,1,\ldots,n-1$).
Also, 
noting that $\vep_k(u)=g_{k\delta}(u)$ ($k=1,2,\ldots$), for each $i=1,\ldots,n$ 
we have 
\begin{equation}
|\vep_k(s_i^{-1})|\le \frac{K}{\sigma^{\htaf(k\delta)}}
=\frac{K}{(\sigma^{n})^k}\quad(k=1,2,\ldots). 
\end{equation}
Hence $\vep(p;s_i^{-1})=\sum_{k=0}^{\infty}p^k\vep_k(s_i^{-1})$ 
converges absolutely in the disc $|p|<\sigma^n$. 
Since  $\vep(p;s_i^{-1})=\sum_{j=0}^{n}(-s_i^{-1})^{j}\vep^{(j)}(p)$, we have
\begin{equation}
1-s_i^{-1}\vep^{(1)}(p)+s_i^{-2} \vep^{(2)}(p)+\cdots+(-s_i^{-1})^{n}\vep^{(n)}(p)=\vep(p;s_i^{-1})
\quad(i=1,\ldots,n). 
\end{equation}
Since $s_i\ne s_j$ ($1\le i<j\le n$),  
solving this system of linear equations for $\vep^{(i)}(p)$ ($j=1,\ldots,n$), 
we see that the formal eigenvalues 
$\vep^{(j)}(p)$ ($j=1,\ldots,n$) are also convergent in the same disc $|p|<\sigma^n$. 
This completes the proof of Theorem \ref{thm:RBCv}. 

\par\medskip
When we make the dependence 
on $s=(s_1,\ldots,s_n)$ and $q,t$ explicit, we write $f(x;s;p|q,t)$ for $f(x;p)$, 
under the condition that $s_j/s_i\notin q^{\bZ}$ $(1\le i<j\le n)$. 
We remark again that the convergence established above 
is uniform on every compact subsets 
with respect to $(s;q,t)\in D\times \bC^\ast$. 
Hence, $f(x;s;p|q,t)$ depends homomorphically on 
the parameters $(s;q,t)\in D\times\bC^\ast$. 

We refer to this normalized joint eigenfunction as 
the {\em $($stationary$)$ Ruijsenaars function} 
in the asymptotic domain $|x_1|\gg\cdots\gg|x_n|\gg|px_1|$. 

\begin{rem}\rm
Suppose that $s_i\ne s_j$ for any distinct $i,j\in\pr{1,\ldots,s}$.  
Then the assumption 
$s_j/s_i\notin q^{\bZ}$ ($1\le i<j\le n$) of Theorem \ref{thm:RBCv} is fulfilled 
automatically for sufficiently small $q$. 
Let $\tau_0\in(0,1]$ satisfy $\tau_0\le|s_j/s_i|\le\tau_0^{-1}$ 
$(1\le i<j\le n)$, and for any $\theta$ with $0<\theta<\tau_0$ set 
\begin{equation}
\tau_1=\min\pr{1-\theta/\tau_0, |1-s_j/s_i|\ \ (i,j\in\pr{1,\ldots,n};\ i\ne j)}. 
\end{equation}
Then, for any $q\in\bC^\ast$ with 
$|q|\le\theta$, we have $|1-q^{l}s_j/s_i|\ge\tau_1$ for all 
distinct $i,j\in\pr{1,\ldots,n}$ and $l\in\bZ_{\ge 0}$.  
Hence, the constants $K$ and $\rho$ in \eqref{eq:estimategf} 
do not depend on $q$ with $|q|\le\theta$.  
This implies that $f(x;s;p)=f(x;s;p|q,t)$ is also a {\em holomorphic} function  
on $q\in\bC$ with $|q|<\tau_0$ and $t\in\bC^\ast$, 
including $q=0$.  Also, we have 
$f(x;s;p|0,t)=1$ by \eqref{eq:estimategf}.  
\end{rem}

\subsection{Elliptic deformation of Macdonald polynomials} 

We now proceed to the problem of  
symmetric joint eigenfunctions around the torus $\bT^n$.  
We apply an idea of integral transforms, similar to that we already 
presented in Section \ref{sec:Prototype},  
to construct elliptic deformations of Macdonald polynomials from 
Ruijsenaars functions in the asymptotic domain $|y_1|\gg\cdots\gg|y_n|\gg|py_1|$. 

\par\medskip
Recall that the {\em Ruijsenaars elliptic gamma function} is defined by 
\begin{equation}
\Gamma(z;p,q)=\frac{(pq/z;p,q)_\infty}{(z;p,q)_\infty},\quad
(z;p,q)_\infty=\prod_{i,j=0}^{\infty}(1-p^iq^j z).  
\end{equation}
Note that $\Gamma(z;p,q)$ is a meromorphic function on $\bC^\ast$, and satisfies the 
functional equations
\begin{equation}
\Gamma(qz;p,q)=\theta(z;p)\Gamma(z;p,q),\ \ 
\Gamma(pz;p,q)=\theta(z;q)\Gamma(z;p,q),\ \ 
\Gamma(pq/z;p,q)=\Gamma(z;p,q)^{-1}.  
\end{equation}
In order to formulate an integral transform in the elliptic case, we use the following 
elliptic deformation of $K(x,y)$ and $w(y)$:
\begin{equation}
K(x,y;p)=\prod_{i=1}^{n}\prod_{j=1}^{n}\frac{\Gamma(x_i/y_j;p,q)}{\Gamma(tx_i/y_j;p,q)},
\quad
w(y;p)=\prod_{1\le i<j\le n}\theta(y_j/y_i;p)\frac{\Gamma(ty_j/y_i;p,q)}{\Gamma(qy_j/ty_i;p,q)}. 
\end{equation}
The weight function $w(y;p)$ is holomorphic in the domain 
\begin{equation}
|pt|<|y_j/y_i|<|t|^{-1}\qquad(1\le i<j\le n).
\end{equation}
Also, the kernel function $K(x,y;p)$ is holomorphic in the domain
\begin{equation}
|pq/t|<|x_i/y_j|<1\qquad(i,j=1,\ldots,n).  
\end{equation}

Supposing that $0<\theta\le\min\pr{1,1/|t|, |t/q|}$ and $|p|<\theta^n$, we define the domain $\cV_{\theta}\subset (\bC^\ast)^n$ by 
\begin{equation}
\cV_{\theta}=\prm{y=(y_1,\ldots,y_n)\in(\bC^\ast)^n\,}
{\ |y_2/y_1|<\theta,\ \ldots,\ |y_n/y_{n-1}|<\theta,\ 
|py_1/y_n|<\theta}.  
\end{equation}
For any $r>0$ and $\sigma>0$ such that $|p|\le\sigma^n<\theta^n$, we take the $n$-cycle 
\begin{equation}
C_{r,\sigma}=\prm{y=(y_1,\ldots,y_n)\in(\bC^\ast)^n}
{\ |y_i|=\sigma^{i-1}r\ \ (i=1,\ldots,n)}\subset(\bC^\ast)^n 
\end{equation}
as in the trigonometric case.    
Then $C_{r,\sigma}$ is contained in $\cV_{\theta}$.  
Also, since $|pt|\le\sigma^n|t|<\sigma^{n-1}$, $w(y;p)$ is holomorphic in an neighborhood 
of $C_{r,\sigma}$.  
On the other hand, if $y\in C_{r,\sigma}$, $K(x,y;p)$ is holomorphic in the domain
\begin{equation}
\cU_{r,\sigma}=\prm{ x=(x_1,\ldots,x_n)\in(\bC^\ast)^n}
{\ |pq/t|r<|x_i|<\sigma^{n-1}r\ \ (i=1,\ldots,n)}. 
\end{equation}
Since $|pq/t|<\sigma^{n-1}$, $\cU_{r,\sigma}$ is not empty. 
Note also that, if $x\in \cU_{r,\sigma}$, then $|pq/t|\sigma^{-n+1}<|x_j/x_i|<|t/pq|\sigma^{n-1}$. 
In view of this fact, we set
\begin{equation}
\cU_{\sigma}=\prm{x=(x_1,\ldots,x_n)\in(\bC^\ast)^n}
{\ |pq/t|\sigma^{-n+1}<|x_j/x_i|<|t/pq|\sigma^{n-1}
\ \ (1\le i<j\le n)},
\end{equation}
so that $\cU_{r,\sigma}\subseteq \cU_{\sigma}$ for any $r>0$. 
\begin{lem}\label{lem:IT-ell}
Let $\theta$ be a positive constant such that $\theta\le\min\pr{1,1/|t|,|t/q|}$ 
and suppose that $|p|<\theta^n$.
For each holomorphic function $\psi(y;p)$ in the domain $\cV_{\theta}$, 
consider the holomorphic function $\varphi(x;p)$ 
in the domain $\cU_{r,\sigma}$ defined by the integral transform
\begin{equation}
\varphi(x;p)=\int_{C_{r,\sigma}}K(x,y;p)w(y;p)\psi(y;p)d\omega_n(y),
\end{equation}
taking $r>0$ and $\sigma>0$ such that $\sigma<\theta$ and $|p|\le \sigma^n$. 
Then $\varphi(x;p)$ is continued to a symmetric holomorphic function in the domain $\cU_{\theta}$, 
which does not depend on the choice of $r$, $\sigma$. 
\end{lem}
\begin{proof}{Proof}
We verify that
\begin{equation}
\cU_{\sigma}=\bigcup_{r>0}\ \cU_{r,\sigma}.  
\end{equation}
For $x\in\cU_{\sigma}$ given, take two indices $k,l\in\pr{1,\ldots,n}$ 
such that $|x_k|=\max\pr{|x_1|,\ldots,|x_n|}$ and $|x_l|=\min\pr{|x_1|,\ldots,|x_n|}$.  
Then we have $|pq/t|\sigma^{-n+1}<|x_l/x_k|\le 1$. 
Since $|x_k|/\sigma^{n-1}<|x_l|/|pq/t|$, we can find an $r>0$ such that 
$|x_k|/\sigma^{n-1}<r<|x_l|/|pq/t|$.  
Then we have  $|pq/t|r<|x_l|\le|x_k|<\sigma^{n-1}r$.  This means that 
$|pq/t|r<|x_i|<\sigma^{n-1}r$ $(i=1,\ldots,n)$, namely $x\in \cU_{r,\sigma}$.
Since the integral transforms defined in $\cU_{r,\sigma}$ 
do not depend on $r>0$, they define a holomorphic function $\varphi(x;p)$ 
in $\cU_\sigma$.  
Since $\sigma>0$ can be taken arbitrarily as far as $|p|^{\frac{1}{n}}\le \sigma<\theta$, 
$\varphi(x;p)$ is continued to $\cU_{\theta}$. 
\end{proof}

We remark that a natural elliptic counterpart of the symmetric weight function $w^{\mathrm{sym}}(y)$ of
\eqref{eq:wsym-trig} is given by
\begin{equation}
w^{\mathrm{sym}}(y;p)=\prod_{1\le i<j\le n}\frac{\Gamma(ty_i/y_j;p,q)\Gamma(ty_j/y_i;p,q)}
{\Gamma(y_i/y_j;p,q)\Gamma(y_j/y_i;p,q)}. 
\end{equation}
Similarly to the trigonometric case, the elliptic Ruijsenaars operator $\cD_{x}(p;u)$ is formally self-adjoint with respect to this weight function $w^{\mathrm{sym}}(y;p)$
in the following sense. 
(For the definition of the formal adjoint of a $q$-difference operator, 
see the explanation below formula \eqref{eq:DP=eP}.) 
\begin{lem}\label{lem:selfadjoint}
\begin{equation}
\cD_{y}(p;u)^\ast=w^{\mathrm{sym}}(y;p)\cD_{y^{-1}}(p;u)w^{\mathrm{sym}}(y;p)^{-1}. 
\end{equation}
\end{lem}
\begin{proof}{Proof}
The elliptic Ruijsenaars operator of $r$th order is expressed as 
\begin{equation}
\cD^{(r)}_y(p)=\sum_{\substack{I\subseteq\pr{1,\ldots,n}\\ |I|=r|}}
\cA_I(y;p)T_{q,y}^{\ep_I},
\quad
\cA_I(y;p)=t^{\binom{r}{2}}\prod_{i\in I;\,j\notin I}\frac{\theta(ty_i/y_j;p)}{\theta(y_i/y_j;p)}. 
\end{equation}
We decompose $w^{\mathrm{sym}}(y;p)$ as
\begin{equation}
w^{\mathrm{sym}}(y;p)=w^{\mathrm{sym}}_{+}(y;p)w^{\mathrm{sym}}_{+}(y^{-1};p),
\quad
w^{\mathrm{sym}}_{+}(y;p)=\prod_{1\le i<j\le n}\frac{\Gamma(ty_i/y_j;p,q)}{\Gamma(y_i/y_j;p,q))}, 
\end{equation}
and look at $I=\pr{1,\ldots,r}$: 
\begin{equation}\label{eq:topcoeff}
\frac{T_{q,y}^{\ep_{\pr{1,\ldots,r}}}(w^{\mathrm{sym}}_{+}(y;p))}{w^{\mathrm{sym}}_{+}(y;p)}
=\prod_{i=1}^{r}\prod_{j=r+1}^{n}
\frac{\theta(ty_i/y_j;p)}{\theta(y_i/y_j;p)}
=\prod_{i\in\pr{1,\ldots,r},\,j\notin\pr{1,\ldots,r}}
\frac{\theta(ty_i/y_j;p)}{\theta(y_i/y_j;p)}.  
\end{equation}
For each $I\subseteq\pr{1,\ldots,n}$ with $|I|=r$, choose a permutation 
$\sigma_I\in\frS_n$ 
such that $\sigma_I(\pr{1,\ldots,r})=I$.   
Then, applying $\sigma_I$ to \eqref{eq:topcoeff} we obtain
\begin{equation}
\frac{T_{q,y}^{\ep_{I}}(\sigma_Iw^{\mathrm{sym}}_{+}(y;p))}
{\sigma_Iw^{\mathrm{sym}}_{+}(y;p)}
=\prod_{i\in I,\,j\notin I}
\frac{\theta(ty_i/y_j;p)}{\theta(y_i/y_j;p)}, 
\end{equation}
namely,
\begin{equation}
\cA_I(y;p)=t^{\binom{r}{2}}
\frac{T_{q,y}^{\ep_{I}}(\sigma_I.w^{\mathrm{sym}}_{+}(y;p))}
{\sigma_I.w^{\mathrm{sym}}_{+}(y;p)}
=t^{\binom{r}{2}}\frac{1}{\sigma_I.w^{\mathrm{sym}}_{+}(y;p)}\,T_{q,y}^{\ep_I} \,\sigma_I.w^{\mathrm{sym}}_{+}(y;p).  
\end{equation}
Hence $\cD^{(r)}_{y}(p)$ and its formal adjoint $\cD^{(r)}_{y}(p)^{\ast}$ 
are expressed as 
\begin{equation}
\cD_y^{(r)}(p)
=t^{\binom{r}{2}}
\sum_{\substack{I\subseteq\pr{1,\ldots,n}\\ |I|=r}}
\frac{1}{\sigma_I.w^{\mathrm{sym}}_{+}(y;p)}\,
T_{q,y}^{\ep_I} \,\sigma_I.w^{\mathrm{sym}}_{+}(y;p), 
\end{equation}
and
\begin{equation}
\cD_y^{(r)}(p)^{\ast}
=t^{\binom{r}{2}}
\sum_{\substack{I\subseteq\pr{1,\ldots,n}\\ |I|=r}}
\sigma_I.w^{\mathrm{sym}}_{+}(y;p)\,T_{q,y}^{-\ep_I} \,
\frac{1}{\sigma_I.w^{\mathrm{sym}}_{+}(y;p)}. 
\end{equation}
Note that $\frS_n$-invariance of  $w^{\mathrm{sym}}(y;p)$ implies
\begin{equation}
w^{\mathrm{sym}}(y;p)=\sigma_I.w^{\mathrm{sym}}_{+}(y;p)\,\sigma_I.w^{\mathrm{sym}}_+(y^{-1};p).  
\end{equation}
Hence we have
\begin{equation}
\begin{split}
&w^{\mathrm{sym}}(y;p)^{-1}
\cD_y^{(r)}(p)^{\ast}\,w^{\mathrm{sym}}(y;p)
\\
&=t^{\binom{r}{2}}
\sum_{\substack{I\subseteq\pr{1,\ldots,n}\\ |I|=r}}
\frac{\sigma_I.w^{\mathrm{sym}}_{+}(y;p)}
{w^{\mathrm{sym}}(y;p)}
\,T_{q,y}^{-\ep_I}
 \,
\frac{w^{\mathrm{sym}}(y;p)}
{\sigma_I.w^{\mathrm{sym}}_{+}(y;p)} 
\\
&=t^{\binom{r}{2}}
\sum_{\substack{I\subseteq\pr{1,\ldots,n}\\ |I|=r}}
\frac{1}{\sigma_I. w^{\mathrm{sym}}_{+}(y^{-1};p)}
T_{q,y}^{-\ep_I} \,\sigma_I.w^{\mathrm{sym}}_{+}(y^{-1};p)
=\cD^{(r)}_{y^{-1}}(p), 
\end{split}
\end{equation}
as desired. 
\end{proof}

Also, by a result of \cite{R2009b} and \cite{KNS2009}, 
the kernel function $K(x,y;p)$ satisfies the kernel function identity
\begin{equation}
\cD_{x}(p;u)K(x,y;p)=\cD_{y^{-1}}(p;u)K(x,y; p).  
\end{equation}
Hence we have 
\begin{equation}
\begin{split}
\cD_{x}(p;u)K(x,y;p)w^{\mathrm{sym}}(y;p)
&= w^{\mathrm{sym}}(y;p)\cD_{y^{-1}}(p;u)K(x, y; p)
\\
&=\cD_{y}(p;u)^{\ast}(K(x, y; p)w^{\mathrm{sym}}(y;p)).
\end{split}
\end{equation}
We rewrite the symmetric weight function $w^{\mathrm{sym}}(y;p)$ as
\begin{equation}
\begin{split}
w^{\mathrm{sym}}(y;p)
&=\prod_{1\le i<j\le n}
\frac
{\Gamma(pty_i/y_j;p,q)\Gamma(ty_j/y_i;p,q)}
{\Gamma(py_i/y_j;p,q)\Gamma(y_j/y_i;p,q)}
\frac{\theta(y_i/y_j;q)}{\theta(ty_i/y_j;q)}
\\
&=
\prod_{1\le i<j\le n}
\frac{\Gamma(qy_j/y_i;p,q)\Gamma(ty_j/y_i;p,q)}
{\Gamma(qy_j/ty_i;p,q)\Gamma(y_j/y_i;p,q)}
\cdot
\prod_{1\le i<j\le n}
\frac{\theta(y_i/y_j;q)}{\theta(ty_i/y_j;q)}
\\
&=
\prod_{1\le i<j\le n}
\theta(y_j/y_i;p)
\frac{\Gamma(ty_j/y_i;p,q)}{\Gamma(qy_j/ty_i;p,q)}
\cdot
\prod_{1\le i<j\le n}
\frac{\theta(y_i/y_j;q)}{\theta(ty_i/y_j;q)}
\end{split}
\end{equation}
to obtain
\begin{equation}
w^{\mathrm{sym}}(y;p)=w(y;p)g(y),\quad g(y)=\prod_{1\le i<j\le n}
\frac{\theta(y_i/y_j;q)}{\theta(ty_i/y_j;q)}, 
\end{equation}
with the same function $g(y)$ as in the trigonometric case.  
Then, by the same argument as in Subsection \ref{ssec:3.2}, we obtain 
the kernel function identity
\begin{equation}
\cD_{x}(p;u)K(x,y;p)w(y;p) = \cE_{y,t^{\rho^\vee}}(p;u)^{\ast}(K(x, y; p)w(y; p))
\end{equation}
in the elliptic case, where $t^{\rho^\vee}=(1,t,\ldots,t^{n-1})$ stands for the reversal 
of $t^{\rho}=(t^{n-1},t^{n-2},\ldots,1)$. 

\begin{thm}\label{thm:ellMacIT}
Suppose that $t^k\notin q^{\bZ}$ $(k=1,\ldots,n-1)$.   
Let $\lambda$ be a partition with $l(\lambda)\le n$, namely $\lambda\in \sP_+$, 
$\br{\ep_n,\lambda}\ge 0$, and set $s=t^{\rho}q^{\lambda}$.  
Let $\theta$ be a positive constant such that 
$\theta\le \min\pr{1,1/|t|,|t/q|}$ and that the Ruijsenaars function 
$f(y;s^\vee;p)$ is holomorphic in the domain 
\begin{equation}
\cV_{\theta}:\quad |y_2/y_1|<\theta,\ \ \ldots,\ \ |y_n/y_{n-1}|<\theta,\ \ |py_1/y_n|<\theta
\end{equation}
for $|p|<\theta^n$, as well as the eigenvalue $\vep(s^\vee;p;u)$. 
Then the integral transform 
\begin{equation}
\varphi_\lambda(x;p)=\int_{C_{r,\sigma}} K(x,y;p)w(y;p)\,y^{\lambda^\vee}\!f(y;s^\vee;p)d\omega_n(y)
\end{equation}
is continued to a symmetric holomorphic function in the domain
\begin{equation}
\cU_{\theta}:\quad
|pq/t|\theta^{-n+1}<|x_j/x_i|<|t/pq|\theta^{n-1}\qquad(1\le i<j\le n). 
\end{equation}
If $|p|<|q|^n\theta^n$, then $\varphi_\lambda(x;p)$ 
satisfies the joint eigenfunction equation
\begin{equation}\label{eq:EFEq-eqn}
\cD_{x}(p;u)\varphi_\lambda(x;p)=\vep(s;p;u)\varphi_\lambda(x;p)
\end{equation}
in $\cU_{|q|\theta}\subset \cU_{\theta}$ 
with the initial condition $\varphi_\lambda(x;0)=b_\lambda P_\lambda(x)$.  
Furthermore, as power series in $p$ we have 
\begin{equation}
\varphi_\lambda(x;p)=b_{\lambda}(p)\cP_{\lambda}(x;p)
\end{equation}
with a holomorphic function $b_\lambda(p)$ in $|p|<\theta^n$. 
\end{thm}
\begin{proof}{Proof}
The first half of Theorem \ref{thm:ellMacIT} is a consequence of Lemma \ref{lem:IT-ell}.  
Supposing that $|p|<|q|^n\theta^n$, 
we show that $\varphi_\lambda(x;p)$ satisfies the joint eigenfunction equation 
\eqref{eq:EFEq-eqn} in the domain $\cU_{|q|\theta}$.   
Let $\sigma$ be a positive constant such that $|p|^{\frac{1}{n}}\le \sigma<|q|\theta$, 
and suppose that $x=(x_1,\ldots,x_n)\in(\bC^\ast)^n$ satisfies 
\begin{equation}
|p/t|\sigma^{-n+1}<|x_j/x_i|<|t/p|\sigma^{n-1}\quad(1\le i<j\le n). 
\end{equation}
Take two indices $k,l\in\pr{1,\ldots,n}$ such that $|x_l|\le |x_i|\le |x_k|$ for $i=1,\ldots,n$. 
Then we have $|p/t|\sigma^{-n+1}<|x_l/x_k|\le 1$.  
Since $\sigma^{-n+1}|x_k|<|t/p||x_l|$, we can take $r>0$ such that 
$\sigma^{-n+1}|x_k|<r<|t/p||x_l|$ so that $|p/t|r<|x_l|\le|x_k|<\sigma^{n-1}r$, namely
$|p/t|r<|x_i|<\sigma^{n-1}r$ $(i=1,\ldots,n)$.   
Then, for $y\in C_{r,\sigma}$ we have
$|p/t||y_j|<|x_i|<|y_j|$, namely 
\begin{equation}
|pq/t|<|qx_i/y_j|<|x_i/y_j|<1\qquad(i,j=1,\ldots,n). 
\end{equation}
Then we can represent $\cD_x(p;u)\varphi_\lambda(x;p)$ as an integral over the 
cycle $C_{r,\sigma}$: 
\begin{equation}
\begin{split}
\cD_x(p;u)\varphi_\lambda(x;p)
&=
\int_{C_{r,\sigma}}\cD_x(p;u)K(x,y;p)w(y;p)\,y^{\lambda^\vee}\!f(y;s^{\vee};p)d\omega_n(y)
\\
&=
\int_{C_{r,\sigma}}(\cE_{y,t^{\rho^\vee}})^\ast\big(K(x,y;p)w(y;p)\big)
\,y^{\lambda^\vee}\!f(y;s^{\vee};p)d\omega_n(y). 
\end{split}
\end{equation}
Recalling that 
\begin{equation}
\cE_{y,t^{\rho^\vee}}(p;u)=\sum_{I\subseteq\pr{1,\ldots,n}}
(-u)^{|I|}t^{\br{\ep_I,\rho^\vee}}\cB_{I}(y;p)T_{q,y}^{\ep_I}, 
\end{equation}
for each $I\subseteq\pr{1,\ldots,n}$, 
we compute the integral 
\begin{equation}
\begin{split}
&
\int_{C_{r,\sigma}}T_{q,y}^{-\ep_I}\left(\cB_{I}(y;p)K(x,y;p)w(y;p)\right)\,
y^{\lambda^\vee}\!f(y;s^\vee;p)d\omega_n(y)
\\
&=
\int_{C_{r,\sigma}}\cB_{I}(q^{-\ep_I}y;p)K(x,q^{-\ep_I}y;p)w(q^{-\ep_I}y;p)\,
y^{\lambda^\vee}\!f(y;s^\vee;p)d\omega_n(y)
\\
&=
\int_{q^{-\ep_I}C_{r,\sigma}}\cB_{I}(y;p)K(x,y;p)w(y;p)\,
(q^{\ep_I}y)^{\lambda^\vee}\!f(q^{\ep_I}y,s^\vee;p)d\omega_n(y). 
\end{split}
\end{equation}
Since the two cycles $q^{-\ep_I}C_{r,\sigma}$ and $C_{r,\sigma}$ are 
homologous in the domain of holomorphy of the integrand, 
by the Cauchy theorem we see that this integral equals to 
\begin{equation}
\begin{split}
&
\int_{C_{r,\sigma}}K(x,y;p)w(y;p)\,
\cB_{I}(y;p)(q^{\ep_I}y)^{\lambda^\vee}\!f(q^{\ep_I}y,s^\vee;p)d\omega_n(y)
\\
&=
\int_{C_{r,\sigma}}K(x,y;p)w(y;p)\,
\cB_{I}(y;p)T_{q,y}^{\ep_I}
\big(y^{\lambda^\vee}\! f(y;s^\vee;p)\big)d\omega_n(y). 
\end{split}
\end{equation}
Note also that the poles $\cB_I(y;p)$ are cancelled by the 
factor $\prod_{1\le i<j\le n}\theta(y_j/y_i;p)$ in $w(y;p)$.
Hence we have 
\begin{equation}
\begin{split}
\cD_x(x;p)\varphi_\lambda(x;p)
&=
\int_{C_{r,\sigma}} K(x,y;p)w(y;p)\cE_{y,t^{\rho^\vee}}(p;u)
\big(y^{\lambda^\vee}\!f(y;s^\vee;p)\big)d\omega_n(y)
\\
&=
\int_{C_{r,\sigma}} K(x,y;p)w(y;p)\,(y^{\lambda^\vee}
\cE_{y,t^{\rho^\vee}q^{\lambda^\vee}}(p;u)\big(f(y;s^\vee;p)\big)d\omega_n(y)
\\
&=
\int_{C_{r,\sigma}} K(x,y;p)w(y;p)\,y^{\lambda^\vee}
\vep(s^\vee;p;u)f(y;s^\vee;p)d\omega_n(y)
\\
&=
\vep(s^\vee;p;u)\varphi_{\lambda}(x;p). 
\end{split}
\end{equation}
This means that 
\begin{equation}
\cD_x(p;u)\varphi_\lambda(x;p)=\vep(s^\vee;p;u)\varphi_\lambda(x;p),
\quad s=t^{\rho}q^{\lambda}. 
\end{equation}
Also, by comparing the integral transform with the one in the trigonometric case, 
from Theorem \ref{thm:MacIT} we obtain 
\begin{equation}
\varphi_{\lambda}(x;0)=b_\lambda P_{\lambda}(x),\quad \vep(s^\vee;0;u)=\vep(s;u)=\prod_{i=1}^{n}(1-us_i).  
\end{equation}
Note also that 
\begin{equation}
K(x,y;p)=K(x,y)
\prod_{i=1}^{n}\prod_{j=1}^{n}
\frac{(ptx_i/y_j;p,q)_\infty(pqy_j/x_i;p,q)_\infty}
{(px_i/y_j;p,q)_\infty(pqy_j/tx_i;p,q)_\infty}. 
\end{equation}
From this expression of $K(x,y;p)$, we see that 
$\varphi_\lambda(x;p)$ has $p$-expansion of the form
\begin{equation}
\varphi_\lambda(x;p)=\sum_{k=0}^{\infty}p^k\varphi_{\lambda,k}(x),
\quad \varphi_{\lambda,k}(x)\in(x_1\cdots x_n)^{-k}\bC\fps{x}^{\frS_n}.  
\end{equation}
Since the joint eigenfunction
$\varphi_\lambda(x)$ has the leading term $\varphi_{\lambda,0}(x)=b_\lambda P_\lambda(x)
\in\bC[x^{\pm1}]^{\frS_n}$, 
by the recurrence \eqref{eq:reck1} we inductively see $\varphi_{\lambda,k}(x)\in\bC[x^{\pm1}]^{\frS_n}$
for all $k=0,1,2,\ldots$, 
which means that $\varphi_\lambda(x;p)\in\bC[x^{\pm1}]^{\frS_n}\fps{p}$.  
Hence, 
we conclude that the Laurent expansion 
of $\varphi_\lambda(x;p)$ around the torus $\bT^n$ 
coincides with the normalized formal solution $\cP_{\lambda}(x;p)$
up to multiplication by a power series in $\bC\fps{p}$, 
and that $\vep(s^\vee;p;u)=\vep(s;p;u)$.  
We set 
\begin{equation}
\varphi_{\lambda}(x;p)=b_\lambda(p)\cP_{\lambda}(x;p),\quad
b_\lambda(p)\in\bC\fps{p}.  
\end{equation}
In fact, $b_\lambda(p)$ is determined as the constant term 
of $x^{-\lambda}\varphi_{\lambda}(x;p)$ with respect to $x$:
\begin{equation}
b_\lambda(p)=\int_{\bT^n}x^{-\lambda}\varphi_\lambda(x;p)d\omega_n(x).  
\end{equation}
Since $b_\lambda(p)$ is a holomorphic function in $|p|<\theta^n$, 
the normalized formal solution $\cP_\lambda(x;p)$ is absolutely convergent 
as well in an neighborhood $\bT^n$ where $\varphi_\lambda(x;p)$ is holomorphic, 
if $|p|$ is sufficiently small.  
\end{proof}

This implies Theorem \ref{thm:RACv} for the case where 
$\lambda\in\sP_+$ is a partition.  
In fact, from Theorem \ref{thm:ellMacIT} it follows that 
$\varphi_\lambda(x;p)=b_\lambda(p)\cP_{\lambda}(x;p)$ 
is holomorphic in the domain 
\begin{equation}\label{eq:domptau}
|p|/\tau<|x_j/x_i|<\tau/|p|\quad(1\le i<j\le n),\quad|p|<\tau
\end{equation}
for any positive $\tau\le |t/q|\theta^{n-1}$. 
If we take a sufficiently small $\tau>0$ such that $b_\lambda(p)\ne 0$ 
for $|p|<\tau$, $\cP_\lambda(x;p)$ is holomorphic in the domain 
\eqref{eq:domptau}. 
If $\lambda\in \sP_{+}$ is a general dominant integral weight, there exists an integer 
$l\in\bZ$ such that $\lambda+(l^n)$ is a partition. 
Since
\begin{equation}
\cP_{\lambda}(x; p) = (x_1\cdots x_n)^{-l}\cP_{\lambda+(l^n)}(x; p), 
\end{equation}
the statement of Theorem \ref{thm:RACv} 
is reduced to the case of a partition. 

\subsection{Orthogonality relation for $\cP_{\lambda}(x;p)$}

For each dominant integral weight $\lambda\in\sP_+$, 
we use the notation $\cP_{\lambda}(x;p)=\cP_\lambda(x;p|q,t)$ 
for the normalized symmetric joint eigenfunction of the elliptic Ruijsenaars 
operator $\cD_x(p;u)$ with initial condition $\cP_\lambda(x;0)=P_\lambda(x)$; 
it is holomorphic in a domain of the form
\begin{equation}\label{eq:aroundT}
\cW_{\tau}=\prm{(x;p)\in (\bC^\ast)^n\times\bC^\ast}{\ 
|p|/\tau<|x_j/x_i|<\tau/|p|\quad(1\le i<j\le n),\quad 0<|p|<\tau}
\end{equation}
for some $\tau\in(0,1]$.  
Note that $\cW_\tau$ is not empty for any $\tau\in(0,1]$, 
and $\cW_{\tau'}\subseteq\cW_{\tau}$ if $\tau'\le\tau$. 
In this subsection, we assume that $|t|<1$ and
$t^k\notin q^{\bZ_{>0}}$ for $k=1,\ldots,n-1$ in 
order to guarantee the existence of $\cP_{\lambda}(x;p)$.  
Note that the eigenfunction equation 
\begin{equation}
\cD_x(p;u)\cP_\lambda(x;p)=\vep_\lambda(p;u)\cP_\lambda(x;p) 
\end{equation}
holds in the domain $\cW_{|q|\tau}$.  

\par\medskip
For a pair of holomorphic functions $\varphi(x;p), \psi(x;p)$ 
in a domain $\cW_{\tau}$ of the form \eqref{eq:aroundT}, 
we define the scalar product 
$\br{\varphi(x;p),\psi(x;p)}$ by the integral\,\footnote{
The left-hand side in \eqref{eq:defscalarprod} would perhaps be better 
written as $\br{\varphi(\cdot;p),\psi(\cdot;p)}$, 
but we find it convenient to abuse the notation here. 
}
\begin{equation}\label{eq:defscalarprod}
\br{\varphi(x;p),\psi(x;p)}=\int_{\bT^n} 
\varphi(x^{-1};p)\psi(x;p)w^{\mathrm{sym}}(x;p)d\omega_n(x)
\end{equation}
with the symmetric weight function 
\begin{equation}
w^{\mathrm{sym}}(x;p)=\prod_{1\le i<j\le n}\frac{\Gamma(tx_i/x_j;p,q)\Gamma(tx_j/x_i;p,q)}
{\Gamma(x_i/x_j;p,q)\Gamma(x_j/x_i;p,q)}.
\end{equation}
Note here that
\begin{equation}
\frac{1}{\Gamma(x_i/x_j;p,q)\Gamma(x_j/x_i;p;q)}=\theta(x_i/x_j;p)\theta(x_j/x_i;q), 
\end{equation}
and hence, $w^{\mathrm{sym}}(x;p)$ is holomorphic in the domain
\begin{equation}
|t|<|x_j/x_i|<|t|^{-1}\quad(1\le i<j\le n),\quad |p|<1.  
\end{equation}

By Lemma \ref{lem:selfadjoint}, we know that the elliptic Ruijsenaars 
operator $\cD_x(p;u)$ is formally self-adjoint with respect to 
$w^{\mathrm{sym}}(x;p)$.  
\begin{lem} For any pair of holomorphic functions $\varphi(x;p)$, 
$\psi(x;p)$ in $\cW_\tau$, we have the following identity 
with respect to the scalar product for holomorphic functions in $\cW_{|q|\tau}$\,$:$ 
\begin{equation}
\br{\cD_{x}(p;u)\varphi(x;p),\psi(x;p)}=
\br{\varphi(x;p),\cD_{x}(p;u)\psi(x;p)}. 
\end{equation}
\end{lem}
\begin{proofa}{Proof}
In view of the fact that the scalar products depend holomorphically on $q$, 
we assume that  $|t|<|q|<1$. 
Note first that 
\begin{equation}
\begin{split}
\cD_{x^{-1}}(p;u)\varphi(x^{-1};p)
&=w^{\mathrm{sym}}(x;p)(\cD_{x^{-1}}(p;u)\varphi(x^{-1};p))
\\
&=\cD_{x}(p;u)^\ast(w^{\mathrm{sym}}(x;p)\varphi(x^{-1};p)). 
\end{split}
\end{equation}
Since
\begin{equation}
\cD_x(p;u)^{\ast}=\sum_{I\subseteq\pr{1,\ldots,n}}
(-u)^{|I|}T_{q,x}^{-\ep_I}\cA_I(x;p),
\end{equation}
we have
\begin{equation}
\begin{split}
&\br{\cD_x(p;u)\varphi(x;p),\psi(x;p)}
\\
&=\int_{\bT}
(\cD_{x^{-1}}(p;u)\varphi(x^{-1};p))\psi(x;p)d\omega_n(x)
\\
&=\int_{\bT}\cD_x(p;u)^{\ast}(w^{\mathrm{sym}}(x;p)\varphi(x^{-1};p))\psi(x;p)d\omega_n(x)
\\
&=
\sum_{I\subseteq\pr{1,\ldots,n}}(-u)^{|I|}
\int_{\bT^n}
\cA_I(q^{-\ep_I}x;p)w^{\mathrm{sym}}(q^{-\ep_I}x;p)\varphi(q^{\ep_I}x^{-1};p)
\psi(x;p)d\omega_n(x)
\\
&=\sum_{I\subseteq\pr{1,\ldots,n}}(-u)^{|I|}
\int_{q^{-\ep_I}\bT^n}
\cA_I(x;p)w^{\mathrm{sym}}(x;p)
\varphi(x^{-1};p)\psi(q^{\ep_I}x;p)d\omega_n(x). 
\end{split}
\end{equation}
Note here that, for each $I\subseteq\pr{1,\ldots,n}$, 
the poles of $\cA_I(q^{-\ep_I})$ are cancelled 
by the zeros of $w^{\mathrm{sym}}(q^{-\ep_I};p)$.  
Hence, under the condition $|t|<|q|<1$, the $n$-cycle $q^{-\ep_I}\bT$ 
is homologous to $\bT$ in the domain of holomorphy of the integrand. 
Applying Cauchy's theorem, we obtain 
\begin{equation}
\begin{split}
&\br{\cD_x(p;u)\varphi(x;p),\psi(x;p)}
\\
&=
\sum_{I\subseteq\pr{1,\ldots,n}}(-u)^{|I|}
\int_{\bT^n}
\varphi(x^{-1};p)
\cA_I(x;p)
\psi(q^{\ep_I}x;p)
w^{\mathrm{sym}}(x;p)
d\omega_n(x)
\\
&=
\int_{\bT^n}
\varphi(x^{-1};p)
\cD_x(p;u)(
\psi(q^{\ep_I}x;p))
w^{\mathrm{sym}}(x;p)
d\omega_n(x)
\\
&=\br{\varphi(x;p),\cD_{x}(p;u)\psi(x;p)}. 
\end{split}
\end{equation}
\pick{\put(460,16){$\square$}}
\end{proofa}

Let $\lambda,\mu\in\sP_+$ be two dominant integral weights, 
and suppose that $\lambda\ne\mu$. 
Then we have 
\begin{equation}
\br{\cD_x(p;u)\cP_\lambda(x;p),\cP_\mu(x;p)}
=
\br{\cP_\lambda(x;p),\cD_x(p;u)\cP_\mu(x;p)},
\end{equation}
and hence
\begin{equation}
\vep_\lambda(p;u)\br{\cP_\lambda(x;p),\cP_\mu(x;p)}
=
\vep_\mu(p;u)\br{\cP_\lambda(x;p),\cP_\mu(x;p)}.  
\end{equation}
Since $\vep_\lambda(p;u)\ne\vep_\mu(p;u)$, 
we obtain $\br{\cP_\lambda(x;p),\cP_\mu(x;p)}=0$.   
This completes the proof of Theorem \ref{thm:RAOr}.  

\section{Symmetries of the Ruijsenaars function}\label{sec:Symmetry}

In this section, we give proofs of 
Theorems \ref{thm:RBRo} and \ref{thm:RBRe} 
on symmetries of the normalized Ruijsenaars function $f(x;s;p|q,t)$ 
in the asymptotic domain $|x_1|\gg\cdots\gg|x_n|\gg|px_1|$.

\subsection{Symmetry of $f(x;s;p|q,t)$ 
with respect the rotation of variables}
In this subsection, we give a proof of Theorem \ref{thm:RBRo}. 
\par\medskip
Recall that the modified Ruijsenaars operator $\cE_{x,s}(p;u)$ is defined by 
\begin{equation}
\cE_{x,s}(p;u)=\sum_{I\subseteq\pr{1,\ldots,n}}(-u)^{|I|}s^{\ep_I}\cB_{I}(x;p)T_{q,x}^{\ep_I},
\end{equation}
where
\begin{equation}
\cB_{I}(x;p)=\prod_{1\le i<j\le n}\frac{\theta(t^{-\br{\ep_I,\ep_i-\ep_j}}x_j/x_i;p)}{\theta(x_j/x_i;p)}.  
\end{equation}
We denote by $\pi$ the automorphism of the affine root lattice $\sQ^{\aff}$ such that
\begin{equation}
\pi(\alpha_i)=\alpha_{i+1}\quad(i=1,\ldots,n-1),\quad \pi(\alpha_{n-1})=\alpha_0,
\end{equation}
and by the same symbol the corresponding transformation of the $x$ and $s$ variables:
\begin{equation}
\begin{split}
&\pi(x_i)=x_{i+1}\quad(i=1,\ldots,n-1),\quad \pi(x_n)=px_1,\\
&\pi(s_i)=s_{i+1}\quad(i=1,\ldots,n-1),\quad \pi(s_n)=s_1. 
\end{split}
\end{equation}
Then we have
\begin{equation}
\begin{split}
\pi(\cB_{I}(x;p))
&=
\prod_{2\le i<j\le n}
\frac{\theta(t^{-\br{\ep_I,\ep_{i-1}-\ep_{j-1}}}x_j/x_i;p)}{\theta(x_j/x_i;p)}
\cdot
\prod_{j=2}^{n}\frac{\theta(t^{-\br{\ep_I,\ep_{j-1}-\ep_n}}px_1/x_j;p)}{\theta(px_1/x_j;p)}
\\
&=
\prod_{2\le i<j\le n}
\frac{\theta(t^{-\br{\ep_{\pi(I)},\ep_{i}-\ep_{j}}}x_j/x_i;p)}{\theta(x_j/x_i;p)}
\cdot
\prod_{j=2}^{n}\frac{\theta(t^{-\br{\ep_\pi(I),\ep_{1}-\ep_j}}x_j/x_1;p)}{\theta(x_j/x_1;p)}
\\
&=\cB_{\pi(I)}(x;p)
\end{split}
\end{equation}
with the action of $\pi$ on the indexing set $\pr{1,\ldots,n}$ 
specified by $\pi(i)=i+1$ ($i=1,\ldots,n-1$) and $\pi(n)=1$. 
Hence we obtain
\begin{equation}
\begin{split}
\pi(\cE_{x,s}(p;u))
&=\sum_{I\subseteq\pr{1,\ldots,n}}(-u)^{|I|}s^{\ep_{\pi(I)}}
\pi(\cB_{I}(x;p))T_{q,x}^{\ep_{\pi(I)}}
\\
&=\sum_{I\subseteq\pr{1,\ldots,n}}(-u)^{|I|}s^{\ep_{\pi(I)}}
\cB_{\pi(I)}(x;p)T_{q,x}^{\ep_{\pi(I)}}
\\
&=\sum_{I\subseteq\pr{1,\ldots,n}}(-u)^{|I|}s^{\ep_{I}}
\cB_{I}(x;p)T_{q,x}^{\ep_{I}}
=\cE_{x,s}(x;p).
\end{split}
\end{equation}
We apply $\pi$ to the eigenfunction equation 
\begin{equation}
\cE_{x,s}(p;u)f(x;s;p)=\vep(s;p;u)f(x;s;p)
\end{equation}
to obtain
\begin{equation}
\pi(\cE_{x,s}(p;u))\pi(f(x;s;p))=\pi(\vep(s;p;u))\pi(f(x;s;p))
\end{equation}
and hence
\begin{equation}
\cE_{x,s}(p;u)\pi(f(x;s;p))=\pi(\vep(s;p;u))\pi(f(x;s;p)).
\end{equation}
Recall that the function $f(x;s;p)$ is normalized so that the constant term with respect to $x$ variables 
should be 1.  
Since $\pi$ does not change the constant term of $f(x;s;p)$, it follows 
that $\pi(f(x;s;p))=f(x;s;p)$ and $\pi(\vep(s;p;u))=\vep(s;p;u)$.  
This completes the proof of Theorem \ref{thm:RBRo}.

\subsection{Transformation of $f(x;s;p|q,t)$ 
under the reflection $t \leftrightarrow q/t$}
In this subsection, we give a proof of Theorem \ref{thm:RBRe}. 
\par\medskip
We modify $f(x;s;p)$ by setting
\begin{equation}
\psi(x,s;p)=G(x;p)f(x;s;p),\quad G(x;p)=\prod_{1\le i<j\le n}\frac{\Gamma(tx_j/x_i;p,q)}{\Gamma(x_j/x_i;p,q)},
\end{equation}
and denote by ${}^{G}\cE_{x,s}(p;u)$ the conjugation of $\cE_{x,s}(p;u)$ by $G(x;p)$\,$:$
\begin{equation}
{}^{G}\cE_{x,s}(p;u)=G(x;p)\cE_{x,s}(p;u)G(x;p)^{-1}.  
\end{equation}
For each subset $I\subseteq\pr{1,\ldots,n}$, we have
\begin{equation}
\cB_{I}(x;p)=
\prod_{\substack{1\le i<j\le n \\ i\in I\,j\notin I}}\frac{\theta(x_j/tx_i;p)}{\theta(x_j/x_i;p)}
\cdot 
\prod_{\substack{1\le i<j\le n \\ i\notin I\,j\in I}}\frac{\theta(tx_j/x_i;p)}{\theta(x_j/x_i;p)}
\end{equation}
and also,
\begin{equation}
\begin{split}
\frac{T_{q,x}^{\ep_I}G(x;p)}{G(x;p)}
&=
\prod_{\substack{1\le i<j\le n\\ i\in I,\,j\notin I}}
\frac{\Gamma(tx_j/qx_i;p,q)}{\Gamma(x_j/qx_i;p,q)}
\frac{\Gamma(x_j/x_i;p,q)}{\Gamma(tx_j/x_i;p,q)}
\\
&\quad\cdot
\prod_{\substack{1\le i<j\le n\\ i\notin I,\,j\in I}}
\frac{\Gamma(qtx_j/x_i;p,q)}{\Gamma(qx_j/x_i;p,q)}
\frac{\Gamma(x_j/x_i;p,q)}{\Gamma(tx_j/x_i;p,q)}
\\
&=
\prod_{\substack{1\le i<j\le n\\ i\in I,\,j\notin I}}
\frac{\theta(x_j/qx_i;p)}{\theta(tx_j/qx_i;p)}
\cdot
\prod_{\substack{1\le i<j\le n\\ i\notin I,\,j\in I}}
\frac{\theta(tx_j/x_i;p)}{\theta(x_j/x_i;p)}. 
\end{split}
\end{equation}
Hence we have
\begin{equation}
{}^G\cE_{x,s}(p;u)
=\sum_{I\subseteq\pr{1,\ldots,n}}(-u)^{|I|}s^{\ep_I}
\prod_{\substack{1\le i<j\le n\\ i\in I,\,j\notin I}}
\frac{\theta(x_j/tx_i;p)\theta(tx_j/qx_i;p)}{\theta(x_j/x_i;p)\theta(x_j/qx_i;p)}
T_{q,x}^{\ep_I},
\end{equation}
which is manifestly invariant under the reflection $t\leftrightarrow q/t$ of the parameter $t$. 
Note that the eigenfunction equation
\begin{equation}
\cE_{x,s}(p;u)f(x;s;p)=\vep(s;p;u)f(x;s;p)
\end{equation}
of $f(x;s;p)$ implies
\begin{equation}
{}^G\cE_{x,s}(p;u)\psi(x,s;p)=\vep(s;p;u)\psi(x,s;p). 
\end{equation}
Denoting by $\tau$ the reflection $t\to q/t$ of the parameter $t$, 
we apply $\tau$ to this equation.  
Since ${}^{G}\cE_{x,s}(p;u)$ is invariant under $\tau$, we have
\begin{equation}
{}^G\cE_{x,s}(p;u)\tau(\psi(x,s;p))=\tau(\vep(s;p;u))\tau(\psi(x,s;p)) 
\end{equation}
and hence
\begin{equation}
\cE_{x,s}(p;u)G(x;p)^{-1}\tau(\psi(x,s;p))=\tau(\vep(s;p;u))G(x;p)^{-1}\tau(\psi(x,s;p)). 
\end{equation}
This means that 
\begin{equation}
G(x;p)^{-1}\tau(\psi(x;p))=
\frac{\tau(G(x;p))}{G(x;p)}\tau(f(x;s;p))
\end{equation}
is an eigenfunction of $\cE_{x,s}(p;u)$, and 
it is also holomorphic in the domain $|x_1|\gg\cdots|x_n|\gg|px_1|$ 
and has leading term 1.  
Hence we see that 
\begin{equation}
\frac{\tau(G(x;p))}{G(x;p)}\tau(f(x;s;p))=\gamma(s;p)f(x;s;p)
\end{equation}
for some power series $\gamma(s;p)$ of $p$, and that 
$\tau(\vep(s;p,u))=\vep(s;p;u)$. 
This implies that
\begin{equation}
\begin{split}
\tau(f(x;s;p))&=
\gamma(s;p)\frac{G(x;p)}{\tau(G(x;p))}f(x;s;p)
\\
&=
\gamma(s;p)\prod_{1\le i<j\le n}\frac{\Gamma(tx_j/x_i;p,q)}{\Gamma(qx_j/tx_i;p,q)}
\cdot f(x;s;p), 
\end{split}
\end{equation}
namely
\begin{equation}\label{eq:f=ggf}
f(x;s;p|q,q/t)=
\gamma(s;p|q,t)
\prod_{1\le i<j\le n}\frac{\Gamma(tx_j/x_i;p,q)}{\Gamma(qx_j/tx_i;p,q)}
\cdot
f(x;s;p|q,t). 
\end{equation}
By the normalization condition of $f(x;s;p|q,t)$ and $f(x;s;p|q,q/t)$, 
we see that 
$\gamma(s,p|q,t)$ is determined as the constant term of 
\begin{equation}
f(x;s;p|q,q/t)
\prod_{1\le i<j\le n}\frac{\Gamma(qx_j/tx_i;p,q)}{\Gamma(tx_j/x_i;p,q)}
\end{equation}
with respect to $x$, 
when regarded as a power series in $\bC\fps{px_1/x_n,x_2/x_1,\ldots,x_n/x_{n-1}}$.  
Also, from the symmetry \eqref{eq:f=ggf}, it follows that $\gamma(s;p|q,t)$ satisfies the functional equation
\begin{equation}
\gamma(s;p|q,t)\gamma(s;p|q,q/t)=1.  
\end{equation}
This completes the proof of Theorem \ref{thm:RBRe}. 

\section{Concluding remarks}
We constructed two kinds of eigenfunctions of the commuting family of difference operators defining the elliptic Ruijsenaars system. 
The first kind of eigenfunction provide an elliptic deformation of the Macdonald polynomials \cite{M1995}, and the second kind generalize the asymptotically free eigenfunctions in \cite{NS2012}. 
We proved that these eigenfunctions define analytic functions in suitable domains of the variables and parameters. 
By restricting parameters of variables to the physical domain, the first kind of eigenfunctions have a Hilbert space interpretation which allows to interpret them as quantum mechanical wave functions diagonalizing the Hamiltonian and momentum operator defining the elliptic Ruijsenaars system \cite{R1987}. 
In order to substantiate this interpretation, however, 
it would be desirable 
to refine our results to guarantee the existence of a common domain 
of convergence which does not depend on $\lambda$. 

We recall the hyperbolic Ruijsenaars model provides a quantum mechanical description of a relativistic quantum field theory known as sine-Gordon theory \cite{R1987} and, motivated by this, the construction of eigenfunctions of the hyperbolic Ruijsenaars system was recently completed by Halln\"as and Ruijsenaars \cite{HR2020}; in this model, space is the real line $\bR$ and, for this reason, the spectrum of the Hamiltonian and momentum operator are purely continuous. 
We regard the elliptic Ruijsenaars system as regularization of the hyperbolic one: space $\bR$ is replaced by a circle of circumference $L>0$ and, by this, the continuous spectrum is changed to a pure point spectrum. 
Such a regularization is standard in quantum field theory, and we therefore hope that our results on the elliptic system will be useful in future work to better understand quantum sine-Gordon theory. 

In the present paper, we concentrated on the analytic properties of the eigenfunctions we considered, and we did not attempt to obtain explicit formulas for these functions. 
The later certainly would be very interesting. 
In fact, one of us (S) recently proposed a fully explicit formula for the functions $f(x;s;p|q,t)$ as a limit $\kappa\to 1$ of a function $f^{\widehat{\mathfrak{gl}}_n}(x,p|s,\kappa|q,t)$  depending on a further parameter $\kappa$ \cite{S2019}, 
and we proved that this latter function is analytic in a suitable domain of the variables and parameters \cite{LNS2020a}. 
More specifically, the main conjecture in \cite{S2019} implies that a function\footnote{The functions $f^{\widehat{\mathfrak{gl}}_n}$ and $f_{n,\infty}$ are related to each other by a simple change of variables \cite[Theorem 3.2]{LNS2020a}.} $f_{n,\infty}(x,p|s,\kappa|q,t)$ given by an explicit formula (see \cite[Definition 3.3]{LNS2020a}) is such that the ratio
\begin{equation}\label{falpha0}
\frac{f_{n,\infty}(x,p|s,\kappa|q,t)}{\alpha_{n,\infty}(p|s,\kappa|q,t)}, 
\end{equation} 
with $\alpha_{n,\infty}(p|s,\kappa|q,t)$ defined as 
the constant term of $f_{n,\infty}(x,p|s,\kappa|q,t)$ with respect to $x$, converges to $f(x;s;p|q,t)$ in the limit $\kappa\to 1$. 
If this is true, then 
\begin{equation}\label{falpha}
\frac{f_{n,\infty}(x,p|t^{\rho}q^\lambda,\kappa|q,t)}{\alpha_{n,\infty}(p|t^{\rho}q^\lambda,\kappa|q,t)}\qquad (\rho_i=n-i) 
\end{equation} 
analytically continues to a function converging to $\cP_\lambda(x;p|q,t)$ in the limit $\kappa\to 1$. 

As already mentioned after Theorem \ref{thm:RAOr}, 
one can regard the elliptic deformations of the Macdonald polynomials, $\cP_\lambda(x;p|q,t)$, as an orthogonal system with respect to the scalar product in 
\eqref{eq:cPorth}
defined with a weight functions invariant under the exchange $p\leftrightarrow q$. 
This suggests that these functions $\cP_\lambda(x;p|q,t)$ 
are invariant under this exchange as well.\footnote{We are grateful to S. Ruijsenaars for emphasizing  to us the importance of this symmetry.} 
Another argument supporting this expectation is the fact that the Ruijsenaars operators $D^{(r)}_x(p|q,t)$ in \eqref{eq:defcDrx} 
commute with the operators $D^{(s)}_x(q|p,t)$ obtained from them by this exchange ($r,s=1,\ldots,n$).
However, we have only been able to prove this invariance for 
the ``groundstate'' function:\footnote{Our current proof of this is lengthy, and we therefore do not include it in the present paper.}
\begin{equation} 
\cP_0(x;p|q,t) = \cP_0(x;q|p,t).
\end{equation} 
This, together with other results proved in the present paper,  implies that this 
function has an expansion 
\begin{equation} 
\cP_0(x;p|q,t) = \sum_{k,l=0}^\infty p^kq^l\cP_{0;k,l}(x;t) 
\end{equation} 
where $\cP_{0;k,l}(x;t)=\cP_{0;l,k}(x;t)$ and where all functions $(x_1\cdots x_n)^{\min(k,l)}\cP_{0;k,l}(x;t)$ are symmetric polynomials in the variables $x=(x_1,\ldots,x_n)$. 
It would be interesting to extend this result to all functions $\cP_\lambda(x;p|q,t)$. 
A further reason we believe this can be done is the main conjecture in \cite{S2019} discussed in the previous paragraph, together with the following complementary conjecture:  
the ratio in \eqref{falpha} is invariant under the exchange $p\leftrightarrow q$ for generic values of $\kappa$.

It would also be interesting to extend our results to other related models, including the van Diejen systems \cite{vD1994}.

\section*{Acknowledgment}

This work is supported by VR Grant No.~2016-05167 (E.L.) 
and by JSPS Kakenhi Grants (B) 15H03626 (M.N.), (C) 19K03512 (J.S.). 
M.N. is grateful to the Knut and Alice Wallenberg Foundation for funding 
his guest professorship at KTH.  
We are grateful to the Stiftelse Olle Engkist Byggm\"{a}stare, Contract 184-0573, 
for a travel grant that initiated our collaboration.



\begin{thebibliography}{99}
\bibitem{C2009}
I.~Cherednik:
Whittaker limits of difference spherical functions. 
Int. Math. Res. Not. IMRN 2009, 3793--3842. 
\bibitem{vD1994}
J.F.~van Diejen:
Integrability of difference Calogero-Moser systems. 
J. Math. Phys. {\bf 35} (1994), 2983--3004.
\bibitem{vD1995}
J.F.~van Diejen:
Commuting difference operators with polynomial eigenfunctions. 
Compositio Math. {\bf 95} (1995), 183--233. 
\bibitem{FV1997}
 G.~Felder and A.~Varchenko:
 Elliptic quantum groups and Ruijsenaars models. 
 J. Statist. Phys. {\bf 89} (1997), 963--980.
\bibitem{FV1998}
G.~Felder and A.~Varchenko: 
Algebraic integrability of the two-particle Ruijsenaars operator. 
Funct. Anal. Appl. {\bf 32 } (1998), 81--92. 
\bibitem{FV2004}
G.~Felder and A.~Varchenko: 
Hypergeometric theta functions and elliptic Macdonald polynomials, 
Int. Math. Res. Not. IMRN 2004, 
1037--1055. 
\bibitem{H1997}
K.~Hasegawa:
Ruijsenaars' commuting difference operators as commuting transfer matrices.
Commun. Math. Phys. {\bf 187}, (1997) 289--325.
\bibitem{HR2020} 
M.~Halln\"as and  S. Ruijsenaars:  
Joint eigenfunctions for the relativistic Calogero-Moser Hamiltonians of hyperbolic type. III. Factorized asymptotics.  
Int. Math. Res. Not. IMRN 2020 
(doi:10.1093/imrn/rnaa193,  arXiv:1905.12918). 
\bibitem{Kac1990}
V.G.~Kac:
{\em Infinite Dimensional Lie Algebras}, Third edition. 
Cambridge University Press, 1990. xxii+400 pp.
\bibitem{Kato1995}
T.~Kato: 
{\em Perturbation Theory for Linear Operators}.  
Reprint of the 1980 edition, Classics in Mathematics, 
Springer-Verlag, 1995. 
\bibitem{KNS2009}
Y.~Komori, M.~Noumi and J.~Shiraishi:
Kernel functions for difference operators of Ruijsenaars type and their applications, 
SIGMA {\bf 5} (2009), 054, 40 pages (arXiv:0812.0279).
\bibitem{KT2002}
Y.~Komori and K.~Takemura:
The perturbation of the quantum Calogero-Moser-Sutherland system and related results.
Comm. Math. Phys. {\bf 227} (2002), 93--118. 
\bibitem{L2000}
E.~Langmann:
Anyons and the elliptic Calogero-Sutherland model, 
Lett. Math. Phys. {\bf 54} (2000), 279--289. 
\bibitem{L2013}
E.~Langmann:
Explicit solution of the (quantum) elliptic Calogero-Sutherland model,
Ann. Henri Poincar\'{e} {\bf 15} (2014), 755--791. 
\bibitem{LNS2020a}
E.~Langmann, M.~Noumi and J.~Shiraishi:
Basic properties of non-stationary Ruijsenaars functions. 
SIGMA {\bf 16} (2020), 105, 26 pages (arXiv:2006.07171).
\bibitem{M1995}
I.G.~Macdonald:
{\em Symmetric Functions and Hall Polynomials\/},
Second Edition. 
Oxford Mathematical Monographs, Oxford University Press, 1995, x+475pp. 
\bibitem{M2003}
I.G.~Macdonald:
{\em Affine Hecke Algebras and Orthogonal Polynomials}. 
Cambridge Tracts in Mathematics {\bf 157}. 
Cambridge University Press, 2003. x+175 pp.
\bibitem{vMS2010}
M.~van Meer and J.~Stokman:
Double affine Hecke algebras and bispectral quantum 
Knizhnik-Zamolodchikov equations. 
Int. Math. Res. Not. IMRN 2010, 969--1040. 
\bibitem{MN1997}
K.~Mimachi and M.~Noumi:
An integral representation of eigenfunctions for Macdonald's $q$-difference operators.
T\^{o}hoku Math. J.  {\bf 49} (1997), 517--525.  
\bibitem{NS2012}
M.~Noumi and J.~Shiraishi:
A direct approach to the bispectral problem for the Ruijsenaars-Macdonald 
$q$-difference operators.
(arXiv:1206.5364, 44 pages).
\bibitem{R1987}
S.N.M.~Ruijsenaars:
Complete integrability of relativistic Calogero-Moser systems and elliptic function identities. 
Comm. Math. Phys. {\bf 110} (1987), 191--213.
\bibitem{R1997}
S.N.M.~Ruijsenaars:
First order analytic difference equations and integrable quantum systems. 
J. Math. Phys. {\bf 38} (1997), 1069--1146.
\bibitem{R2009a}
S.N.M.~Ruijsenaars:
Hilbert-Schmidt operators vs. integrable systems of elliptic 
Calogero-Moser type, 
I. The eigenfunction identities, 
Comm. Math. Phys. {\bf 286} (2009), 629--657.
\bibitem{R2009b}
S.N.M.~Ruijsenaars:
Hilbert-Schmidt operators vs. integrable systems of elliptic 
Calogero-Moser type, 
II. The $A_{N-1}$ case: first steps. 
Comm. Math. Phys. {\bf 286} (2009), 659--680.
\bibitem{S2005}
J.~Shiraishi:
A conjecture about raising operators for Macdonald polynomials.
Lett. Math. Phys. {\bf 73} (2005), 71--81.
\bibitem{S2019}
J.~Shiraishi:
Affine screening operators, affine Laumon spaces, and 
conjectures concerning non-stationary Ruijsenaars functions.
J. Integrable Syst. {\bf 4} (2019), xyz010, 30pp. 
(arXiv:1903.07495, 26 pages).  
\end{thebibliography}
\end{document}